\documentclass[12pt]{article} 
\usepackage{amssymb,amsmath,xcolor,graphicx,xspace,colortbl, rotating, ragged2e}
\usepackage{amsfonts}
\usepackage{amsthm}
\usepackage{geometry}
\usepackage[onehalfspacing]{setspa ce}
\usepackage{subcaption}
\usepackage{tabulary}
\usepackage{textcomp}
\usepackage{hyperref}
\hypersetup{
	colorlinks = true,
	linkcolor= red, 
	citecolor= blue,
	urlcolor = red
}
\usepackage{bm}
\usepackage{tabulary}
\usepackage{epstopdf}
\usepackage{rotating}
\DeclareGraphicsExtensions{.pdf,.svg,.eps,.ps,.png,.jpg,.jpeg}
\newtheorem {example}{Example}[section]

\newtheorem{assumption}{Assumption}[section] 
\newtheorem{theo}{Theorem}[section]
\newtheorem{coro}{Corollary}[theo]
\newtheorem{lem}[theo]{Lemma}

\usepackage[comma,authoryear,round]{natbib}
\usepackage{longtable}
\usepackage{booktabs}
\usepackage{multirow}

\usepackage{siunitx} 

\usepackage[utf8x]{inputenc}
\usepackage{sectsty}
\sectionfont{\fontsize{13}{15}\selectfont}
\subsectionfont{\fontsize{12}{14}\selectfont}
\usepackage{xr}
\newcommand{\norm}[1]{\left\lVert#1\right\rVert}
\renewenvironment{abstract}
{\small
	\begin{center}
		\bfseries \abstractname\vspace{-.5em}\vspace{0pt}
	\end{center}
	\list{}{%
		\setlength{\leftmargin}{6mm}
		\setlength{\rightmargin}{\leftmargin}%
	}%
	\item\relax}
{\endlist}
\begin{document}

\title{\Large  Weak (Proxy) Factors Robust Hansen-Jagannathan Distance For Linear Asset Pricing Models}
\author{Lingwei Kong  
	\thanks{%
		University of Groningen; l.kong@rug.nl.
	} 
}


\date{Preliminary version. December 20, 2020}
 \linespread{1.3}

\maketitle

\thispagestyle{empty}
\begin{abstract} { \indent 
The Hansen-Jagannathan (HJ) distance statistic is one of the most dominant measures of model misspecification. However, the conventional HJ specification test procedure has poor finite sample performance, and we show that it can be size distorted even in large samples when (proxy) factors exhibit small correlations with asset returns. In other words, applied researchers are likely to falsely reject a model even when it is correctly specified. 
We provide two alternatives for the HJ statistic and two corresponding novel procedures for model specification tests, which are robust against the presence of weak (proxy) factors, and we also offer a novel robust risk premia estimator. Simulation exercises support our theory. Our empirical application documents the non-reliability of the traditional HJ test since it may produce counter-intuitive results, when comparing nested models, by rejecting a four-factor model but not the reduced three-factor model, while our proposed methods are practically more appealing and show support for a four-factor model for Fama French portfolios.  
}
\end{abstract}


 ~
 \bigskip

\doublespace 

\setcounter{page}{1}
 \pagenumbering{arabic}

\newpage
\section{Introduction}Linear factor models have gained tremendous popularity in the empirical asset pricing literature, see e.g.,  \cite{fama1993common},  \cite{lettau2001consumption}, \cite{kan2004hansen}, \cite{kan2008model}. The low dimensional factor structure in asset returns is well-documented  (e.g., \cite{kleibergen2015unexplained}), and \cite{harvey2016and} list hundreds of papers with factors that attempt to explain the cross-section of expected returns. Since so many factors are introduced, the proposed factors are at best proxies for some unobserved common factors, and the asset pricing models are at best approximations. Therefore, it is more appealing to determine whether or not the data reject a model, namely how good a model can approximate the data than to identify important factors  (or factors with significant risk premia). The assessment of model performance is where specification tests play a role. To evaluate these factors and diagnose the model specifications, the HJ distance,  proposed in \cite{hansen1997assessing}, has emerged as one of the most dominant measures of model misspecification in the empirical asset pricing literature (e.g., \cite{jagannathan1996conditional}, \cite{kan2004hansen}).

However, this paper shows that the conventional HJ statistic can be unreliable. Previous studies have shown that the lack of model identification can lead to spuriously significant risk premia (\cite{kleibergen2009tests}, \cite{2016},  \cite{anatolyev2018factor}), and a misleading gauge of model fit based on the second pass $R^2$ (\cite{kleibergen2015unexplained}). This paper demonstrates that when models are weakly identified the HJ statistic does not measure model fit satisfactorily and the specification test via the HJ statistic, which we refer to as the HJ specification test in the sequel, is not reliable. In an empirically relevant setting where proxy factors weakly correlate with the unobserved common factors, the HJ specification test can be size distorted even in large samples. The boundary, determined via the HJ specification test, between correct model specifications and misspecifications begins to blur in these so-called weak identification cases. Another potential issue would be the omitted-strong-factor problem. The resulting strong cross-sectional dependence in the error term can exaggerate all sorts of distortions when some included (proxy)\footnote{Sometimes researchers consider a latent factor structure in asset returns and regard included factors in empirical studies as proxies for priced latent common factors (e.g., \cite{kleibergen2015unexplained}, \cite{giglio2017inference}), sometimes factors are assumed to be directly observed and priced weak common factors lead to problems (e.g., \cite{kleibergen2009tests}, \cite{anatolyev2018factor}).  
	We mostly adopt the former idea, but our discussions and methods are also valid in the latter case, and we emphasize this by enclose the term proxy in brackets.} factors are weak (\cite{kleibergen2015unexplained}). One of the reasons for these failures is that sampling errors are no longer negligible asymptotically in the presence of weak (proxy) factors. Therefore, the conventional asymptotic justification may fail in empirically relevant settings, as weak (proxy) factors are commonly observed in many recent studies (e.g., \cite{kleibergen2009tests}, \cite{anatolyev2018factor}).

This paper not only shows the potential failure of the HJ test, but also aims to improve the performance of specification tests. This contributes to the literature on providing the identification robust statistical tools. Recent papers have developed different techniques to incorporate some of these aforementioned issues, most of which focus on the identification and inference of risk premia.  \cite{2016} provides an estimation approach using shrinkage-based dimension-reduction technique which excludes weak/useless (proxy) factors. \cite{anatolyev2018factor} propose an estimation procedure based on split-sample instrumental variables regression with proxies for the missing factor structure. \cite{giglio2017inference} propose a three-pass estimation procedure and bypass the omitted factors bias by projecting risk premia of observed factors on those of strong factors extracted via principal components analysis (PCA). Alongside with these estimation techniques there are identification robust test statistics
to correct for the overly optimistic statistical inference of the risk premia (e.g. \cite{kleibergen2009tests}, \cite{kleibergen2019Consumption}, \cite{kong2018}). As for the specification tests of asset pricing models, \cite{gospodinov2017spurious} discuss the potential power loss of the $\mathcal{J}$ specification test when spurious/useless factors, which are completely uncorrelated with asset returns, are present.

This paper focuses on the robust model specification tests, and provides two easy-to-implement specification test procedures to remedy the size distortion of the HJ test resulting from weak (proxy) factors that are minorly correlated with asset returns. The \textit{first proposed test procedure} is a two-step Bonferroni-type method, and it is robust against identification issues when the number of asset returns is limited. This method takes into account the identification strength via a first-step confidence set, and we verify that it improves power compared with the $\mathcal{J}$ test. The \textit{second approach} relies on a novel four-pass estimator, and the test procedure provides valid inference results in an asymptotic framework where the number of assets is comparable to the number of the observation periods. Our proposed four-pass estimator directly leads to a novel risk premia estimator, and thus we also contribute to the literature on estimation of risk premia in the presence of weak (proxy) factors and omitted factors. For linear asset pricing models, the conventional approach for estimating risk premia is known as the Fama-Macbeth (FM) two-pass estimation procedure (\cite{fama1973risk}), where risk premia estimates result from regressing average asset returns on first-pass estimated risk exposures (factor loadings $\beta$'s). The two-pass procedure is easy to implement but can result in unreliable estimates and inference when some included factors are not strongly correlated with asset returns such that their risk exposures do not dominate corresponding sampling errors (\cite{kleibergen2009tests}, \cite{anatolyev2018factor}), which resembles the failure of the 2SLS estimator in instrumental variable regression when instruments are weak. Besides,  \cite{anatolyev2018factor} show that the missing factor structure exacerbates the weak (proxy) factor problem. We show our risk premia estimator is robust to the presence of both weak (proxy) factors and missing factors. 

Our empirical application documents the strange behavior of the HJ test. Counter-intuitively, it can reject a four-factor model but not the corresponding three-factor model nested within the four-factor model.  We attribute this behavior to the additional fourth factor being a weak proxy factor which leads to a undesirably high rejection rate of the HJ test. Our proposed procedures do not have this problem and reflect the factor structure in asset returns in a more informative way.  \\

%
%
%

%
%

%

The paper is organized as follows: Section \ref{sec: HJ} reviews the basic model setting and shows the drawbacks of the HJ statistic;  Section \ref{sec HJS} and \ref{sec:HJN} introduce our proposed model specification test procedures, where Section \ref{sec HJS} discusses our two-step Bonferroni-like method and Section \ref{sec:HJN} considers an approach that is valid with a double-asymptotic framework; Section \ref{sec:emp} presents results of our empirical application.

This paper uses the following notation: $P_X$ stands for $X(X'X)^{-1}X'$ for a full column rank matrix $X$, $M_X$ for $I-P_X$, $X^{\frac{1}{2}}$ for the upper triangular matrix from the Cholesky decomposition of the positive definite matrix $X$  such that $X=(X^{\frac{1}{2}})'X^{\frac{1}{2}}$. Besides, in the following discussion, the notation would be more precise with $N, T$ in the sub- or superscripts. For example to model the weak (proxy) factors, it might be reasonable to use notation such as $\beta_{T,N}, d_{g,T,N}, \gamma_{N,T}, \theta_{g,N,T}$  since the parameter values may change according to the sample dimensions in order to model the local to zero behavior. To avoid a more cumbersome notation, we ignore these subscripts when there can be no misunderstanding.

\section{Models and Problems} \label{sec: HJ}
This section introduces linear asset pricing models and  the conventional model selection and specification test procedure based on the HJ distance. We use the term HJ statistic to denote the conventional squared HJ distance estimator, and to distinguish it from the other two estimators, our so-called HJN and HJS statistics, proposed in Sections 3 and 4. We start by introducing our baseline model setting. We next derive asymptotic properties of the HJ statistic in the presence of weak (proxy) factors to clarify the problems we focus on.      
\subsection{Baseline model setting} 
We work with the linear asset pricing model because of its popularity in empirical studies. It imposes that all asset returns share common risk factors described by a small set of proposed factors. We regard the proposed factors in empirical studies as proxies for latent ones in the form as suggested in (\cite{kleibergen2015unexplained}). Assumption \ref{ass} summarizes the baseline model setting:
\begin{assumption} \label{assum: factor structure in r}  \label{ass} \label{ass: delta}
	For the $N\times 1$ vector of asset gross returns $r_t$, we assume that 
	\begin{align}
		r_t  = c + \beta f_t + u_t, \label{eq1}
	\end{align}
	with $f_t$ a $K\times 1 $ vector of (possibly) unobserved zero-mean factors, $u_t$ an $N\times 1$ vector of idiosyncratic components, and with a $K\times 1$  vector of proxy factors $g_t$ 
	\begin{align}
		f_t = &  d_g\left( g_t - \mu_g \right) + v_t,   
	\end{align}   
	where $\mu_g=\mathbb{E}g_t$, $g_t$ is uncorrelated with $u_t,v_t$, $d_g$ is of full rank and $g_t,v_t,u_t$ are stationary with finite fourth moments. Furthermore, 
	\begin{align}
		c= \iota_N \lambda_0  + \beta \lambda_f,\label{eq:res on c}
	\end{align}
	where $\lambda_0\neq 0 $ is the zero-beta return, $\lambda_f$ is a $K\times 1$ vector of risk premia, and the parameter space of $\lambda_0,\lambda_f$ is compact.  
\end{assumption} 
Assumption \ref{ass} describes the beta representation of linear asset pricing models, and the DGP is similar to the one employed in \cite{kleibergen2015unexplained}. The moment conditions (or in other words the structural assumptions imposed on the constant term), $c= \iota_N \lambda_0  + \beta \lambda_f$, are commonly used in linear asset pricing model (e.g. \cite{cochrane2009asset}). If $f_t$ is observed, then we would have perfect proxies with $d_g=I_K$, $v_t=0$. Therefore, this model setting also embeds the model specification used in e.g. \cite{kleibergen2009tests}, \cite{anatolyev2018factor} where factors $f_t$ are assumed to be observed. Using the observed factors $g_t$, model (\ref{eq1}) can be rewritten as 
\begin{align}
	r_t = c+ \beta_g \bar{g}_t +  \widetilde{u}_{g,t} , \label{1}
\end{align}
with $\beta_g = \beta d_g, \bar{g}_t= g_t-\bar{g}, \bar{g}= \sum_{t=1}^T g_t /T, \widetilde{u}_{g,t} = {u}_{g,t}+ \beta_g \left( \bar{g}- \mu_g \right) u_{g,t} = \beta v_t + u_t$. The estimation of the risk premia $\lambda_g$ is usually accomplished by the FM two-pass estimator (\cite{fama1973risk}, \cite{shanken1992ebp}). In the first pass, the risk exposures ${\beta}_g$ are estimated by regressing asset returns $r_t$ on a constant and factors $\bar{g}_t$, and in the second pass the FM estimator $\widehat{\lambda}_g$ results from regressing average asset returns $\bar{r}= \sum_{t=1}^T r_t/T$ on an $N\times 1$ unity vector $\iota_N$ and the risk exposure estimates $\widehat{\beta}_g$.

Another well-known representation of asset pricing models is the stochastic discount factor (SDF) representation, based on which the HJ distance is defined. \cite{cochrane2009asset} shows that for linear asset pricing models, there is a corresponding SDF $m_t$ that is linearly spanned by the latent factors
\begin{align}
	m_t(\theta)  =  F_t'\theta,   \label{eq:linear SDF}
\end{align}
with $F_t=(1,f_t')'$, and re-scaled risk premia $\theta= (1/\lambda_0, -V_{ff}^{-1}\lambda_f/ \lambda_0)$.  The moment conditions (\ref{eq:res on c}) are then equivalent to the following ones 
\begin{align}
	\mathbb{E}\left(m_{t}(\theta_{0}) r_{t}\right) = \iota_N, 	\label{eq:SDF_r}
\end{align}
with $\iota_N$ an $N\times 1 $ unity vector with all entries equal to one. The population pricing errors which are the deviations from the moment conditions (\ref{eq:SDF_r}) are denoted by 
\begin{align}
	e(\theta) =\iota_N - \mathbb{E}\left(m_{t}(\theta_{}) r_{t}\right) , \label{eq:pricing error}
\end{align}
With a linear SDF (\ref{eq:linear SDF}) $e(\theta)=\iota_N -q\theta$, for $q=\mathbb{E}\left(r_tF_t' \right) $ a full column rank $N\times K$ matrix.

\cite{hansen1997assessing} (HJ) propose the minimum distance between the SDF of an asset pricing model and a set of correct SDFs as a measure of model misspecification. It also serves as a measure of goodness-of-fit. A smaller value of the HJ distance indicates a better model fit, and this is used for model selection. The population squared HJ distance $\delta$  has an explicit expression:  
\begin{align}
	\delta^2 = \inf_\theta e(\theta)' Q_r^{-1} e(\theta),  
\end{align}
with $Q_r=  \mathbb{E}\left(r_tr_t' \right)$ a full column rank $N\times N$ matrix. With a linear SDF, after some simple algebra, we can write the squared HJ distance explicitly as $\iota_N' \left(Q_r^{-1}-Q_r^{-1} q \left( q'Q_r^{-1} q \right)^{-1} q'Q_r^{-1}  \right) \iota_N,$ which is also numerically equal to  $\iota_N' (Q_r^{-1}-Q_r^{-1} B (B'Q_r^{-1}B)^{-1} B'Q_r^{-1}   ) \iota_N$ with  $B=(c,\beta)$, and it is zero if and only if moment conditions (\ref{eq:SDF_r}) hold. Given the observed proxy factors $g_t$, the sample counterpart of the squared HJ distance, the HJ statistic, is  
\begin{align}
	\widehat{\delta}^2_g = &\inf_{\theta_{G}}  e_{g,T}(\theta_G)'\widehat{Q}_{r}^{-1} e_{g,T}(\theta_G) 
	. 
\end{align}

In a linear asset pricing model, the sample pricing errors $ e_{g,T}(\theta_G) =\sum_{t=1}^{T} e_{g,t}(\theta_G)/T =\iota_N - q_{G,T}\theta_G,~     e_{g,t}(\theta_G)= \iota_N -r_tG_t'\theta_G,~    q_{G,T}= \sum_{t=1}^{T} r_t G_t'/T,~       \widehat{Q}_{r} = \sum_{t=1}^{T} r_tr_t'/T$,   $G_t=(1,\bar{g}_t')'$, and the estimator resulting from this quadratic optimization problem is 
\begin{align}
	\widehat{\theta}_G = \left(q_{G,T}'\widehat{Q}_r^{-1} q_{G,T} \right)^{-1}q_{G,T}'\widehat{Q}_r^{-1}\iota_N. 	
\end{align}

\cite{jagannathan1996conditional} propose the HJ specification test by testing the moment conditions (\ref{eq:SDF_r}) via the HJ statistic. Under the null hypothesis that the moment conditions (\ref{eq:SDF_r}) hold, \cite{jagannathan1996conditional} show that the asymptotic distribution of the HJ statistic follows a weighted sum of $\chi^2(1)$ random variables, which is because of the weighting matrix used in the HJ statistic. If we weight by the long-run covariance matrix of the sample pricing errors, we would have a regular chi-square-type limiting distribution. However, since we weight the HJ statistic differently with the second moment of asset returns as the weighting matrix, each of these $\chi^2(1)$ random variables has a weight different from one. Therefore,  
the critical values for the HJ statistic are obtained from the weighted sum of $\chi^2(1)$ random variables 
\begin{align}
	\sum_{i=1}^{N-K-1}  {p}_i x_i,
	\label{eq:critical values}
\end{align}with $x_i$ being independent $\chi^2(1)$ distributed random variables,  and $p_i$ being the positive eigenvalues of the matrix  
\begin{align*}
	\widehat{S}^{\frac{1}{2}} \left(\widehat{Q}_r^{-1}-\widehat{Q}_r^{-1}q_{G,T}  \left(q_{G,T} '\widehat{Q}_r^{-1} q_{G,T}  \right)^{-1}q_{G,T} '\widehat{Q}_r^{-1}  \right)  S_T^{\frac{1}{2}'}, 
\end{align*}  
with $\widehat{S} = S_T(\widehat{\theta}_G)$ and $S_T({\theta}_G)$ a consistent estimator 
of the long-run variance
matrix of the sample pricing errors $e_{g,T}(\theta_G)$ (for example, one may simple choose $S_T({\theta}_G) = \frac{1}{T}\sum_{t=1}^T e_{g,t} ({\theta}_G)e_{g,t} ({\theta}_G)'$, $e_{g,t}({\theta}_G)= \iota_N -r_tG_t'{\theta}_G$ provided that Assumptions \ref{assum: factor structure in r}-\ref{assum:factor loading strength} hold).

\subsection{Problems and asymptotic properties}
Our interest lies in the performance of the HJ statistic in the presence of weak identification issues, in particular, when observed proxies $g_t$ are only weakly correlated with asset returns. \cite{ahn2004small} document the poor finite sample performance of the HJ specification test, and they argue that the size distortion is due to the critical value of the test which requires the estimation of the covariance matrix of $e_t(\theta)$ that performs badly with a limited number of observation periods. This is also consistent with the findings in \cite{kleibergen2019Consumption}, \cite{kong2018}. In later parts, we show that not only in finite samples but also in large samples, the HJ specification test can be severely size distorted in the presence of weak identification issues. We focus on \textit{two issues} that can cause the potential deficiency of the HJ statistic.

\textit{1. Weak (proxy) factors}. The HJ statistic depends on the estimator $\widehat{\theta}_G$. This is a GMM estimator based on the moment conditions (\ref{eq:SDF_r}) with weighting matrix $\widehat{Q}_r$. Similar to the FM risk premia estimator, this estimator can be constructed in two steps. In the first step the regressor in the SDF, $q_g$, is estimated via the estimator $q_{G,T}$, and in the second step $\widehat{\theta}_G$ results from regressing $\widehat{Q}_r^{-\frac{1}{2}}\iota_N $ on the first stage estimates $\widehat{Q}_r^{-\frac{1}{2}}q_{G,T}$. This close link with the FM estimator raises concern for the quality of the $\theta_G$ estimator.  

The FM estimator is unreliable under weak identification (e.g., \cite{kan1999two},  \cite{kleibergen2009tests},  \cite{kleibergen2015unexplained}, \cite{kleibergen2019Consumption}, \cite{kong2018}, \cite{anatolyev2018factor}). For linear asset pricing models, the identification strength is reflected by the rank of $B_g=(c,\beta_g)$ (e.g. \cite{kleibergen2019Consumption}). The weak identification issues result from the empirical observation that the matrix $B_g$ might be of reduced rank or near reduced rank. For example, this can happen when some (proxy) factors used in the estimation are weakly correlated with the asset returns. One way of modeling these weak (proxy) factors is to consider a sequence of models (or a sequence of parameter values) such that along the sequence, factor loadings are smaller and thus less informative for identifying risk premia. For example, 
suppose the $\beta_g$ matrix is small, modeled by a drifting to zero sequence of order $O(1/\sqrt{T})$, then the sampling errors in the first stage estimator $\widehat{\beta}_g$, which are of the same order, are no longer negligible. These non-negligible sampling errors lead to the asymptotic invalidity of the FM estimator under weak (proxy) factors.

Following the same reasoning, the asymptotic justification for the estimator $\widehat{\theta}_G$ fails when $q_g$ is small and comparable to its sampling error (see Theorem \ref{theo:asymptotic for theta}).  Since $q_g=\beta_g V_g$ with $V_g =\mathbb{E}g_tg_t'$,  we model weak (proxy) factors using drifting to zero risk exposures (see Assumption \ref{assum:factor loading strength}) to mimic the behavior of a small $q_g$, which is in line with the literature on weak factors (e.g. \cite{kleibergen2009tests}).

\textit{2. The missing factor structure}. Omitted factors have received attention in recent studies (e.g., \cite{kleibergen2015unexplained},  \cite{giglio2017inference},  \cite{anatolyev2018factor}).   When we work with observed factors $g_t$ in a latent factor setting, equation (\ref{1}) suggests that the omitted factors $v_t$ contribute to the error term  $u_{g,t}$, and we also allow that unobserved factors explain most of the cross-sectional dependence in $u_{g,t}$. Similar to the discussion of  the FM two-pass risk premia estimator in \cite{anatolyev2018factor} ,
the missing factor structure could exacerbate the problem caused by the weak (proxy) factors and enlarge the bias in the estimator $\widehat{\theta}_g$ (see Theorem \ref{theo:asymptotic for theta}), as the presence of an unobserved (missing) factor structure in the error terms creates the classical omitted-variables problem in the second step regression of the $\widehat{\theta}_g$ estimator when some (proxy) factors are weak.

Therefore, the HJ statistic may use an estimate that is potentially far away from the true value, and thus selection and inference based on the HJ statistic can be misleading. Before we continue to verify this, we make two assumptions.

\begin{assumption} \label{assumption: weak cross-sectional correlation of the idiosyncratic component } \label{assum: forth moments of ee}
	Suppose $u_t$ can be decomposed into two parts: a missing  factor structure with a $K_z\times 1$ ($K_z\geq 0$) vector of unobserved strong factors $z_t$  and  weakly cross-sectional correlated  noise $e_t$;  
	\begin{align}
		u_t =  \gamma z_t +  e_t, 
	\end{align}  
	where (i) $e_t$ (with mean zero and bounded fourth moment $\sup_{i}\mathbb{E} e_{it}^4 < L< \infty$) is independent from $e_s, s\neq t$ and $g_{t'}, \forall t'$; (ii) denote $\Omega_e = \mathbb{E}e_te_t'  $, then $\lim_{N,T} \text{tr}\left(  \Omega_e \right)/N =a>0$  and $ 0<l<\lim\inf_{N,T}  \lambda_{\min} \left(  \Omega_e \right) <\lim\sup_{N,T} \lambda_{\max}\left(  \Omega_e \right)  < L < \infty$ with $\lambda_{\min}(X)$ the smallest eigenvalues of matrix $X$ and $\lambda_{\max}(X)$ the largest eigenvalues of matrix $X$; (iii) $\mathbb{E}\left| \frac{1}{\sqrt{N}} \sum_{i=1}^{N}  \left(e_{it}^2-  \mathbb{E}e_{it}^2\right) \right|^4 < L < \infty$.
\end{assumption} 

\begin{assumption} \label{assum:factor loading strength} \label{assum:factor loading strength 2} Denote $Q_r$ by the second moments of $r_t$, $\eta =\left(\eta_{B_g},  \gamma \right) $ with $\eta_{B_g}=\left(c,  \eta_{\beta_g} \right), \eta_{\beta_g}=\left(\beta_{g,1} , \sqrt{T}\beta_{g,2} \right), \beta_{g,i} = \beta d_{g,i}, i=1,2$ being of dimension $N\times K_{g,i}, i=1,2$  ($K_{g,1}+K_{g,2}=K, K_{g,2}\geq 0$) and full column rank matrices.  
	(i) for fixed $N$,  we assume $ \eta     '  \eta   $ is a $(1+K+K_z )\times (1+K+K_z )$  positive definite matrix; 
	(ii) as $N,T$ approach to infinity, $ N^{-1}\eta'\eta   $ converges to a $(1+K+K_z )\times (1+K+K_z )$  positive definite matrix.  
\end{assumption} 

Our framework involves the observed factors $g_t$, and the omitted ones $v_t, z_t$. They have factor loadings $\beta_g, \beta, \gamma$ respectively. Assumption \ref{assum:factor loading strength} specifies the strengths of these factors. The loadings, $\beta_{g,2}$,  of the $K_{g,2}$ proxy factors $g_{2,t}$ are modeled as drifting to zero sequences, so we call $g_{2,t}$ weak proxy factors. We do not restrict the strength of the priced latent factors $f_t$, and allow for weak priced latent factors.   This  assumption resembles the factor loading assumption in \cite{anatolyev2018factor}, but our risk exposure matrix $(\beta_g, \beta, \gamma)$ is of reduced rank.

Assumption \ref{assumption: weak cross-sectional correlation of the idiosyncratic component }, which is similar to assumptions in \cite{onatski2012asymptotics} and \cite{anatolyev2018factor}, does not fully rule out the cross-sectional dependence in the idiosyncratic error term $e_t$. This assumption allows the explanatory power of the cross-sectional variation in $e_t$ to be comparable to the weak proxy factors when N,T increase proportionally, and the weak identification issue appears when the explanatory power of the proxy factors are roughly of the same order as $e_t$.    
When the cross sectional size $N$ is fixed, Assumptions  \ref{assumption: weak cross-sectional correlation of the idiosyncratic component }.(ii)-(iii) imposed on the noise term $e_t$ hold naturally as long as Assumption \ref{assumption: weak cross-sectional correlation of the idiosyncratic component }.(i) holds and we can not really distinguish weak and strong factors.  In later parts of this paper when $N$ is fixed, we do not make use of the assumptions imposed on $e_t$ but only the independence assumption (Assumption \ref{assumption: weak cross-sectional correlation of the idiosyncratic component }.(i)). We assume that $e_t$ is independent across periods which is consistent with the efficient market hypothesis, and since empirical studies mostly use monthly or even less frequent data, this is not a unrealistic assumption.


%
%
%
%

\begin{lem} \label{lem:B_g}
	Suppose Assumption \ref{ass}, \ref{assum:CLT for the proxy factors} hold, let T increase to infinity then 
	\begin{align*}
		\widehat{B}_g =_d \left(c, \beta_g\right) + \psi_{B_g}/\sqrt{T} 
	\end{align*}
	where $\widehat{B}_g = q_{G,T}\widehat{Q}_G^{-1},  \widehat{Q}_G=\sum_{t=1}^TG_tG_t'/T, \psi_{B_g}=\left(\psi_c,  \psi_{\beta_{g,1}},  \psi_{\beta_{g,2}}\right) $ and $\text{vec}(\psi_{B_g})$ being zero-mean normal random vectors. 	
	
	\noindent
	Proof: See Appendix  \ref{Proof of Lemma  {lem:B_g dist}}.
\end{lem}

\begin{theo}\label{theo: charaterization of the delta}
	Suppose Assumptions \ref{assum:CLT for the proxy factors} and \ref{assum: factor structure in r}-\ref{assum:factor loading strength} hold, let T increase to infinity with fixed N, then the behavior of the HJ statistic is characterized by: 
	\begin{align*}
		T	\widehat{\delta}^2_g \rightarrow_d \widetilde{\psi}_{B_g} ' M_{Q_r^{-\frac{1}{2}}  \left( \eta _{B_g} + (0\vdots \psi_{\beta_{g,2}} )  \right) } \widetilde{\psi}_{B_g}  
	\end{align*} 
	with $\widetilde{\psi}_{B_g}= \psi_{B_g} Q_G \theta_G $.\\
	
	\noindent
	Proof: See Appendix \ref{Proof of {theo:fixed N asymptotic for theta}}. 
\end{theo}

Theorem \ref{theo: charaterization of the delta} is derived assuming that the linear model is correctly specified, and it suggests that with strong factors, $K_{g,2}=0$, the weighted sum of $\chi^2$'s provides a reasonable approximation (Corollary \ref{theo: charaterization of the delta 2}).

\begin{coro}\label{theo: charaterization of the delta 2}
	Suppose Assumption \ref{assum:2.2.2} and the assumptions in  Theorem \ref{theo: charaterization of the delta} hold with $K_{g,2}=0$, then  	
	\begin{align*}
		T	\widehat{\delta}^2_g \sim_d \sum_{i=1}^{N-K-1}  {p}_i x_i   
	\end{align*}  
	with $x_i$ being independently $\chi^2(1)$ distributed random variables and $p_i$ being the positive eigenvalues of the matrix $\widehat{S}^{\frac{1}{2}} \left(\widehat{Q}_r^{-1}-\widehat{Q}_r^{-1}q_{G,T}  \left(q_{G,T} '\widehat{Q}_r^{-1} q_{G,T}  \right)^{-1}q_{G,T} '\widehat{Q}_r^{-1}  \right)  \widehat{S}^{\frac{1}{2}'}$. \\
	
	\noindent
	Proof: This is a direct result of Theorem \ref{theo: charaterization of the delta}. 
\end{coro}

However, the conventional specification test procedure can be unreliable and suffer from severe size distortion even in large samples due to the irregular distribution of the HJ statistic (Corollary \ref{theo: strong beta asymptotics}).  
\begin{coro} \label{theo: strong beta asymptotics}
	Suppose the assumptions in Corollary \ref{theo: charaterization of the delta 2} hold with $0< K_{g,2} \leq K$,      
	\begin{align}
		\lim\sup_{N} \lim\sup_{T}~~		 \mathbb{P} \left(  T\widehat{\delta}_g \geq \widehat{c}_{1-\alpha}  \right) =1 \nonumber
	\end{align}
	where $\widehat{c}_{1-\alpha} $ is the conventional critical value derived from the distribution (\ref{eq:critical values}).\\ 
	
	\noindent
	Proof: See Appendix \ref{Proof of Corollary {theo: strong beta asymptotics}}. 
\end{coro}


Corollary \ref{theo: strong beta asymptotics}  shows  that under certain conditions the conventional specification test rejects the model specification with probability converging to one  even when the moment conditions hold. Thus the conventional specification testing procedure based on the HJ statistic may mistake the "weak identification" resulting from the weak (proxy) factors for model misspecification, and  leads to over-rejection when models are correctly specified.

%


\subsection{Simulation exercises}
We conduct simulation exercises to show that the HJ distance statistic as a model selection criterion might favor the presence of useless factors, and the HJ specification test suffers from severe size distortions.   

In the first simulation exercise (Figures \ref{fig:delota density}, \ref{fig:scaled_diss}), we calibrate the data generating process to match the data set of monthly gross asset returns on 25 size and book to market sorted portfolios from 1963 to 1998 and the three Fama French (FF) factors used by \cite{lettau2001consumption}. The data is simulated in the following way: we simulate three proxy factors $g_t \sim ~i.i.d~ N(0,V_F)$, three omitted factors $v_t~ i.i.d~\sim N(0.99 V_F)$ and three strong factors are then generated by $f_t=0.1*g_t + v_t,$. $V_F$ is calibrated to the sample covariance of the FF factors. We also generate three completely useless factors $w_t\sim i.i.d. N(0,V_F)$.  We then generate returns via $r_t=\iota_N + \beta \lambda+\beta f_t +u_t$, $u_t\sim ~i.i.d ~N(0,V_u)$, 
where we set $\lambda$ to be the sample risk premia estimated via the FM two-pass estimator, $\beta$ is the sample slope parameter between the assets returns and FF factors,
and  $V_u$ is the sample covariance of the residuals resulting from regressing asset returns on a constant and FF factors from the data.

Figure \ref{fig:delota density} compares the density functions of the simulated HJ statistics evaluated with various combinations of the factors $g_t, f_t, w_t$. For example, the black solid curve is drawn using three strong factors $f_t$, the black dashed curve is drawn with two strong factors $f_{1t}, f_{2t}$. Ideally, the black solid curve should be the most left, since the model with three strong factors should be most likely to be selected by the HJ statistic. However,  comparing the red solid and black solid curves shows that adding additional useless factors leads to a shift of the distribution to the left, so it reduces the HJ statistic and leads to a "preferred model".  The blue sold curve illustrates the density function of the HJ statistic of the model with three weak proxy factors. By construction, the moment conditions (\ref{eq:SDF_r}) are satisfied by the three weak proxy factors, and this model is correctly specified. If we compare the blue solid curve with the blacked dashed one which is constructed with only two strong factors, then the misspecified model with two strong factors is more likely to be selected.  These observations imply that the HJ statistic is not a satisfying model selection tool.   
\begin{figure}[h!] 
	\centering	\includegraphics[width=0.8\textwidth,height=0.3\textheight]{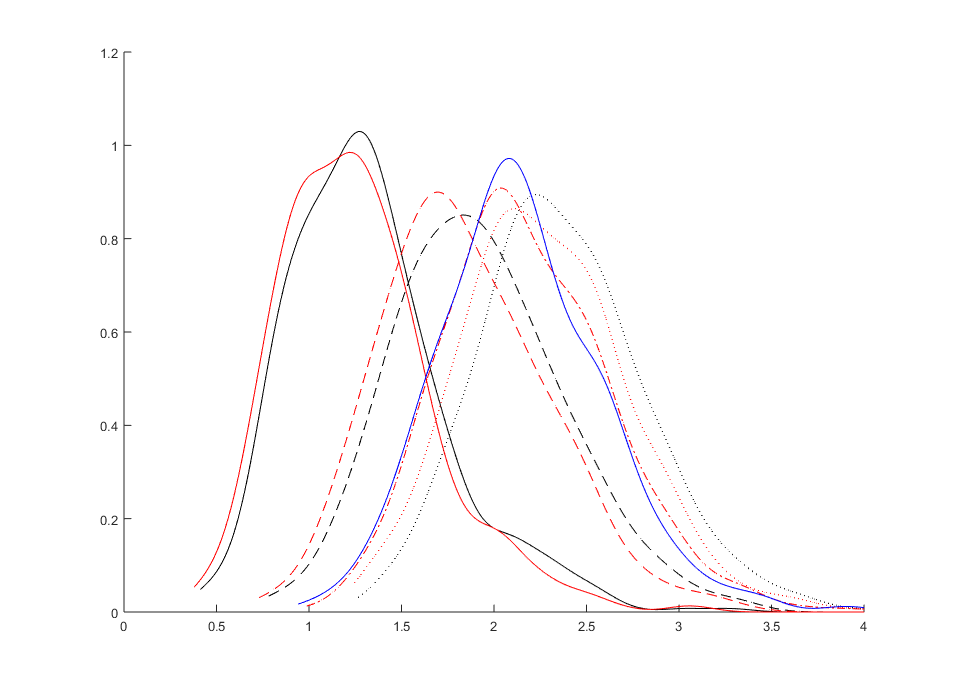} 
	\caption{Density functions of the HJ statistic: (1) black solid: three strong factors; (2) black dashed: two strong factors; (3) black dotted: one strong factors; (4) blue solid: three weak proxy factors; (5) red solid: three strong factors and one useless factor; (6) red dashed: two strong factors and one useless factor; (7) red dotted: one strong factors  and one useless factor; (8) red dash-dotted:one strong factors  and two useless factors }
	\label{fig:delota density}
\end{figure}

The observations in Figure \ref{fig:delota density} show that values of the HJ statistic can not properly distinguish between weakly identified models  and misspecified models. This is what motivates us to look further into the HJ specification test. Figure \ref{fig:scaled_diss} compares two different approaches for approximating the distribution of the HJ statistic: one uses the conventional weighted sum 
of $\chi^2$'s, from which the critical values of the HJ statistic result, and another one uses the infeasible distribution from Theorem \ref{theo: charaterization of the delta}. The left-hand side panels of Figure \ref{fig:scaled_diss} use three strong factors, while the right-hand side panels use three week ones.

The upper panels of Figure \ref{fig:scaled_diss} show that both approximations for the distribution of the HJ statistic (the conventional weighted sum of $\chi^2(1)$s and the infeasible one from Theorem \ref{theo: charaterization of the delta}) are bad when T is small, and shift to the left compared with the density function of the HJ statistic. This observation is consistent with the one in \cite{ahn2004small} that the HJ specification test over-rejects correct model specifications in small samples. 
With a limited number of observation periods, not only sampling errors in the $q_G$ estimators need to be taken into account but also those of other estimators such as the covariance estimator $\widehat{S}$ (e.g. \cite{kleibergen2019Consumption},  \cite{kong2018}). The infeasible distribution improves slightly by taking into account the sampling errors in the $q_G$ estimators. When T is large, the randomness in the covariance estimators becomes small, but the sampling errors in $q_{G,T}$ still matter when proxy factors are weak. 
As shown in the lower panels of Figure \ref{fig:scaled_diss}, with a larger sample size the conventional approximation works fine when factors are strong but not when weak proxy factors are present. With weak proxy factors, the distribution of the HJ statistic is not properly approximated by the weighted sum of $\chi^2$'s even in large samples, and the HJ specification test is still likely to over-reject models when moment conditions do hold (Corollary \ref{theo: strong beta asymptotics}).

\begin{figure}[h!] 
	\begin{center}
		\includegraphics[scale=0.4]{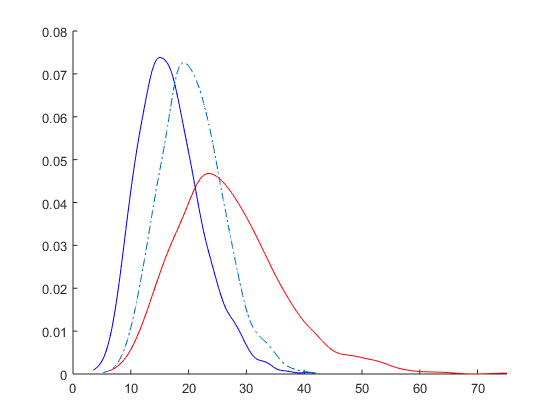}
		\includegraphics[scale=0.4]{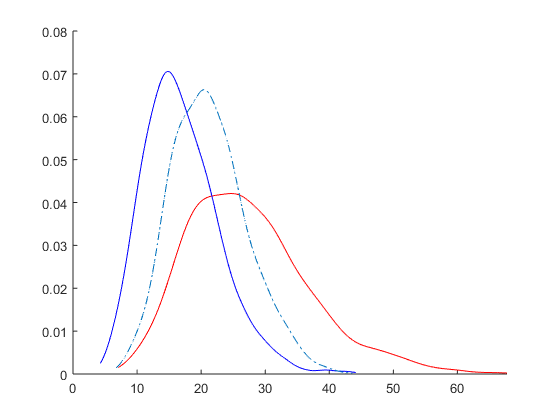}	\includegraphics[scale=0.4]{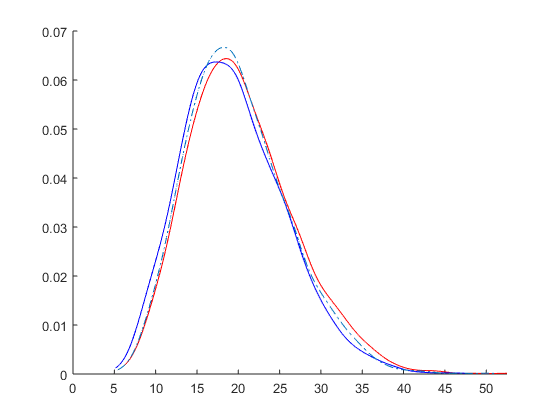}
		\includegraphics[scale=0.4]{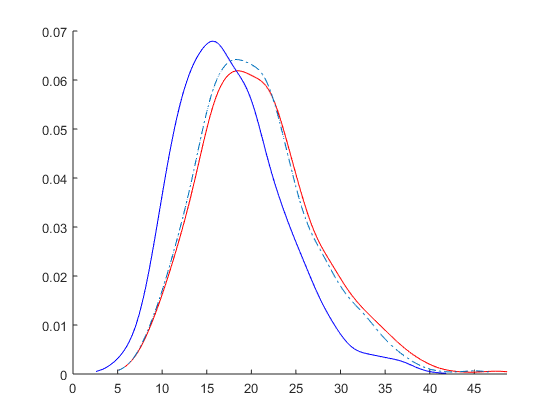}
		\caption{Density functions of the HJ statistic (red solid curve), weighted sum of $\chi^2$'s approximation (blue solid curves), asymptotic distribution based on Theory \ref{theo: charaterization of the delta} (dot-dashed curve).  
			Top  left-hand side panel:  $T=100,$ three strong factors $f_t$;  top right-hand side panel:  $T=100,$ three weak factors $g_t$; bottom left-hand side panel:  $T=1000,$ three strong factors $f_t$; bottom  right-hand side panel:  $T=1000,$ three weak factors $g_t$. } \label{fig:scaled_diss}
	\end{center}
\end{figure} 

Our second simulation exercise considers a simple single factor model in order to further illustrate the size distortion of the HJ specification test. We calibrated parameters to the data set from \cite{kroencke2017asset}. 
We simulate the proxy factor $g_t\sim ~i.i.d.~N(0,V_f/4)$, the omitted factor $v_t \sim ~i.i.d.~ N(0, Vff- d_g*(V_f/4)*d_g')$, the latent factor $f_t =d_g g_t + v_t$ and thus the variance of the $f_t$ remain unchanged to different values of the $d_g$. We calibrate $V_f$ to match the sample variance of the consumption growth factor from the data. The factor proposed in \cite{kroencke2017asset} has been shown to be weak (e.g. \cite{kleibergen2019Consumption}). Therefore,
we choose $\beta$ to be of $10\widehat{\beta}$ with $\widehat{\beta}$ the sample regression parameter from \cite{kroencke2017asset}, and thus $r_t$ is generated with one single strong factor $f_t$  via $r_t=\iota_N + \beta \lambda +  \beta f_t +u_t$, $u_t\sim ~i.i.d ~N(0,V_u)$, 
where we match $\lambda$ to the estimated risk premium from \cite{kroencke2017asset}.  We arbitrarily choose $d_g=1.9;0.9$ to mimic a strong and weak proxy factor.

Table \ref{table:1} shows that the HJ specification tests have poor finite sample performances, size distortions increase with the number of assets and with relative weak proxy factors the distortion is more severe, and these observations support Corollary \ref{theo: strong beta asymptotics}.     

\begin{table}[h!] \centering
	\begin{small}
		\begin{tabular}{@{}lrrrrrrrrrrrr@{}}\toprule
			\textbf{$\beta$} & \textbf{$N=5$} & \textbf{$N=10$}&\textbf{$N=15$}& \textbf{$N=31$}&   \\ \midrule
			\textbf{T=100}, $d_g=1.9$ &  0.5032 &  0.8824&  0.9711&  0.9992\\  
			\textbf{T=10000}, $d_g=1.9$ & 0.1210   & 0.1486 &  0.2298 &  0.2238  \\  
			\midrule 
			\textbf{T=100}, $d_g=0.9$ &  0.7132 &  0.9330& 0.9814&  1 \\  
			\textbf{T=10000}, $d_g=0.9$ &0.5174  & 0.8906 & 0.8834 &0.9978  \\   
			\bottomrule
		\end{tabular} 
	\end{small}
	\caption{Rejection frequency table under the null  (at 0.05 significance level) via the HJ specification test with critical values drawn from weighted sum of $\chi^2$'s}\label{table:1}
\end{table}

\section{Specification test with limited N: HJS} \label{sec HJS}
As in previous discussions, the HJ specification test can not provide valid inference when weak (proxy) factors are present, and this is because the HJ specification test procedure ignores some non-negligible sampling errors in the estimates of parameters that can not be properly identified in the presence of weak identification issues. In this section, we  suggest a  numerically simple and identification robust test procedure which replaces the estimates of these parameters with potential identification issues by those lying in a robust confidence set.  This approach is related to the widely studied weak instrument problem, where confidence sets with asymptotically correct coverage can be constructed for parameters with potential identification issues (e.g. \cite{kleibergen2005testing}, \cite{mikusheva2010robust}). 

\subsection{HJS specification test}
Our proposed HJS specification test procedure is conducted in three steps:  

\noindent
\textbf{Step (1)}:  Construct an identification robust confidence set, $CS_{r, \alpha_1}  $, for $\theta_G$ by inverting an Anderson-Rubin (AR) type test statistic (e.g., \cite{kleibergen2009tests}, \cite{gospodinov2017spurious}):   
\begin{align}
	CS_{r, \alpha_1} &= \left\{\theta\in \Theta: AR(\theta) \leq  c_{1-\alpha_1}  \right\};   ~~ 
	AR(\theta) = T e_T(\theta)' S_T^{-1}( {\theta})   e_T(\theta). 
\end{align}   
with $c_{1-\alpha_1}$ the $100(1-\alpha_1)\%$ percentile of the $\chi^2(N)$ distribution. 

\noindent
\textbf{Step (2)}: Compute the HJS statistic: 
\begin{align}
	\widehat{\delta}^*_{g}=  \inf_{\theta \in 	CS_{r, \alpha_1}     } \delta_{g,T}(\theta), 
\end{align}
with $\delta_{g,T} = (\theta) e_{g,T}(\theta_G)'\widehat{Q}_{r}^{-1} e_{g,T}(\theta_G) $. 
To complete the construction of the HJS statistic we set $\widehat{\delta}_g^* = \infty$ when the confidence set $	CS_{r, \alpha_1}  $ is empty.   

\noindent
\textbf{Step (3)}: This test would then reject the null hypothesis that moment conditions (\ref{eq:SDF_r}) hold if $$T\widehat{\delta}^*_{g}> c^*_{1-\alpha},$$ 
and the critical value is:  
\begin{align}
	c^*_{1-\alpha} = \sup_{\theta \in 	CS_{r, \alpha_1}   } c^*_{1-\alpha_2}(\theta),  
\end{align}
where $c^*_{1-\alpha}(\theta)$ is the $100(1-\alpha_1)\%$ percentile of the weighted sum of  $N$ $\chi^2(1)$ random variables with weights being the non-zero eigenvalues of $
S^{\frac{1}{2}}_T( {\theta})  \widehat{Q}_r^{-1}   S^{\frac{1}{2}'}_T( {\theta})$, 
$\alpha_1,$  $ \alpha_2$ are chosen such that $\alpha_1>0,$  $ \alpha_2>0$ and $\left(1-\alpha_1 \right)\left(1-\alpha_2 \right)=1-\alpha$ with $\alpha$ the overall significance level.

Our HJS specification test procedure combines a less powerful but robust statistic (AR) with a non-robust one (HJ) to incorporate the model identification strength in our testing procedure, and in later discussion we show that this test improves performance in size (compared with the HJ test) and power (compared with the $\mathcal{J}$ test).

Before we proceed to show the size and power performances of the HJS specification test (Theory \ref{theo:2} and Theory \ref{theo:power of HJS} ), we first discuss the properties of the robust confidence set $CS_{r,\alpha}$.  \cite{kleibergen2019Consumption} study a similar robust risk premia confidence set using the GRS-FAR statistic. They show that this kind of set can be unbounded in certain cases. Therefore, for practical reason, we restrict the parameter space $\Theta$ to be a compact set (Assumption \ref{ass}), of which the robust confidence set is a subset. By construction, when the model is strongly identified we would expect the confidence set to shrink to a point as sample size grows, and when the model is weakly identified or even unidentified the diameter of this set can be arbitrarily large.  

\begin{lem}\label{lem: }
	Suppose the assumptions in Corollary \ref{theo: charaterization of the delta 2} hold,  then 
	\begin{align*}
		\lim\inf_T			\mathbb{P}\left(\theta_G\in  CS_{r,\alpha} \right) \geq  1- \alpha
	\end{align*}
	\noindent
	\textit{Proof:} See Appendix \ref{Proof of Lemma reflem: }.	
\end{lem}
Lemma \ref{lem: } implies that the confidence set covers the true value with the requested probability asymptotically even   in the presence of weak (proxy) factors,      
which is essential for the correct size performance of the HJS test. This result holds under more general cases, for example it holds even when the model is not identified, and  this correct coverage probability of the confidence set directly results from the correct size of the identification robust AR test statistic.

\begin{theo} \label{theo:2}
	Suppose the assumptions in Lemma \ref{lem: } hold, 
	\begin{align*}
		\lim\sup_T			\mathbb{P}\left(T\widehat{\delta}^*_g \geq  c_{1-\alpha}^* \right) \leq \alpha
	\end{align*}  	\noindent
	\textit{Proof:} See Appendix \ref{proof of theorem theo:2}.
\end{theo}
Theorem \ref{theo:2} shows that $c_{1-\alpha}^*$ provides a upper bound for the HJS statistic, which is also a upper bound for the HJ statistic as the HJ statistic is smaller than the HJS statistic by construction, and that the HJS specification test is size correct in the presence of weak (proxy) factors. The proof of it implies that the size property of the HJS specification test is a direct result of Lemma \ref{lem: }, and given that the lemma holds for more general conditions, we know that the HJS specification test can be extended to more general cases as well. Theorem \ref{theo:2} also implies that the HJS specification test is conservative, which is understandable as we use the infimum to construct the HJS statistic instead of the supremum. However given the diameter of the robust confidence set can be arbitrarily large, using the supremum can lead to size distortion (see Example \ref{example}).

Even though it is conservative, the HJS specification test has better power performance compared with another well-know specification test, the $\mathcal{J}$ specification test. The $\mathcal{J}$ specification test statistic is also constructed based on the AR statistic such that  
\begin{align*}
	\mathcal{J} = \inf_{\theta} AR(\theta) 
\end{align*}  
\cite{gospodinov2017spurious} show that the $\mathcal{J}$ specification test is size correct in the presence of spurious/useless factors, which means $q_G$ is of reduced rank and the model is not identified, but it has a complete power loss in such cases. We extend their results to weakly identified models. Theory \ref{theo:power of HJS} shows that in both unidentified and weakly identified models, the $\mathcal{J}$ specification test suffers from power loss, while our HJS test still maintains proper power performance.      
\begin{theo}\label{theo:power of HJS} 
	Suppose Assumptions \ref{assum:CLT for the proxy factors}, \ref{assum: factor structure in r}, \ref{assum: forth moments of ee} hold, but instead of the correct proxy factors $g_t$, proxy factors $\widetilde{g}_t$ are used such that $\widetilde{g}_t$ is a $\widetilde{K}\times 1$ vector, $\norm{\iota_N -q_{\widetilde{G}} \theta_{\widetilde{G}}}>a > 0, \forall \theta_{\widetilde{G}} \in \Theta $ with $q_{\widetilde{G}} =\mathbb{E}(r_t\widetilde{G}_t'), \widetilde{G}_t=(1,\widetilde{g}_t)$ and the model is misspecified. In addition, assume that the $\psi_{B_{{g}}}$, when we replace ${G}_t$ with $\widetilde{G}_t$, in Lemma \ref{lem:B_g} satisfies that $\psi_{B_{{g}}}\sim N(0, Q_{\widetilde{G}} \otimes \Sigma)$ with $Q_{\widetilde{G}} = \mathbb{E}(\widetilde{G}_t\widetilde{G}_t')$ and $\Sigma$ the covariance matrix of $\widetilde{u}_{g,t}$. Let $H=\left(\iota_N, q_{\widetilde{G}} \right)$.   
	\begin{itemize}
		\item[(i)] (\textit{\cite{gospodinov2017spurious}},  Theorem 2, unidentified model under misspecification) Suppose $H$ has a column rank $K+1-k$ for an integer $k\geq 1$, then we have 
		\begin{align*}
			\mathcal{J} 	\preceq_d w_{k,i} 
		\end{align*}   
		where $w_{k,i}$ is the smallest eigenvalue of $\mathcal{W}_k\sim \mathcal{W}_k\left(N-K-1+k,I_k \right)$ and $\mathcal{W}_k\left(N-K-1+k,I_k \right)$ denotes the Wishart distribution with $N-K-1+k$ degrees of freedom and a scaling matrix $I_k$. Furthermore,
		\begin{align*}
			&	\lim\sup_T			\mathbb{P}\left( \mathcal{J} \geq  c_{\chi_{N-K}^2, 1-\alpha} \right) \leq \alpha, \\
			& 	\lim\inf_T			\mathbb{P}\left(T\widehat{\delta}^*_g \geq  c_{1-\alpha}^* \right) =1,
		\end{align*} 
		with $c_{\chi_{N-K}^2, 1-\alpha}$ the $1-\alpha$ quantile of $\chi_{N-K}^2$. 
		\item[(ii)] (\textit{weakly identified model under misspecification}) Suppose $(HQ_{x})'*(HQ_{x})$ with $Q_x=\textit{diag}(I_{K+1-k},\sqrt{T}I_{k})$ converges to a positive definite matrix, then we have 
		\begin{align*}
			\mathcal{J} \rightarrow_d w_{k,ii} 
		\end{align*}   
		where $w_{k,ii}$ is the smallest eigenvalue of     
		$\mathcal{W}_k\sim \mathcal{W}_k\left(\mu, N-K-1+k,I_k \right)$ and $\mathcal{W}_k\left(\mu,N-K-1+k,I_k \right)$ denotes the non-central Wishart distribution with $N-K-1+k$ degrees of freedom and a scaling matrix $I_k$, a location parameter $\mu$ ($\mu$ is specified in the proof). Furthermore,
		\begin{align*}
			&	\lim\inf_T			\mathbb{P}\left( \mathcal{J} \geq  c_{\chi_{N-K}^2, 1-\alpha} \right) <1, \\
			& 	\lim\inf_T			\mathbb{P}\left(T\widehat{\delta}^*_g \geq  c_{1-\alpha}^* \right) =1,
		\end{align*} 
		with $c_{\chi_{N-K}^2, 1-\alpha}$ the $1-\alpha$ quantile of $\chi_{N-K}^2$. 
	\end{itemize}

	\noindent
	\textit{Proof:} See Appendix \ref{Proof of Theorem theo:power of HJS}.
\end{theo}





\subsection{Simulation exercises} \label{sec:simu HJS}
In this section, we conduct a simple simulation exercise with a single-factor model to evaluate the empirical rejection rates (the size and power performance) of our proposed HJS specification test. We calibrate the data generating process in our simulations to match the data set from \cite{kroencke2017asset}. We simulate the factor $f_t\sim ~i.i.d.~N(0,V_f)$, where we set $V_f$ to match the sample variance of the consumption growth factor. $r_t$ is generated with one factor $f_t$ via $r_t=\iota_N + \beta \lambda+ \beta_\perp d +\beta d_g f_t +u_t$, $u_t\sim ~i.i.d ~N(0,V_u)$, where we match $\lambda$ to the estimated risk premium, $\beta$ is the sample slope parameter between the assets returns and  consumption growth factor, and  $V_u$ is the sample covariance of the residuals resulting from regressing asset returns on a constant and the consumption growth factor. $\beta_\perp$ is a vector which is orthogonal to $\iota_N, \beta$ and $\norm{\sqrt{T}\beta_\perp}=1$. We set $T=100$.
We use $d_g$ to tune the identification strength of the factors in our simulation exercise where a larger $d_g$ means  a stronger factor, and $d$ to tune the model misspecification level where a lager $d$ means a larger deviation from the moment conditions (\ref{eq:SDF_r}) for our simulated data.

For the size performance comparison, we set $d=0$ and thus moment conditions (\ref{eq:SDF_r}) hold for our simulated data. Figure \ref{fig:size_curve_HJHJS} shows that the HJ specification test is highly size distorted and the distortion only drops down slightly when we increase the identification strength of the factor, while the HJS specification test has a better finite sample behavior and remains size correct.    
\begin{figure}[h!] 
	\centering	\includegraphics[scale=0.65]{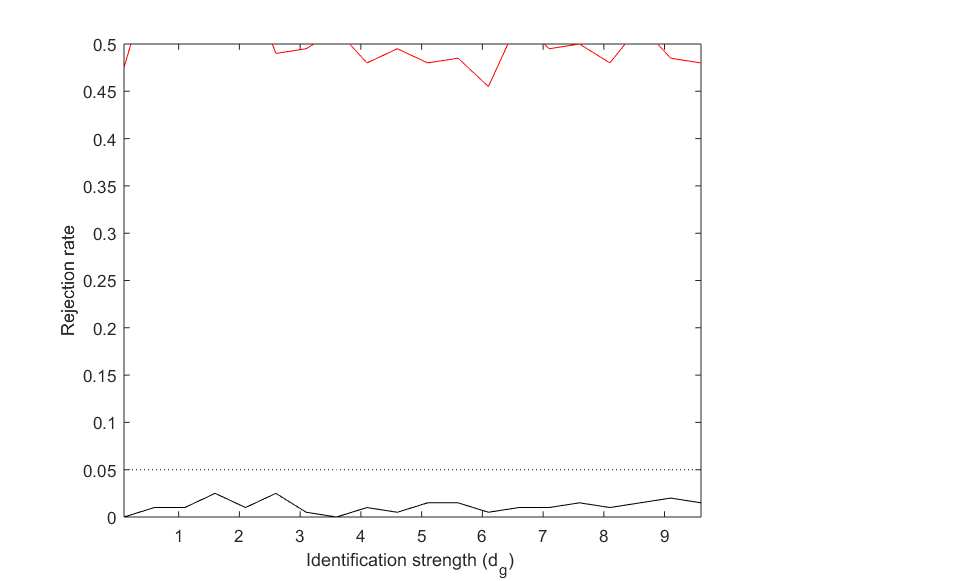} 	
	\caption{Size against strength of the factor ($d_g$): rejection frequency of  the HJ specification test (red) and  rejection frequency of  the HJS specification test (black)}
	\label{fig:size_curve_HJHJS}
\end{figure} 

For the power performance comparison, we set $d_g=0$ which means $f_t$ only serves as a spurious factor. Figure \ref{fig:power_curve_HJHJS T=100} shows that the rejection frequency of the HJS specification test increases much faster compared with the one of the $\mathcal{J}$ specification test when the level of model misspecification ($d$) increases. The rejection frequency of the $\mathcal{J}$ specification test remains relatively small even when the HJS specification test rejection frequency is close to one, and this implies the HJS specification test has better power performance.

These observations support our theory and show that the HJS specification test has good performance in both size and power.  
\begin{figure}[h!] 
	\centering	\includegraphics[scale=0.5]{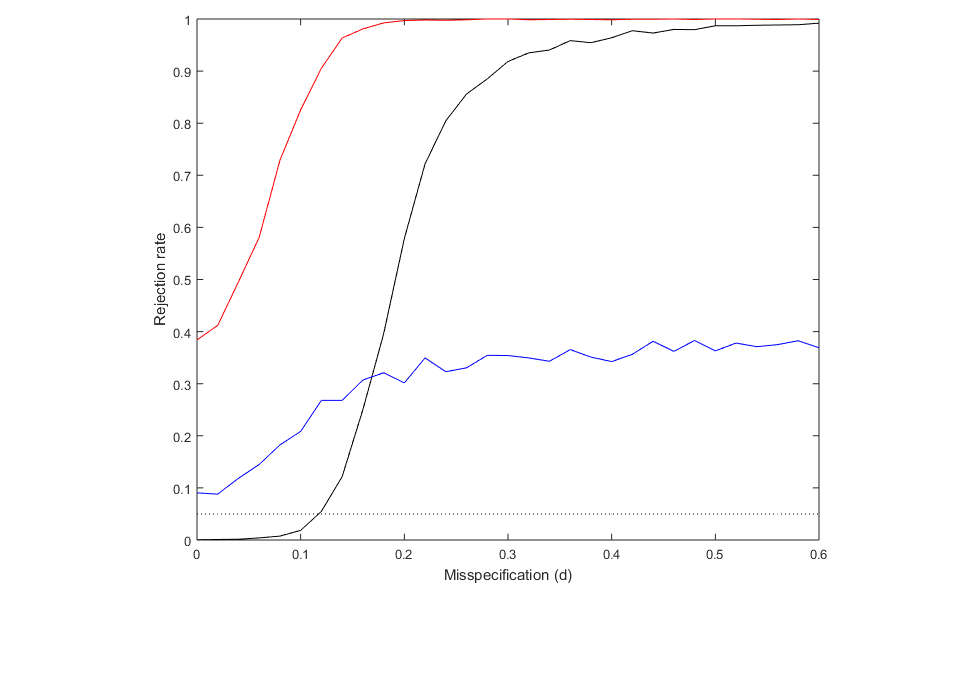} 	
	\caption{Power against the level of the model misspecificaion ($d$) with single spurious factor ($d_g=0$): rejection frequency of  the HJ specification test (red), rejection frequency of  the HJS specification test (black) and  rejection frequency of  the $\mathcal{J}$ specification test (blue)}
	\label{fig:power_curve_HJHJS T=100}
\end{figure}

\section{Specification testing with large N: HJN} \label{sec:HJN}
In the previous section, we construct the HJS statistic using a robust confidence set of $\theta_G$ since it is only weakly identified with a limited number of asset returns. The HJS specification testing procedure involves optimization steps, which is commonly done  in practice  through a grid search procedure. In this section, we provide another novel valid specification test statistic, the HJN statistic, which does not involve any time consuming optimization procedure. The construction of the HJN statistic uses a consistent $\theta_G$ estimator, and thus we first introduce our $\theta_G$ estimator and then the HJN statistic.

\subsection{Four-pass estimator} 
When we work within a double-asymptotic framework such that both the number of time periods and the number of asset returns grow, weak (proxy) factors do not necessarily lead to a weak identification problem (\cite{anatolyev2018factor}), which is similar to the case of many weak instruments that information about some parameters though limited,
aggregates slowly. Even though  $\widehat{\theta}_G$ is not consistent (Theory \ref{theo:asymptotic for theta}), another consistent estimator for $\theta_G$ can be constructed. With an extended number of asset returns, we can
estimate $\theta_G$ consistently by removing the missing factor structure via PCA and using an IV-type technique to correct for the remaining issues. The consistent estimator gives another way to construct a statistic for the HJ distance, based on which we propose a
novel specification test statistic, our HJN statistic. In the following, we first introduce our four-pass $\theta_G$ estimator with the extended number of asset returns, and thereafter provide the motivation for this four-pass procedure.

We propose the following steps to estimate $\theta_G$ with $N$ base portfolios of gross returns $r_t$:  

\noindent
\textbf{Step (1)}: Estimate $\widehat{c}, \widehat{\beta}_g$ in the linear observed-(proxy)-factor model (\ref{1}) via OLS with $N$ base portfolios of returns.   

\noindent
\textbf{Step (2)}: Determine the omitted factor structure using the following two steps: 

\indent
\textbf{(2.1)}
Determine the number of factors, $K_{vz}$, in $\widehat{u}_{g,t}=r_t - \widehat{c}-\widehat{\beta}_g \bar{g}_t$ by  
\begin{align}
	\widehat{K}_{vz} = \arg \min_{0\leq j \leq K_{vz,\max}} \left( N^{-1}T^{-1} \lambda_j\left( \widehat{u}_{g}\widehat{u}_{g}' \right)  + j\phi(N,T)  \right) -1   \nonumber 
\end{align}
where $\lambda_j(A)$ is the $j$-th largest eigenvalue of a given matrix $A$, 
$\widehat{u}_{g}$ is $T\times N$ matrix stacked with the OLS residuals $\widehat{u}_{g,t}$, $K_{vz,\max}$ is an arbitrary upper bound for $K_{vz}$ and $\phi(N,T)$ is a penalty function with the properties $\phi(N,T)\rightarrow 0, \phi(N,T)/(N^{-\frac{1}{2}}+T^{-\frac{1}{2}}) \rightarrow \infty$ (e.g., in later simulation exercise and empirical application we simply choose $\phi(N,T)= N^{-\frac{1}{4}}+T^{-\frac{1}{4}}$); 

\indent
\textbf{(2.2)}: Estimate the $T\times N$  common component matrix ${cc}=   {x}  {b}'$ stacked with the common components $cc_{t}$, $\widetilde{cc} = \widehat{x} \widehat{b}$, such that $ \widehat{x} $ is equal to $\sqrt{T}$ times the eigenvector associated with the $\widehat{K}_{vz}$ largest eigenvalues of the matrix $\widehat{u}_g\widehat{u}_g' $, and $\widehat{b}= \widehat{x}' \widehat{u}_g/T$ corresponds with the OLS estimator regressing $\widehat{u}_g$ on $\widehat{x}$:
$$ \left( \widehat{x},  \widehat{b} \right) = \arg \min_{b_i, x_t  ~s.t.~\sum_{t=1}^{T} x_tx_t'  /T = I_{\widehat{K}_{vz}}  } \sum_{i,t}  \left( \widehat{u}_{g,it} -   b_{i}'  x_{t} \right)^2  $$

\noindent 
\textbf{Step (3)}: Split the sample into two non-overlapping subsamples along the time index  and remove the missing factor structure from the regressors  in the SDF of both subsamples: 
$$\widetilde{q}_{G,T}^{(i)}= q_{G,T}^{(i)}-  \widetilde{cc}_{G,T}^{(i)'} , i=1,2$$ where
\begin{align*}
	& q_{G,T}^{(1)} = \frac{1}{\lfloor\frac{T}{2} \rfloor} \sum_{t=1}^{\lfloor\frac{T}{2} \rfloor} r_t G_t';~ q_{G,T}^{(1)} = \frac{1}{T-\lfloor\frac{T}{2} \rfloor} \sum_{t=\lfloor\frac{T}{2} \rfloor+1}^{T} r_t G_t' \\
	& \widetilde{cc}_{G,T}^{(2)} =\frac{1}{\lfloor\frac{T}{2} \rfloor} \sum_{t=1}^{\lfloor\frac{T}{2} \rfloor} \widetilde{cc}_t G_t';~ \widetilde{cc}_{G,T}^{(2)} =\frac{1}{T-\lfloor\frac{T}{2} \rfloor} \sum_{t=\lfloor\frac{T}{2} \rfloor+1}^{T} \widetilde{cc}_t G_t'. 
\end{align*}

\noindent
\textbf{Step (4)}: We then use IV regression to derive two estimators, $\widetilde{\theta}_{G}^{(i)}, i=1,2$, where we use $\widetilde{q}_{G,T}^{(1)}$ as instrument for $\widetilde{q}_{G,T}^{(2)}$ and vice versa. Thereafter, our proposed four-pass estimator is derived by taking the average of both estimators:   
$$\widetilde{\theta}_G = \sum_{i=1}^2 \widetilde{\theta}_{G}^{(i)}/2, $$ 
with $
\widetilde{\theta}_{G}^{(1)} = \left(\widetilde{q}^{(1)' }_{G,T} P_{\widetilde{q}^{(2)}_{G,T}} \widetilde{q}^{(1)}_{G,T} \right)^{-1}  \widetilde{q}^{(1)'}_{G,T}   P_{ \widetilde{q}^{(2)}_{G,T}  } \iota_N 
$ and $
\widetilde{\theta}_{G}^{(2)} = \left(\widetilde{q}^{(2)' }_{G,T} P_{\widetilde{q}^{(1)}_{G,T}} \widetilde{q}^{(2)}_{G,T} \right)^{-1}  \widetilde{q}^{(2)'}_{G,T}   P_{ \widetilde{q}^{(1)}_{G,T}  } \iota_N$.\\

Our estimation approach for $\theta_G$ resolves the problems of the missing factor structure and the weak (proxy) factors simultaneously. We make use of the results from  \cite{bai2002determining}, \cite{bai2003inferential} and \cite{giglio2017inference} in step (2) to recover the common components in the error terms using principal component analysis, and we use the instrumental variable idea applied for the factor models, which is used in \cite{anatolyev2018factor}, in step (4) to solve potential endogeneity issues. Compared with the estimator proposed in \cite{anatolyev2018factor}, our proposed estimator relaxes the restrictions on the number of omitted factors and the restrictions on the rank of the loadings of all the factors present in the model. 

To illustrate why our proposed procedure is robust against weak (proxy) factors and a missing factor structure, we start by comparing it with the conventional $\theta_G$ estimator. To do so, we first rewrite equation (\ref{eq:SDF_r}) ($\iota_N = \mathbb{E}q_{G,T} \theta_G$) as
\begin{align*}
	\iota_N=	 {q}_{G,T}  {\theta}_{G}  - \epsilon_{q_G}{\theta}_{G},
\end{align*}
with $ \epsilon_{q_G} = \left({q}_{G,T} - \mathbb{E}{q}_{G,T} \right)$ which is correlated with $ q_{G,T}$. The term $\epsilon_{q_G}$ vanishes asymptotically, and so it is dominated by $q_{G,T}$ when all proxy factors are strong. The conventional estimator, which results from regressing $\iota_N$ on $ q_{G,T}$, is then valid in large samples, since $\epsilon_{q_G}$ becomes negligible. However, if some (proxy) factors are weak, some columns of $q_{G,t}$ are of the same order as $\epsilon_{q_G}$, then there would be a classic endogeneity problem if we simply regress $\iota_N$ on  ${q}_{G,T} $. 

To solve the endogeneity problem, a valid instrument can be constructed in our framework with a split-sample technique and this idea is also employed in \cite{anatolyev2018factor}. Given the independence of the $e_t$ from non-overlapping sub-samples,  $q_{G,T}^{(1)}$ can serve as an instrument for $q_{G,T}^{(2)}$ and vice versa  when there is no missing factor structure ($K_{vz}=0$) and this is the starting point of our proposed procedure. When there is a missing factor structure with factors that might be correlated across time,  $q_{G,T}^{(1)}$ is no longer a valid instrument for $q_{G,T}^{(2)}$. Therefore, we use $\widetilde{q}_{G,T}^{(i)}, i=1,2$ which results from removing the missing factor structure from $q_{G,T}^{(i)},i=1,2$. By doing so, $\widetilde{q}_{G,T}^{(1)}$ is asymptotically uncorrelated with  $\widetilde{\epsilon}_{q_G}^{(2)}$, and is a valid instrument.

As shown in Theorem \ref{theo:consistency of theta tilde}, our estimation procedure provides $\min \{  \sqrt{T}, \sqrt{N}   \}$-consistent results for $\theta_G$, of which a non-linear transformation leads to a consistent risk premia estimator (Corollary \ref{theo:consistency of theta tilde}).

%
%
\begin{theo}\label{theo:consistency of theta tilde} 
	Suppose Assumptions \ref{assum: factor structure in r} - \ref{assum:factor loading strength 2}, \ref{ass: factor loadings additional requirement }  - \ref{assum: (New, products of factor and e)} hold,  and $N/T \rightarrow c$. 
	\begin{align*}
		\sqrt{NT}	Q^{-1}_{B_g,T}\left(\widetilde{\theta}_{G} - \theta_G\right) \rightarrow O_p(1),  
	\end{align*}	
	with $Q_{B_g,T} = \text{diag}(I_{1+K_{g,1}},\sqrt{T}I_{1+K_{g,1}})$.\footnote{With some additional regular assumptions, we can construct $\widehat{\Sigma}_{\theta_G}$ such that $	\widehat{\Sigma}_{\theta_G}^{1/2}\left(\widetilde{\theta}_{g} - \theta_G\right) \rightarrow_d N(0,I)$. See  Appendix \ref{Proof of Theorem theo:consistency of theta tilde}. } 	 \\

	\noindent
	Proof: See Appendix \ref{Proof of Theorem theo:consistency of theta tilde} .
\end{theo}


\begin{coro} Suppose the assumptions in Theorem \ref{theo:consistency of theta tilde} hold, then
	\begin{align*}
		\sqrt{NT}	Q^{-1}_x\left(\widetilde{\lambda}_{g} - \lambda_g\right) \rightarrow O_p(1),  
	\end{align*}
	with $\widetilde{\lambda}_{g}=-V_g \text{diag}(0_{K\times 1},I_{K})\widetilde{\theta }_{G}/\widetilde{\theta }_{G,1}$.\\ 
	
	\noindent
	Proof: This is a direct result of Theorem \ref{theo:consistency of theta tilde}.  
\end{coro}

\subsection{HJN specification test}
\cite{kleibergen2018identification} study risk premia on mimicking portfolios by projecting non-traded factors on traded base portfolios, and then carry out identification robust tests using a set of testing portfolios. We use a similar idea
to construct the HJN specification test: \\
\textbf{Step (1)} Estimate $\theta_G$ from a set of N \textit{base portfolios} $r_t$ of asset returns using our proposed four-pass estimator $\widetilde{\theta}_G$

\noindent
\textbf{Step (2)} Estimate $q_G$ from a set of \textit{testing portfolios} $R_t$ of n asset returns and $n$ is fixed such that $\widetilde{q}_G=  \frac{1}{T}\sum_{t=1}^{T} R_tG_t'$

\noindent
\textbf{Step (3)} The HJN statistic is  $$
\widetilde{\delta}_{g}^{2} = \widetilde{e}_T' \widehat{Q}_R^{-1} \widetilde{e}_T,$$  
with sample pricing errors $\widetilde{e}_T= \iota_n - \widetilde{q}_G\widetilde{\theta}_G$,  $\widehat{Q}_R= \frac{1}{T}\sum_{t=1}^{T} R_tR_t'$.\\

\noindent
\textbf{Step (4)} This test would reject the null hypothesis that moment conditions (\ref{eq:SDF_r}) hold if $$T\widetilde{\delta}_{g}> \widetilde{c}_{1-\alpha},$$ 
where $\widetilde{c}_{1-\alpha}$ is the $1-\alpha$ quantile of the weighted sum of  $N$ $\chi^2(1)$ random variables with weights being the positive eigenvalues of the matrix  $\widetilde{S}^{\frac{1}{2}} (\widehat{Q}_R)^{-1}\widetilde{S}^{\frac{1}{2}'}$ with $\widetilde{S}$ a consistent estimator 
of the long-run variance
matrix of the sample pricing errors ${e}_{T,R}(\theta_G)=\iota_n - \widetilde{q}_G {\theta}_G$. \\

\noindent
\textbf{Remark:} the HJN specification test does not require  the base portfolios and testing portfolios to be non-overlapping.  

\begin{coro} \label{coro:novel test 1} Suppose the assumptions in Theorem \ref{theo:consistency of theta tilde} hold, then
	\begin{align*}
		T 	\widetilde{\delta}_{g}^{2}  \sim_d  \sum_{i=1}^{n}  {p}_i x_i   
	\end{align*}
	with $x_i$ being independently $\chi^2(1)$ distributed random variables and $p_i$ being the positive eigenvalues of the matrix  $\widetilde{S}^{\frac{1}{2}} (\widehat{Q}_R)^{-1}\widetilde{S}^{\frac{1}{2}'}$ with $\widetilde{S}$ a consistent estimator 
	of the long-run variance
	matrix of the sample pricing errors ${e}_{T,R}(\theta_G)$ (for example, one may simple choose $\widetilde{S}= \frac{1}{T}\sum_{t=1}^T e_{g,t,R} (\widetilde{\theta}_G)e_{g,t,R} (\widetilde{\theta}_G)', e_{g,t,R}({\theta}_g)= \iota_N -R_tG_t'{\theta}_g$). \\
	
	\noindent
	Proof: See Appendix \ref{proof coro:novel test 1}.
\end{coro}

Theorem \ref{theo:3} shows that our HJN specification test is size correct even with weak (proxy) factors.

\begin{theo} \label{theo:3}
	Suppose the assumptions in Theorem \ref{theo:consistency of theta tilde} hold, 
	\begin{align*}
		\lim\sup_T			\mathbb{P}\left(T\widetilde{\delta}_g \geq  \widetilde{c}_{1-\alpha} \right) =\alpha
	\end{align*}  
	with $\widetilde{c}_{1-\alpha}$ the $1-\alpha$ quantile of the weighted sum of  $N$ $\chi^2(1)$ random variables with being the positive eigenvalues of the matrix  $\widetilde{S}^{\frac{1}{2}} (\widehat{Q}_R)^{-1}\widetilde{S}^{\frac{1}{2}'}$ with $\widetilde{S}$ a consistent estimator 
	of the long-run variance
	matrix of the sample pricing errors ${e}_{T,R}(\theta_G)$ (for example, one may simple choose $\widetilde{S}= \frac{1}{T}\sum_{t=1}^T e_{g,t,R} (\widetilde{\theta}_G)e_{g,t,R} (\widetilde{\theta}_G)', e_{g,t,R}({\theta}_g)= \iota_N -R_tG_t'{\theta}_g$).\\

	\noindent
	\textit{Proof:} This is a direct result from  Corollary \ref{coro:novel test 1}.
\end{theo}

\subsection{Simulation exercise and empirical application}
Similar to section \ref{sec:simu HJS}, we again evaluate the empirical rejection rates of our HJS specification test via simulation exercises. 

We calibrate to the data set used in  \cite{anatolyev2018factor}: the monthly returns on 100 Fama-French portfolios sorted by size and book-to-market and three Fama-French factors ($g_t$). From the portfolio returns we obtain the first four principal components (PC), and we regard the first three PCs as priced latent common factors, $f_t$, and the fourth one as the omitted factor, $z_t$. With normalization $\left(f,z\right)'\left(f,z\right)/T =I_{4}$, we set the variance of these factors to be $1$. We regress demeaned returns on $f_t, z_t$ for their risk exposures, and calculate the sample mean $\mu_{\beta\gamma}$  and  the sample variance $V_{\beta\gamma}$  of the risk exposures. We compute the sample variance $\sigma_e^2 I_N$ of the residuals after regressing returns on $f_t,z_t$. To maintain the relation between observed factors $g_t$ and PCs $f_t$, we regress $f_t$ on the three Fama French factors to obtain the slope $d_g$ and residual covariance matrix $V_v$, and $d_g$ captures the quality of the proxy factors. 

We then simulate our data in the following way. In the first step we simulate observed factors from i.i.d $N(0,I_3)$ and latent factors $f_t$ are generated by $Ad_g g_t + v_t$ with $v_t$ simulated as i.i.d $N\left(0, (I-A)d_gd_g'(I-A)+V_v \right)$. $A$ is a diagonal matrix which we use to adjust the strengths of our (proxy) factors, and we set $A=\text{diag}(I_2,d_{\alpha})$ in our simulations with $d_{\alpha}$ tuning the strength of the simulated factors. As for the corresponding risk exposures,  we use $(\beta_i', \gamma_i')'\sim$ i.i.d $N(\mu_{\beta\gamma}, V_{\beta\gamma})$. Then in the end, we generate $r_t=\iota_N+\beta_\perp d+\beta \lambda+\beta f_t +\gamma z_t + e_t$, $e_t\sim ~i.i.d ~N(0,\sigma_e^2 I_N)$, 
where we match $\lambda$ to the estimated risk premia resulting from the data. $\beta_\perp$ is a vector which is orthogonal to $\iota_N, \beta$ and $\norm{\beta_\perp}=1$. Similar to the previous simulation setting in Figure \ref{fig:power_curve_HJHJS T=100}  we use $d$ to tune the model misspecification level and $d=0$ when we simulate size curves.

In our simulations, we fix $N=100$. For the HJN specification test we use all the simulate 100 asset gross returns to form the base portfolios and the first 25 to form the testing portfolios, and we use the testing portfolios for the conventional HJ specification test.

Figure \ref{fig:size_curve_HJHJN} compares the size curves of the HJ specification test and of the HJN specification test. It shows that the HJ specification test is highly size distorted even when $d_{\alpha}$ is large  (left hand side panel of Figure \ref{fig:size_curve_HJHJN}) and the size distortion increases when proxy factors become weaker (smaller $d_{\alpha}$), while the HJN specification test roughly remains size correct. Even with relatively large number of time periods, the conventional HJ specification test still over-reject. Observations from Figure \ref{fig:size_curve_HJHJN} also seem to imply that our HJN specification test tends to under-reject in finite samples, and to show it performs well when the model is misspecified we also simulate power curves in Figure \ref{fig:power_curve_HJHJN T300 T1000}.    

\begin{figure}[h!] 
	\centering	\includegraphics[scale=0.55]{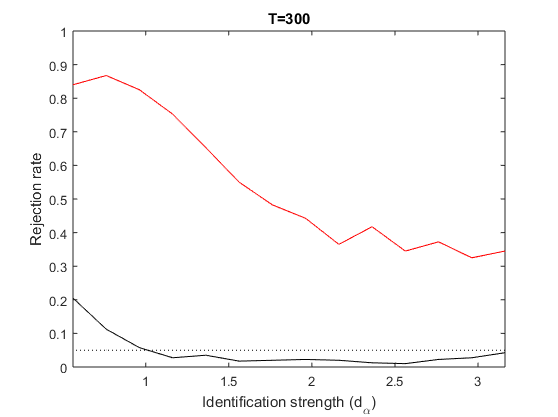} 	\includegraphics[scale=0.55]{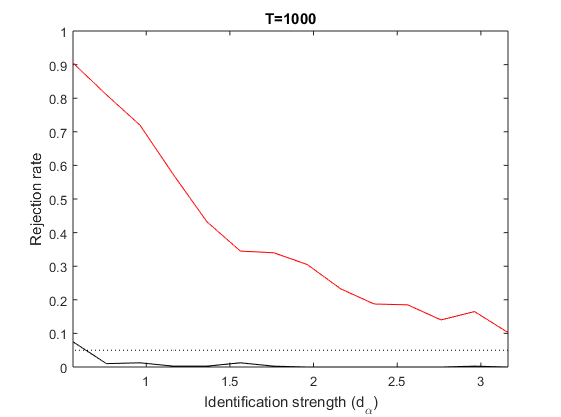} 
	\caption{Size against strength of the factor ($d_{\alpha}$): rejection frequency of  the HJ specification test (red) and  rejection frequency of  the HJN specification test (black). Left hand side panel: T=300;  right hand side panel: T=1000 }
	\label{fig:size_curve_HJHJN}
\end{figure}

Figure \ref{fig:power_curve_HJHJN T300 T1000} shows power curves of the HJ specification test and the HJN specification test respectively. The left hand side panel of Figure \ref{fig:power_curve_HJHJN T300 T1000}  uses $d_{\alpha}=0.5$ to mimic one weak proxy factor, while all proxy factors in the right hand side panel are strong with a larger value of $d_{\alpha}$.  Figure \ref{fig:power_curve_HJHJN T300 T1000}
shows that the HJN specification test has proper power performance regardless of the presence of weak (proxy) factors, and rejection frequency increases faster when the proxy factor is stronger. 
\begin{figure}[h!] 
	\centering	
	\includegraphics[scale=0.55]{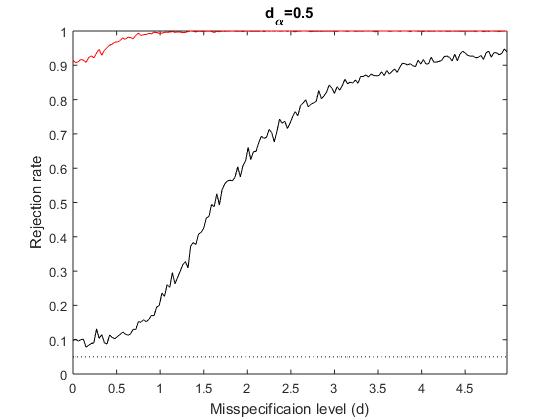} 	\includegraphics[scale=0.55]{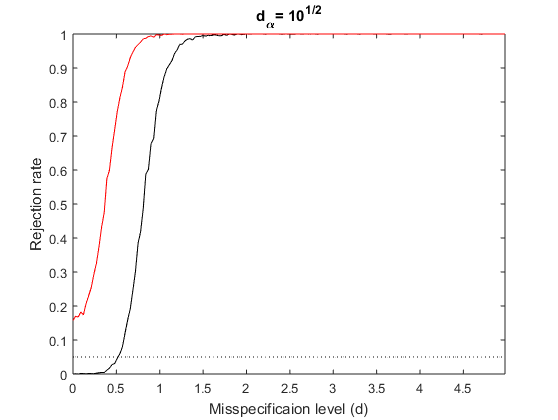} 
	\caption{Power against the level of the model misspecificaion ($d$): rejection frequency of  the HJ specification test (red) and rejection frequency of  the HJS specification test (black). Left hand side panel: with one weak proxy factor $d_{\alpha}=0.5$; Right hand side panel: with strong factors $d_{\alpha}=\sqrt{10}$.}
	\label{fig:power_curve_HJHJN T300 T1000}
\end{figure}

\section{Empirical application}\label{sec:emp}
We apply our proposed test procedures on the data set of monthly returns on 100 Fama-French portfolios sorted by size and book-to-market and the three Fama-French factors (market, SmB, HmL) and the momentum factor.    

An intuitive measure for the factor structure in asset returns is total variation of the asset returns explained by the principal components\footnote{This corresponds with the nuclear norm of the demeaned asset returns} (\cite{kleibergen2015unexplained}). We
construct the spectral decomposition of the sample covariance matrix of
the 100 portfolio returns, and denote $\lambda_1 > \lambda_2 > \cdots $ the characteristic roots (or eigenvalues of the PCs of asset returns) in descending
order. We use the characteristic roots ratios (CRRs)   ${\lambda_i}/{\left(\sum_{j}\lambda_j\right)}, i=1,2,3,4$, which represent the total variation of the portfolio returns explained by the first four PCs respectively, to check the factor structure of portfolio returns (see Figures \ref{fig:emp20yearwindow1yearstep} and \ref{fig:emp10yearwindow1yearstep}).

Figures \ref{fig:emp20yearwindow1yearstep} and \ref{fig:emp10yearwindow1yearstep} also report the p-values of specification tests (HJ and HJN) with respect to a three-FF-factor model and a four-factor (adding the momentum) model from 1963-09 to 2019-08 using rolling windows of 240 and 120 months respectively\footnote{We first choose the window size of 240 months in Figure \ref{fig:emp20yearwindow1yearstep},  because our simulations suggest sample size around 300 seems to be enough for carrying out our tests properly.}.
For the HJN specification test we use all the 100 asset returns to form the base portfolios and the first 25 to form the testing portfolios, and we use the testing portfolios for the conventional HJ specification test. They also report measures for the presence of a factor structure in the asset returns: the fraction of the total variation of the portfolio
returns that is explained by their principal components.

Figures \ref{fig:emp20yearwindow1yearstep} shows that when 
comparing nested models,  the HJ test can produce counter-intuitive results by rejecting a four-factor model but not the reduced three-factor model (see points near the coordinate '2015-01').  This is an unfortunate outcome since the four-factor mode apparently embeds the three-factor model and if the four-factor model is rejected we would expect the three-factor model to be rejected as well. We attribute this strange behavior to the momentum factor having only weak correlation with the returns and thus inducing a larger rejection rate of the HJ test, while our HJN specification test does not have such problem.     

The HJN specification test also captures changes in the factor structure of asset returns in a more sensible way compared with the HJ test. As shown in Figure \ref{fig:emp20yearwindow1yearstep}, when CRRs vary in different time periods (e.g., the total variation of the portfolio returns is mostly explained by the first PC for points near the  coordinate '2000-01' while the other PCs only account for a much lower percentage of the variation),  the HJN specification tests reflect the changes in the factor structure of asset returns with variations in p-values of tests of a four-factor model, while the HJ specification tests reject both three-factor and four-factor models for most time periods and is not informative for the factor structure of asset returns.

Both of the HJ and the HJN tests in Figure \ref{fig:emp20yearwindow1yearstep} seem to have larger p values near the coordinate '2015-01' while the patterns of characteristic roots ((a) and (b)) seem to be rather stable. This is because a 240-month window size is a bit too long, and some changes in the factor structure might be averaged out and thus not detected by CRRs. We choose a smaller rolling window size (120 months) in Figure \ref{fig:emp10yearwindow1yearstep}, and it shows the change in the factor structure  (Figure \ref{fig:emp10yearwindow1yearstep}.(a)) after the coordinate '2010-01'. Similar to what we observe in Figure \ref{fig:emp20yearwindow1yearstep}, Figure \ref{fig:emp10yearwindow1yearstep} also shows our HJN specification test respond to the factor structure in the asset returns in a more informative way, while the HJ tests only report small p values for most time periods for both three- and four-factor models.

\begin{sidewaysfigure}[h!] 
	\centering	\includegraphics[scale=0.55]{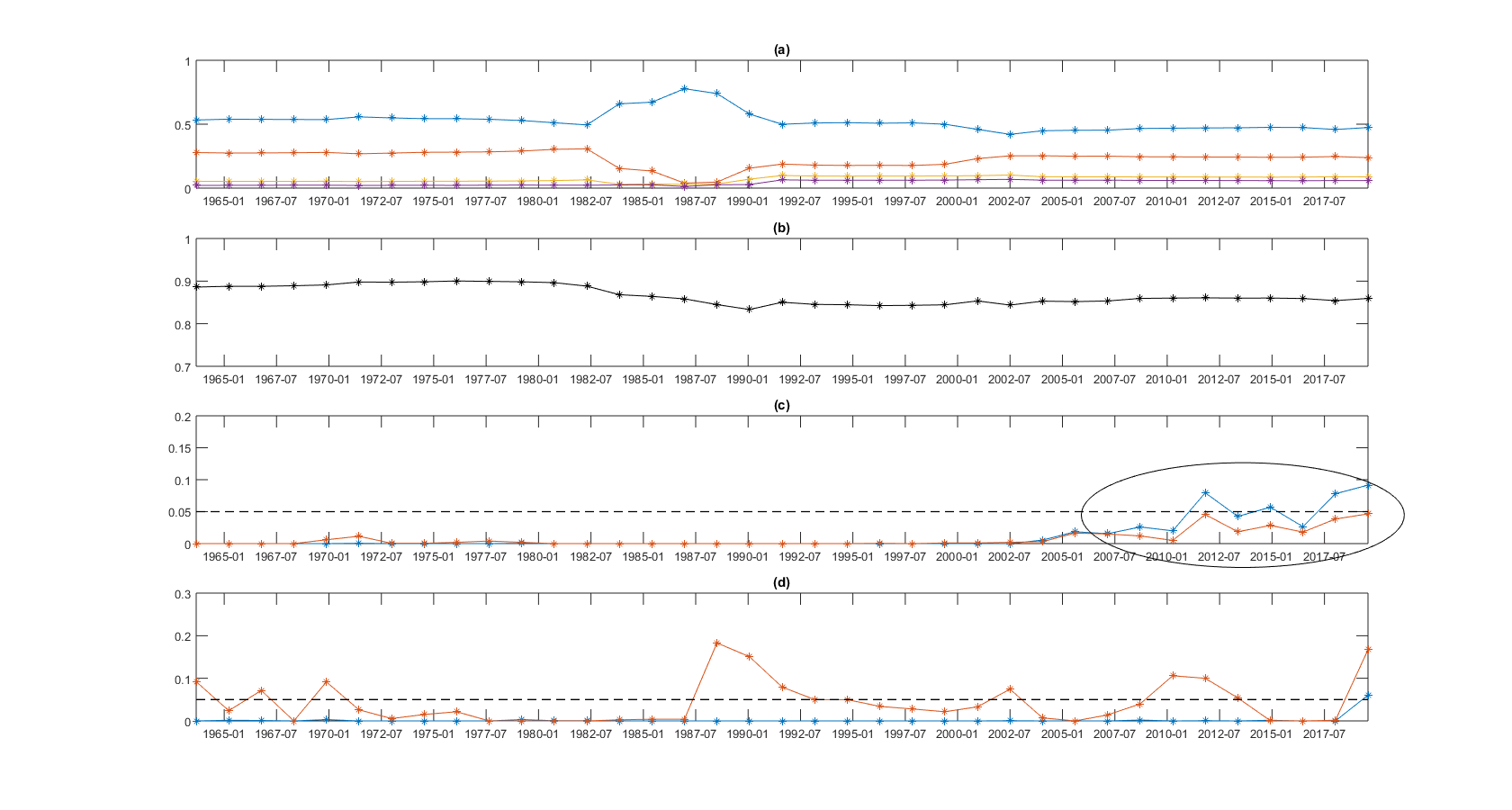} 
	\caption{\footnotesize  The time series of estimates from 1963-09 to 2019-08 with rolling windows of size T = 240 (20 years) and  the number of increments between successive rolling windows being 12 (1 year)  (x-axis is labeled with the ending period of each rolling window, for example, the first rolling window uses data from 1963-09 to 1983-08, so x-axis should label the first point with '1983-08').  (a) The fraction of the total variation of the portfolio returns that is explained by their four largest principal components respectively (CRRs) ($ \lambda_i/(\sum_{j=1}^N \lambda_j), i=1,\cdots,4$); (b) The fraction of the total variation of the portfolio returns that is explained by the sum of the four largest principal components ($ \sum_{i=1}^4\lambda_i/(\sum_{j=1}^N \lambda_j)$); (c) p-values of the HJ specification test of the three-FF-factor model (blue), p-values of the HJ specification test of the four-factor (three FF and momentum factors) model (red);(d) p-values of the HJN specification test of the three-FF-factor model (blue), p-values of the HJN specification test of the four-factor model (red).}
	\label{fig:emp20yearwindow1yearstep}
\end{sidewaysfigure} 

\begin{sidewaysfigure}[h!] 
	\centering	\includegraphics[scale=0.85]{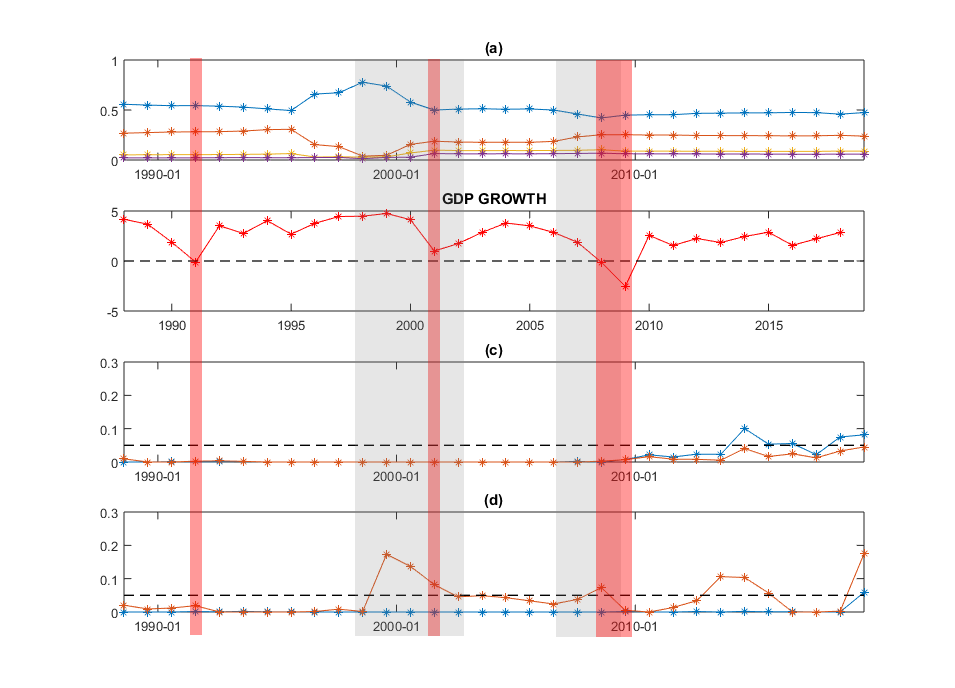} 
	\caption{The time series of estimates from 1968-09 to 2019-08 with rolling windows of size T = 240 (20 years) and  the number of increments between successive rolling windows being 12 (1 year)  (x-axis is labeled with the ending period of each rolling window, for example, the first rolling window uses data from 1968-09 to 1988-08, so x-axis should label the first point with '1988-08').  (a) The fraction of the total variation of the portfolio returns that is explained by their four largest principal components respectively (CRRs) ($ \lambda_i/(\sum_{j=1}^N \lambda_j), i=1,\cdots,4$); (b) (GDP growth) US GDP
		annual growth rate 1988-2018 (downloaded from the World Bank); (c) p-values of the HJ specification test of the three-FF-factor model (blue), p-values of the HJ specification test of the four-factor (three FF and momentum factors) model (red);(d) p-values of the HJN specification test of the three-FF-factor model (blue), p-values of the HJN specification test of the four-factor model (red).}
	\label{fig:emp20yearwindow1yearstep_2}
\end{sidewaysfigure}

\begin{sidewaysfigure}[h!] 
	\centering
	\includegraphics[scale=0.2]{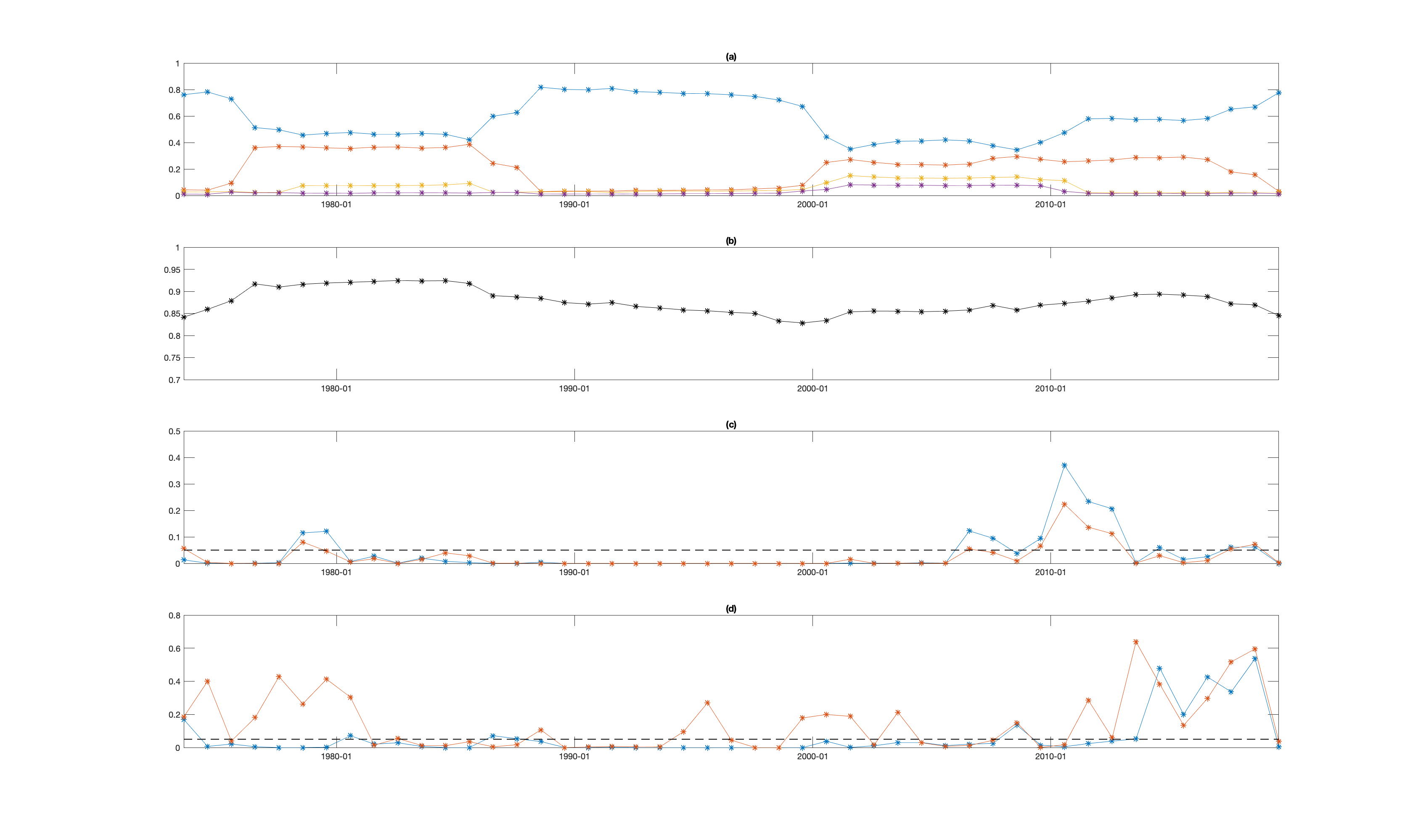}
	\caption{\footnotesize  The time series of estimates from 1963-09 to 2019-08 with rolling windows of size T = 240 (20 years) and  the number of increments between successive rolling windows being 12 (1 year)  (x-axis is labeled with the ending period of each rolling window).  (a) The fraction of the total variation of the portfolio returns that is explained by their four largest principal components respectively (CRRs) ($ \lambda_i/(\sum_{j=1}^N \lambda_j), i=1,\cdots,4$); (b) The fraction of the total variation of the portfolio returns that is explained by the sum of the four largest principal components ($ \sum_{i=1}^4\lambda_i/(\sum_{j=1}^N \lambda_j)$); (c) p-values of the HJ specification test of the three-FF-factor model (blue), p-values of the HJ specification test of the four-factor (three FF and momentum factors) model (red);(d) p-values of the HJN specification test of the three-FF-factor model (blue), p-values of the HJN specification test of the four-factor model (red). }
	\label{fig:emp10yearwindow1yearstep}
\end{sidewaysfigure} 
\clearpage

We observe in Figures \ref{fig:emp20yearwindow1yearstep} and \ref{fig:emp10yearwindow1yearstep} that the HJN specification tests in some rolling windows do not reject a four-factor model. To further study this observation, Table \ref{table:new} reports results based on the data from 1977-08 to 2019-08. We see in Table \ref{table:new} that both the HJ and HJN specification tests reject the three-factor model, and while the HJ specification test rejects the four-factor model, the HJN specification test  does not reject it. Our HJS specification test seems to be a bit conservative and does not reject both models in this application.  The estimates for the four-factor model using our proposed approach indicate a larger change in values corresponding to the momentum factor, and this might result from the momentum factor being weak. Our specification tests support a four-factor model for Fama French portfolios, and observations show that the momentum factor might only serve as a weak proxy factor which can explain the difference between the HJ and the HJN specification test results and the differences in estimated parameter values. \\

\begin{table}[h!] \centering
	\begin{small}
		\begin{tabular}{@{}cccccccc@{}}\toprule\toprule
			& HJ(p-val) & HJN(p-val) & HJS(rejected)  \\ \midrule
			n=25\\\midrule
			Three factor &  {0.000}&  {0.000 } & No\\
			\midrule 
			Four factor & \textbf{0.000}&\textbf{0.0694} &No\\
			\midrule 
			\bottomrule
		\end{tabular} 
	\end{small}
	\caption{Tests of specification using monthly returns on 100 portfolios sorted by size and book-to-market and the three Fama-French factors and the momentum factor from 1977-08 to 2019-08}\label{table:new}
\end{table}

\section{Conclusions}
We show that the HJ statistic is not a valid model selection tool and model specification test statistic when weak (proxy) factors are present. We propose two novel approaches that provide size-correct model specification tests, alongside with which we also propose novel weak (proxy) factors robust risk premia estimators. Our empirical application supports a four factor structure for Fama French portfolios despite that the momentum factor is a weak proxy factor.
\clearpage
\newpage

\bibliographystyle{plainnat}
 \bibliography{References}

\begin{thebibliography}{31}
\providecommand{\natexlab}[1]{#1}
\providecommand{\url}[1]{\texttt{#1}}
\expandafter\ifx\csname urlstyle\endcsname\relax
  \providecommand{\doi}[1]{doi: #1}\else
  \providecommand{\doi}{doi: \begingroup \urlstyle{rm}\Url}\fi

\bibitem[Ahn and Gadarowski(2004)]{ahn2004small}
Seung~C Ahn and Christopher Gadarowski.
\newblock Small sample properties of the gmm specification test based on the
  hansen--jagannathan distance.
\newblock \emph{Journal of Empirical Finance}, 11\penalty0 (1):\penalty0
  109--132, 2004.

\bibitem[Anatolyev and Mikusheva(2018)]{anatolyev2018factor}
Stanislav Anatolyev and Anna Mikusheva.
\newblock Factor models with many assets: strong factors, weak factors, and the
  two-pass procedure.
\newblock \emph{arXiv preprint arXiv:1807.04094}, 2018.

\bibitem[Bai and Ng(2002)]{bai2002determining}
J.~Bai and S.~Ng.
\newblock Determining the number of factors in approximate factor models.
\newblock \emph{Econometrica}, 70\penalty0 (1):\penalty0 191--221, 2002.

\bibitem[Bai(2003)]{bai2003inferential}
Jushan Bai.
\newblock Inferential theory for factor models of large dimensions.
\newblock \emph{Econometrica}, 71\penalty0 (1):\penalty0 135--171, 2003.

\bibitem[Bryzgalova(2016)]{2016}
Svetlana Bryzgalova.
\newblock Spurious factors in linear asset pricing models.
\newblock \emph{Working Paper}, 2016.

\bibitem[Cochrane(2009)]{cochrane2009asset}
John~H Cochrane.
\newblock \emph{Asset pricing: Revised edition}.
\newblock Princeton University Press, 2009.

\bibitem[Fama and French(1993)]{fama1993common}
Eugene~F Fama and Kenneth~R French.
\newblock Common risk factors in the returns on stocks and bonds.
\newblock \emph{Journal of Financial Economics}, 33\penalty0 (1):\penalty0
  3--56, 1993.

\bibitem[Fama and MacBeth(1973)]{fama1973risk}
Eugene~F Fama and James~D MacBeth.
\newblock Risk, return, and equilibrium: Empirical tests.
\newblock \emph{Journal of Political Economy}, pages 607--636, 1973.

\bibitem[Fang et~al.(1994)Fang, Loparo, and Feng]{fang1994inequalities}
Yuguang Fang, Kenneth~A Loparo, and Xiangbo Feng.
\newblock Inequalities for the trace of matrix product.
\newblock \emph{IEEE Transactions on Automatic Control}, 39\penalty0
  (12):\penalty0 2489--2490, 1994.

\bibitem[Giglio and Xiu(2017)]{giglio2017inference}
Stefano Giglio and Dacheng Xiu.
\newblock Inference on risk premia in the presence of omitted factors.
\newblock Technical report, National Bureau of Economic Research, 2017.

\bibitem[Gospodinov et~al.(2017)Gospodinov, Kan, and
  Robotti]{gospodinov2017spurious}
Nikolay Gospodinov, Raymond Kan, and Cesare Robotti.
\newblock Spurious inference in reduced-rank asset-pricing models.
\newblock \emph{Econometrica}, 85\penalty0 (5):\penalty0 1613--1628, 2017.

\bibitem[Hansen and Jagannathan(1997)]{hansen1997assessing}
Lars~Peter Hansen and Ravi Jagannathan.
\newblock Assessing specification errors in stochastic discount factor models.
\newblock \emph{The Journal of Finance}, 52\penalty0 (2):\penalty0 557--590,
  1997.

\bibitem[Harvey et~al.(2016)Harvey, Liu, and Zhu]{harvey2016and}
Campbell~R Harvey, Yan Liu, and Heqing Zhu.
\newblock … and the cross-section of expected returns.
\newblock \emph{The Review of Financial Studies}, 29\penalty0 (1):\penalty0
  5--68, 2016.

\bibitem[Horn et~al.(2013)Horn, Horn, and Johnson]{horn1990matrix}
Roger~A Horn, Roger~A Horn, and Charles~R Johnson.
\newblock \emph{Matrix analysis}.
\newblock Cambridge University Press, 2013.

\bibitem[Jagannathan and Wang(1996)]{jagannathan1996conditional}
Ravi Jagannathan and Zhenyu Wang.
\newblock The conditional capm and the cross-section of expected returns.
\newblock \emph{The Journal of Finance}, 51\penalty0 (1):\penalty0 3--53, 1996.

\bibitem[Kan and Zhou(2004)]{kan2004hansen}
R.~Kan and G.~Zhou.
\newblock Hansen-jagannathan distance: Geometry and exact distribution.
\newblock 2004.

\bibitem[Kan and Robotti(2008)]{kan2008model}
Raymond Kan and Cesare Robotti.
\newblock Model comparison using the hansen-jagannathan distance.
\newblock \emph{The Review of Financial Studies}, 22\penalty0 (9):\penalty0
  3449--3490, 2008.

\bibitem[Kan and Zhang(1999)]{kan1999two}
Raymond Kan and Chu Zhang.
\newblock Two-pass tests of asset pricing models with useless factors.
\newblock \emph{The Journal of Finance}, 54\penalty0 (1):\penalty0 203--235,
  1999.

\bibitem[Kleibergen(2005)]{kleibergen2005testing}
F.~Kleibergen.
\newblock {Testing parameters in GMM without assuming that they are
  identified}.
\newblock \emph{Econometrica}, pages 1103--1123, 2005.

\bibitem[Kleibergen and Zhan(2018)]{kleibergen2018identification}
F.~Kleibergen and Z.~Zhan.
\newblock Identification-robust inference on risk premia of mimicking
  portfolios of non-traded factors.
\newblock \emph{Journal of Financial Econometrics}, 16\penalty0 (2):\penalty0
  155--190, 2018.

\bibitem[Kleibergen(2009)]{kleibergen2009tests}
Frank Kleibergen.
\newblock Tests of risk premia in linear factor models.
\newblock \emph{Journal of Econometrics}, 149\penalty0 (2):\penalty0 149--173,
  2009.

\bibitem[Kleibergen and Paap(2006)]{kleibergen2006generalized}
Frank Kleibergen and Richard Paap.
\newblock Generalized reduced rank tests using the singular value
  decomposition.
\newblock \emph{Journal of Econometrics}, 133\penalty0 (1):\penalty0 97--126,
  2006.

\bibitem[Kleibergen and Zhan(2015)]{kleibergen2015unexplained}
Frank Kleibergen and Zhaoguo Zhan.
\newblock Unexplained factors and their effects on second pass r-squared’s.
\newblock \emph{Journal of Econometrics}, 189\penalty0 (1):\penalty0 101--116,
  2015.

\bibitem[Kleibergen and Zhan(2020)]{kleibergen2019Consumption}
Frank Kleibergen and Zhaoguo Zhan.
\newblock Robust inference for consumption-based asset pricing.
\newblock \emph{The Journal of Finance}, 75\penalty0 (1):\penalty0 507--550,
  2020.

\bibitem[Kleibergen et~al.(2020)Kleibergen, Kong, and Zhan]{kong2018}
Frank Kleibergen, Lingwei Kong, and Zhaoguo Zhan.
\newblock Identification robust testing of risk premia in finite samples.
\newblock \emph{Working Paper}, 2020.

\bibitem[Kroencke(2017)]{kroencke2017asset}
Tim~A Kroencke.
\newblock Asset pricing without garbage.
\newblock \emph{The Journal of Finance}, 72\penalty0 (1):\penalty0 47--98,
  2017.

\bibitem[Lettau and Ludvigson(2001)]{lettau2001consumption}
Martin Lettau and Sydney Ludvigson.
\newblock Consumption, aggregate wealth, and expected stock returns.
\newblock \emph{The Journal of Finance}, 56\penalty0 (3):\penalty0 815--849,
  2001.

\bibitem[Mikusheva(2010)]{mikusheva2010robust}
Anna Mikusheva.
\newblock Robust confidence sets in the presence of weak instruments.
\newblock \emph{Journal of Econometrics}, 157\penalty0 (2):\penalty0 236--247,
  2010.

\bibitem[Onatski(2012)]{onatski2012asymptotics}
Alexei Onatski.
\newblock Asymptotics of the principal components estimator of large factor
  models with weakly influential factors.
\newblock \emph{Journal of Econometrics}, 168\penalty0 (2):\penalty0 244--258,
  2012.

\bibitem[Peligrad et~al.(2006)Peligrad, Utev, et~al.]{peligrad2006central}
Magda Peligrad, Sergey Utev, et~al.
\newblock Central limit theorem for stationary linear processes.
\newblock \emph{The Annals of Probability}, 34\penalty0 (4):\penalty0
  1608--1622, 2006.

\bibitem[Shanken(1992)]{shanken1992ebp}
J.~Shanken.
\newblock On the estimation of beta-pricing models.
\newblock \emph{The Review of Financial Studies}, pages 1--33, 1992.

\end{thebibliography}
 \appendix
 
 	%
 
\clearpage
\newpage
 \appendix
 \small
  \section{Additional figures and tables}
This section provides additional figures and tables. Mainly, we provide more results for our empirical application in Section \ref{sec:emp}, and further illustrate how the conventional tests may fail in some empirically relevant settings. 

Figure \ref{fig:delota density 2} illustrates the density function of  the HJS statistic under same settings as in Figure \ref{fig:delota density}.  We use our previous simulation experiment calibrated to data from \cite{lettau2001consumption} to illustrate properties of the HJS statistic. The black dotted curve in Figure \ref{fig:delota density 2} suggests the value of the HJS statistic could be small when the model is not correctly specified. These smaller values are due to much smaller confidence sets $CS_r$ when model is not correctly specified. Therefore, using the HJS statistic as a model selection tool may not be a good idea even though it provides size correct specification test.   

%
%

Figures \ref{fig:emp20year} and \ref{fig:emp10year} follow the same settings as in Figures \ref{fig:emp20yearwindow1yearstep} and \ref{fig:emp10yearwindow1yearstep}, and they provide additional results by adding the p-values of the rank test (\cite{kleibergen2006generalized}) of $q_{G}$ (the rank of $q_{G}$ reflects the identification strength of the model, e.g., \cite{kleibergen2019Consumption}, \cite{kong2018}, as it shows whether the sample pricing errors vary enough as a function of $\theta$), and the p-values of the $\mathcal{J}$ specification test.

Figures \ref{fig:emp20year} and \ref{fig:emp10year} show that the $\mathcal{J}$ specification tests tend to give larger p-values, and thus it is less informative. When we use 10-year window size (Figure \ref{fig:emp10year}) instead of the 20-year window size (Figure \ref{fig:emp20year}), the p-values of the rank test increase and the lack of identification strength also increases the p-values of the $\mathcal{J}$ test. For points near the coordinate '1990-01' in Figure \ref{fig:emp10year}, the rank test can not reject that the $q_{G}$ is of reduced rank (lack of identification strength), and the corresponding $\mathcal{J}$-test p-values increase while the HJ tests tend to have larger p-values for testing the reduced three-factor model than the four-factor model. In short, in the presence of weak (proxy) factors, both well-know test statistics can not provide satisfying inference results.

Tables \ref{table:new new} and \ref{table:new 2} follow the same settings as in  Table \ref{table:new}  and reports additional results: tests with different value of n and  parameter estimates. If n is getting too large, the HJN test also suffers from finite sample issues and tends to have smaller p-values, as its validity requires $n/N\rightarrow 0$. We leave the construction of a high-dimensional robust test statistic for further study.

~\\

\begin{table}[h!] \centering
	\begin{small}
		\begin{tabular}{@{}cccccccc@{}}\toprule\toprule
			& HJ(p-val) & HJN(p-val) & HJS(rejected)  \\ \midrule
			n=30\\\midrule
			Three factor &  {0.000}&  {0.000 } & No\\
			\midrule 
			Four factor & \textbf{0.000}&\textbf{0.054} &No\\
			\midrule 
			n=35\\\midrule
			Three factor &  {0.000}&  {0.000 } & No \\
			\midrule 
			Four factor & \textbf{0.000}&\textbf{0.027} & No\\
			\midrule 
			n=40\\\midrule
			Three factor &  {0.000}&  {0.000 } &No \\
			\midrule 
			Four factor & \textbf{0.000}&\textbf{0.026} & No\\
			\midrule 
			\bottomrule
		\end{tabular} 
	\end{small}
	\caption{\footnotesize Tests of specification using monthly returns on 100 portfolios sorted by size and book-to-market and the three Fama-French factors and the momentum factor from 1977-08 to 2019-08}\label{table:new 2}
\end{table}

\begin{table}[h!] \centering
	\begin{small}  
		\begin{tabular}{@{}cccccccc@{}}\toprule\toprule
			&   Market & SMB & HML & MOM&   \\ \midrule
			(1)  &&&&&   \\
			\midrule 
			$\widehat{\theta}_G$ &    0.0480& 
			-0.0013& 
			-0.0127& $-$\\  
			$\widetilde{\theta}_G$   &    0.3081& 
			-0.1296& 
			0.1370&   $-$\\  
			$\widehat{\lambda}_g$ &  -5.2108& 
			0.4037& 
			-0.0347 &$-$\\  
			$\widetilde{\lambda}_g$ &    -5.4740& 
			0.3250& 
			-0.2261 & $-$\\  
			\midrule
			(2)&&&&&  \\
			\midrule 
			$\widehat{\theta}_G$ &     0.0472& 
			0.0174& 
			-0.0172& 
			0.0240
			\\  
			$\widetilde{\theta}_G$ &     0.1146& 
			-0.4108& 
			0.2233& 
			-0.8945& 
			\\  
			$\widehat{\lambda}_g$ &    -3.7893& 
			0.6822& 
			-0.3995& 
			1.5882\\  
			$\widetilde{\lambda}_g$ &    -2.8368&
			0.9104&
			-0.6567&
			3.6223 \\   
			\bottomrule
		\end{tabular} 
	\end{small}
	\caption{\footnotesize  Estimates with Fama-French factors and the momentum factor using monthly returns on 100 portfolios sorted by size and book-to-market from 1977-08 to 2019-08}\label{table:new new}
\end{table}

\begin{figure}[h!] 
	\centering	\includegraphics[scale=0.3]{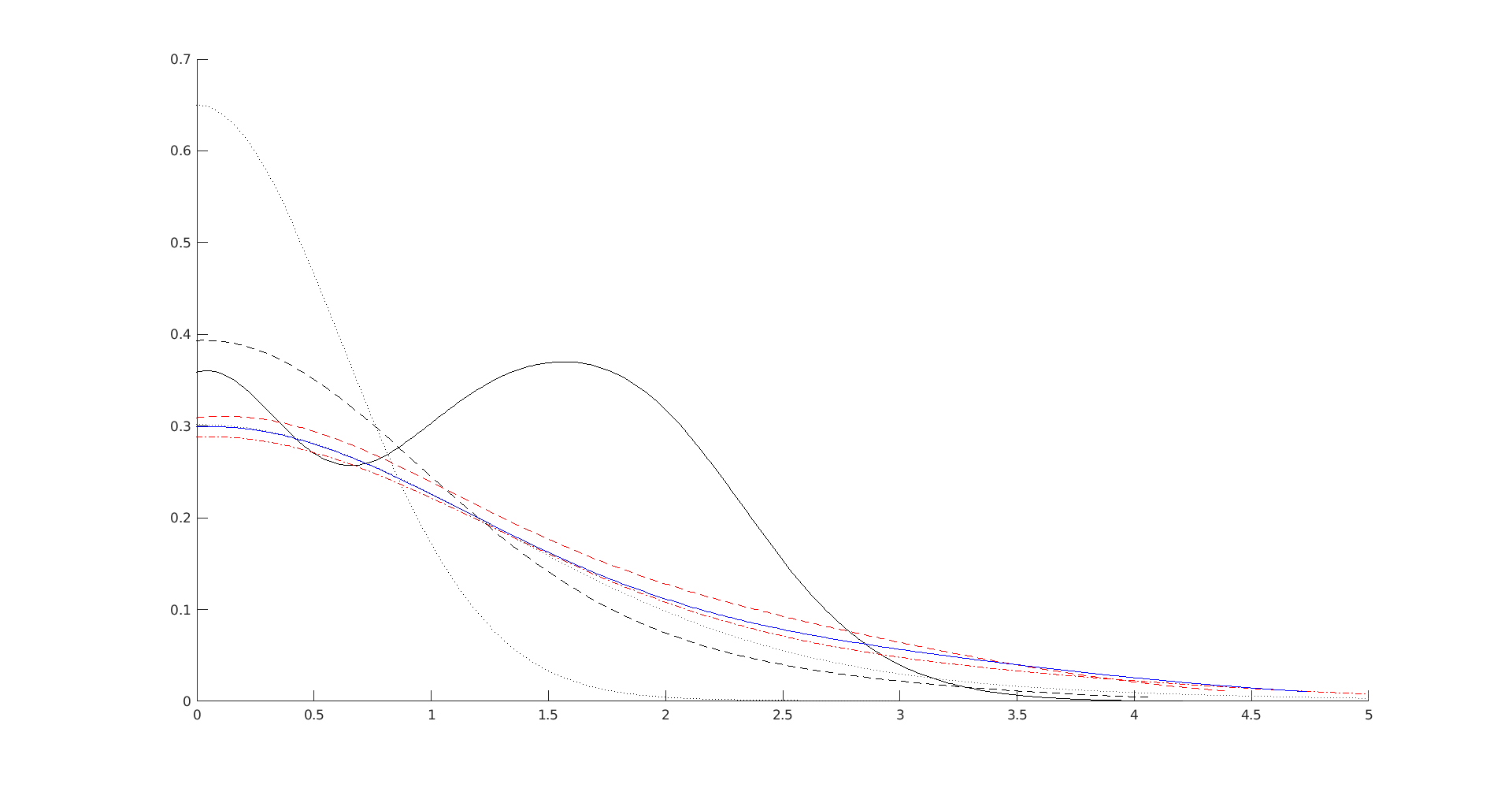} 
	\caption{\footnotesize  Densities of the HJS statistic: (1) black solid: three strong factors; (2) black dashed: two strong factors; (3) black dotted: one strong factors; (4) blue solid: three weak proxy factors; (5) red solid: three strong factors and one useless factor; (6) red dashed: two strong factors and one useless factor; (7) red dotted: one strong factors  and one useless factor; (8) red dash-dotted:one strong factors  and two useless factors }
	\label{fig:delota density 2}
\end{figure}

\begin{sidewaysfigure}[h!] 
	\centering	\includegraphics[scale=0.5]{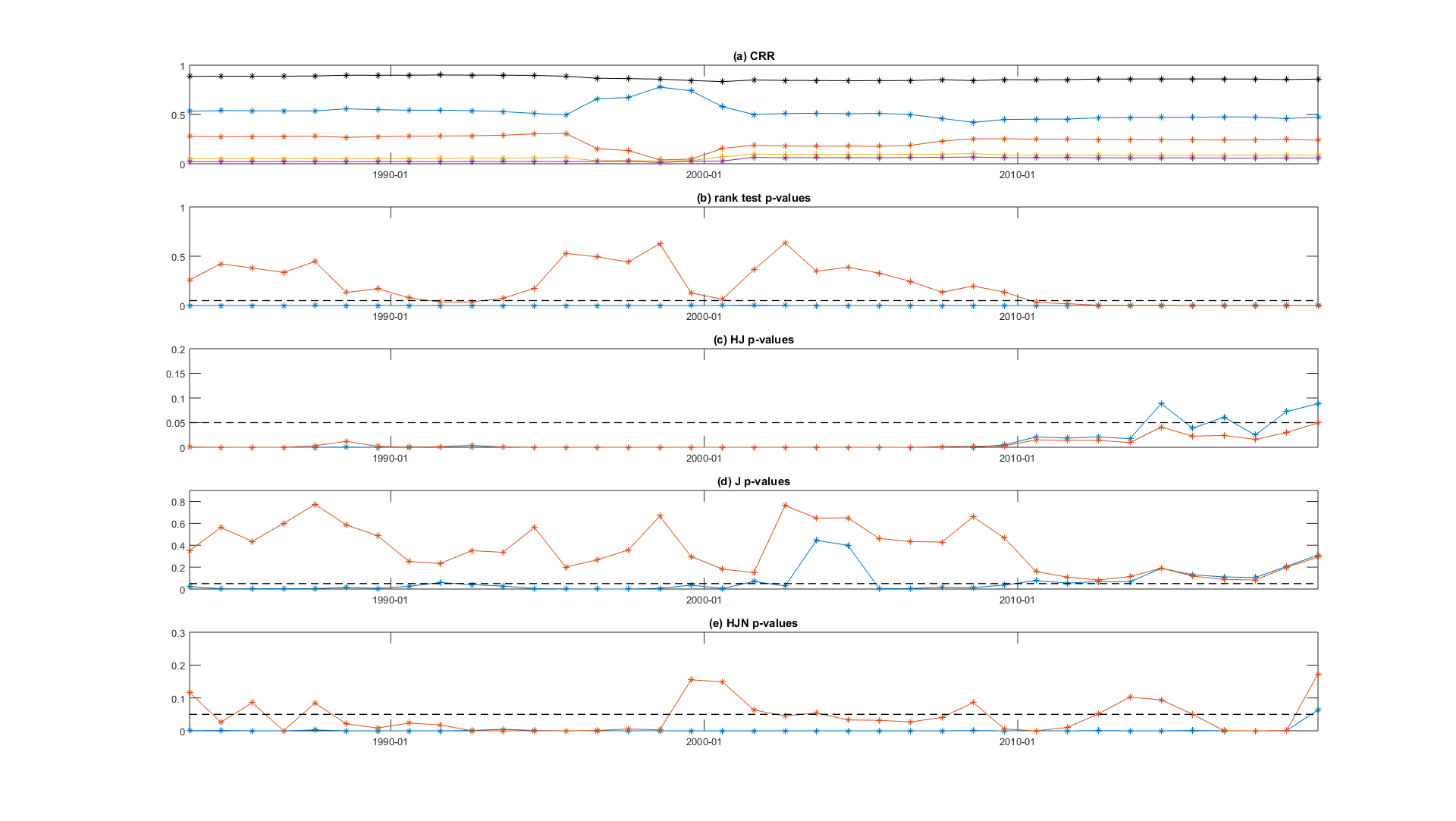} 
	\caption{\scriptsize The time series of estimates from 1963-09 to 2019-08 with rolling windows of size T = 240 (20 years) and  the number of increments between successive rolling windows being 12 (1 year)  (x-axis is labeled with the ending period of each rolling window, for example, the first rolling window uses data from 1963-09 to 1983-08, so x-axis should label the first point with '1983-08').  (a) The fraction of the total variation of the portfolio returns that is explained by their four largest principal components respectively (CRRs) ($ \lambda_i/(\sum_{j=1}^N \lambda_j), i=1,\cdots,4$); the black curve is the fraction of the total variation of the portfolio returns that is explained by the sum of the four largest principal components (sum of the first four CRRs, $ \sum_{i=1}^4\lambda_i/(\sum_{j=1}^N \lambda_j)$); (b) p-values of the rank test (\cite{kleibergen2006generalized}) of $q_{G}$ using the estimator $q_{G,T}$ with $G$ constructed by the three FF factors (blue), and the four factors (three FF and momentum factors) (red) respectively; (c) p-values of the HJ specification test of the three-FF-factor model (blue), p-values of the HJ specification test of the four-factor (three FF and momentum factors) model (red); (d) p-values of the $\mathcal{J}$ specification test of the three-FF-factor model (blue), p-values of the HJN specification test of the four-factor (three FF and momentum factors) model (red); (e) p-values of the HJN specification test of the three-FF-factor model (blue), p-values of the HJN specification test of the four-factor (three FF and momentum factors) model (red).}
	\label{fig:emp20year}
\end{sidewaysfigure} 

\begin{sidewaysfigure}[h!] 
	\centering
	\includegraphics[scale=0.5]{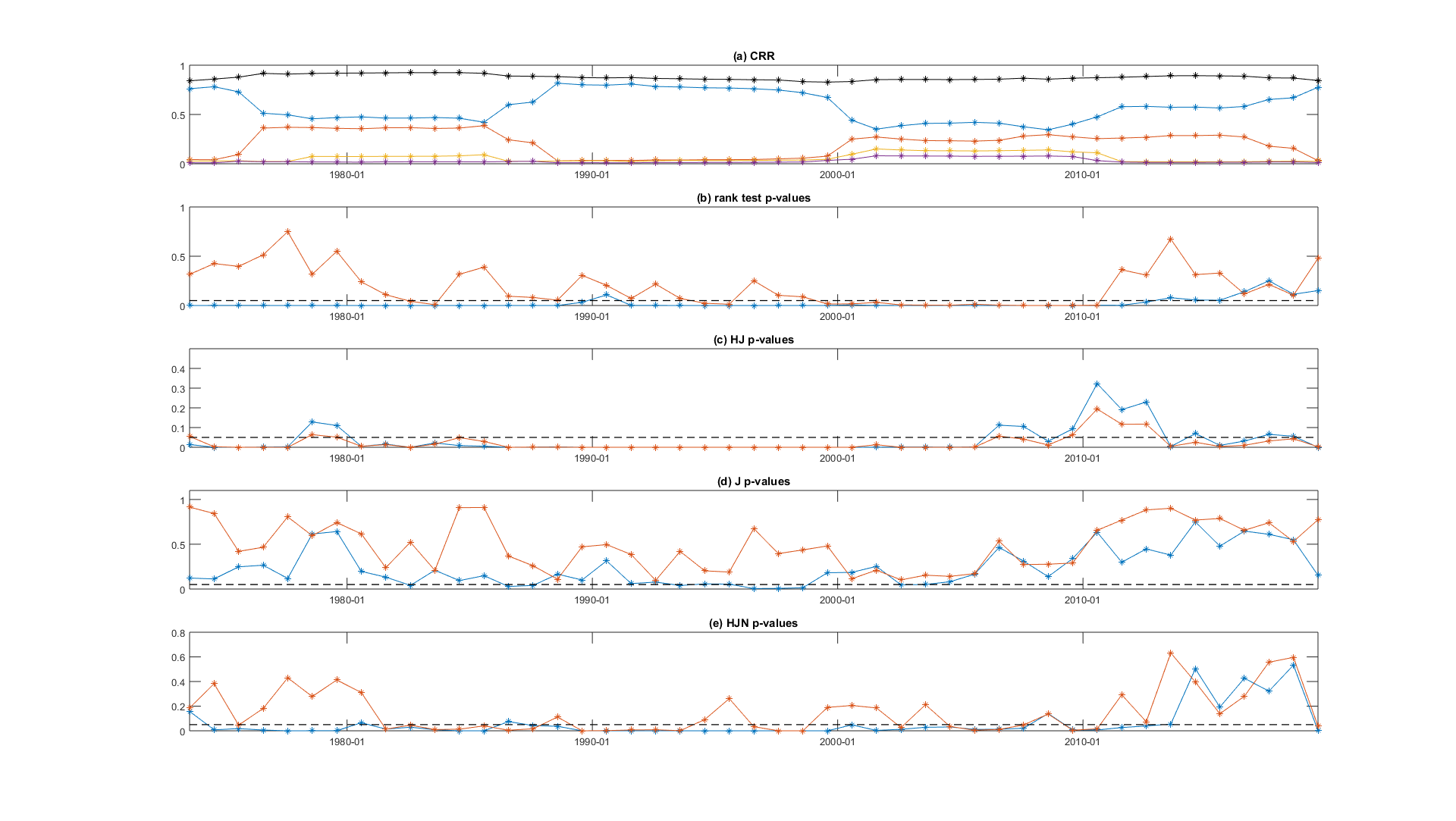}
	\caption{\scriptsize The time series of estimates from 1963-09 to 2019-08 with rolling windows of size T = 240 (20 years) and  the number of increments between successive rolling windows being 12 (1 year)  (x-axis is labeled with the ending period of each rolling window).  (a) The fraction of the total variation of the portfolio returns that is explained by their four largest principal components respectively (CRRs) ($ \lambda_i/(\sum_{j=1}^N \lambda_j), i=1,\cdots,4$); the black curve is the fraction of the total variation of the portfolio returns that is explained by the sum of the four largest principal components (sum of the first four CRRs, $ \sum_{i=1}^4\lambda_i/(\sum_{j=1}^N \lambda_j)$); (b) p-values of the rank test (\cite{kleibergen2006generalized}) of $q_{G}$ using the estimator $q_{G,T}$ with $G$ constructed by the three FF factors (blue), and the four factors (three FF and momentum factors) (red) respectively; (c) p-values of the HJ specification test of the three-FF-factor model (blue), p-values of the HJ specification test of the four-factor (three FF and momentum factors) model (red); (d) p-values of the $\mathcal{J}$ specification test of the three-FF-factor model (blue), p-values of the HJN specification test of the four-factor (three FF and momentum factors) model (red); (e) p-values of the HJN specification test of the three-FF-factor model (blue), p-values of the HJN specification test of the four-factor (three FF and momentum factors) model (red).}
	\label{fig:emp10year}
\end{sidewaysfigure}

%
\clearpage
\newpage 
\section{Proofs related to sections \ref{sec: HJ} and \ref{sec HJS} } \label{apend}
Appendices \ref{apend} and \ref{a2} use $\lambda_j(A)$ to denote the $j$th largest eigenvalue of a given matrix A, $\lambda_{\min}(A), \lambda_{\max}(A)$ the minimum and the maximum eigenvalues. With $A=(a_{ij})$, multiple matrix norms are denoted as $\norm{A}=\sqrt{\lambda_{\max}(A'A)},  \norm{A}_{1}= \max_j \sum_i |a_{ij}|, \norm{A}_{\infty}= \max_i \sum_j |a_{ij}|,  \norm{A}_{F}= \sqrt{tr(A'A)},  \norm{A}_{\max}= \max_{i,j} |a_{ij}|$. \\

\begin{assumption}\label{assum:CLT for the proxy factors}
	The following asymptotic distributions hold jointly: $\left(  \xi_{g,T}, \xi_{gg,T},  
	\xi_{v G, T},
	\xi_{z G, T},
	\xi_{e G, T},  
	\right) \rightarrow_d 	    \left( 
	\xi_g, \xi_V, 
	\xi_{v G},
	\xi_{z G},
	\xi_{e G}  \right)$, where 
	$\xi_{g,T}= \frac{1}{\sqrt{T}}\left(\sum_{t=1}^T g_t - \mu_g \right), $
	$\xi_{gg,T}= \frac{1}{\sqrt{T}}\sum_{t=1}^T \bar{g}_t \bar{g}_t', $ $\xi_{vG,T}= \frac{1}{\sqrt{T}}\sum_{t=1}^T v_tG_t', $$\xi_{zG,T}= \frac{1}{\sqrt{T}}\sum_{t=1}^T z_tG_t', $ $\xi_{eG,T}= \frac{1}{\sqrt{T}}\sum_{t=1}^T e_tG_t', $  and  
	$\xi_g,$ $\text{vec}(\xi_{vG}), \text{vec}(\xi_{zG}),$ $\text{vec}(\xi_{eG})$ are zero-mean normal random vectors, $\xi_V$ is a Gaussian random matrix.  
\end{assumption}  
Assumption \ref{assum:CLT for the proxy factors} is a central limit theorem for the different components in  equation (\ref{1}) interacted with a constant and the proxy
factors. Some of the statements such as $\xi_{g,T} \rightarrow_d \xi_g$ would hold if proxy factors are stationary with finite fourth moments and satisfy some strong mixing conditions (see e.g. \cite{peligrad2006central}). We specify a relatively strong assumption here instead of dealing with heavy technical details, but our results can be extended to general cases.    \\

\noindent \begin{proof}[Proof of Lemma \ref{lem:B_g}] 
	\label{Proof of Lemma  {lem:B_g dist}} 
	Assumption \ref{ass} implies $r_t=c    + \beta_g  \bar{g}_t + \beta_g \left( \bar{g} -\mu_g \right)  + \beta v_t + u_t$ then
	\begin{align*}
		\widehat{B}_g  = & \sum_{i=1}^T r_iG_i' \left(\sum_{t=1}^T G_tG_t' \right)^{-1} 
		= \sum_{i=1}^T \left(  c    + \beta_g  \bar{g}_t + \beta_g \left( \bar{g} -\mu_g \right)  + \beta v_t + u_t \right)  G_i'  \left(\sum_{t=1}^T G_tG_t'  \right)^{-1} \nonumber \\
		=& \left(c  , \beta_g  \right) +\frac{1}{\sqrt{T}}\left(\beta_g  \xi_{g, T}  , 0\right)  +\frac{1}{\sqrt{T}} \left(  \beta,\gamma  \right) \xi_{vz G, T} \widehat{Q}_G^{-1}       +  \frac{1}{\sqrt{T}} \xi_{eG,T} \widehat{Q}_G^{-1}
	\end{align*}  
	The conclusion is then a direct result from Assumption \ref{assum:CLT for the proxy factors}.	
\end{proof}

\noindent \begin{proof}[Proof of Theorem \ref{theo: charaterization of the delta}] \label{Proof of {theo:fixed N asymptotic for theta}}
	Assumption \ref{ass} and \ref{assum:CLT for the proxy factors} imply that  \begin{align}
		& \widehat{Q}_r = Q_r + O_p(1/\sqrt{T})  \label{eq:a}\\
		& \iota_N ={B}_g Q_G \theta_G \label{eq:b}
	\end{align} 
	Rewrite $T\widehat{\delta}_g^2$ as 
	\begin{align}
		\sqrt{T} \left(\iota_N - \widehat{B}_g Q_G \theta_G  \right) ' W  \sqrt{T}\left(\iota_N - \widehat{B}_g Q_G \theta_G  \right) \label{eq:13}
	\end{align}
	where \begin{align*}
		&W =(\widehat{Q}_r^{-1}-\widehat{Q}_r^{-1} \widehat{B}_g Q_{B_g,T} (Q_{B_g,T} \widehat{B}_g'\widehat{Q}_r^{-1}\widehat{B}_gQ_{B_g,T} )^{-1} Q_{B_g,T} \widehat{B}_g'\widehat{Q}_r^{-1} ) \\
		&Q_{B_g,T} = \text{diag}\left(I_{1+K_{g,1}}, \sqrt{T}I_{K_{g,2}} \right)
	\end{align*} 
	Lemma \ref{lem:B_g}	and equation (\ref{eq:b}) imply that \begin{align}
		\sqrt{T} \left(\iota_N - \widehat{B}_g Q_G \theta_G  \right)\rightarrow_d \widetilde{\psi}_{B_g}  \label{eq:11}
	\end{align}
	Lemma \ref{lem:B_g} and Assumption \ref{assum:factor loading strength}
	imply the following equation holds jointly with equation (\ref{eq:11}) 
	\begin{align}
		\widehat{B}_g Q_{B_g,T} \rightarrow_d \eta_{B_g} + (0\vdots \psi_{\beta_{g,2}} ) \label{eq:12}
	\end{align}
	Plug equations (\ref{eq:a})(\ref{eq:11})(\ref{eq:12}) in equation (\ref{eq:13}), then we would derive the conclusion. 
\end{proof}

\begin{lem}\label{lem:ass hold}
	Suppose Assumption \ref{assum:CLT for the proxy factors} and \ref{assum: factor structure in r}-\ref{assum:factor loading strength} hold, let N,T increase and then the restrictions on $e_t$ from
	Assumption \ref{assumption: weak cross-sectional correlation of the idiosyncratic component } (i)(ii)(iii) also hold for $\widetilde{e}_t$ with $\widetilde{e}_t = Q_r^{-\frac{1}{2}}e_t$ with $Q_r^{-\frac{1}{2}} = A\Lambda^{-\frac{1}{2}} A' $ such that $Q_r= A\Lambda A'$ with $A'A=I_N$ and we assume $Q_r^{-\frac{1}{2}}$ is a row diagonally-dominant matrix\footnote{In the proof of this lemma we need that $Q_r^{-\frac{1}{2}}$ has bounded absolute row sum.  This is not a wild assumption if we consider the Gershgoring-type eigenvalue inclusion theorem and all eigenvalues of $Q_r^{-1/2}$ are bounded. }.\\

	\noindent \begin{proof}[Proof of Lemma \ref{lem:ass hold}] ~\\
		$\widetilde{e}_t,t=1,\cdots,T$ are i.i.d. mean zero random vectors
		with finite fourth moments by construction. Next we show $\sup_i \mathbb{E}\widetilde{e}_{it}^4 $ is bounded. Assumption \ref{assum: factor structure in r} implies $Q_r = c'c + \beta V_f \beta' + \gamma V_z \gamma' + \Omega_e$ and thus we have the following results by eigenvalue inequalities (see e.g. \text{7.3.P16} of  \cite{horn1990matrix}):
		\begin{align}
			\lambda_{\max}\left(Q_r \right)=\lambda_{\max}\left(c'c + \beta V_f \beta' + \gamma V_z \gamma' + \Omega_e \right)\leq  \lambda_{\max}\left(c'c + \beta V_f \beta' + \gamma V_z \gamma'  \right) +L \\
			\lambda_{\min}\left(Q_r \right)=\lambda_{\min}\left(c'c + \beta V_f \beta' + \gamma V_z \gamma' + \Omega_e \right)\geq \lambda_{\min}\left(c'c + \beta V_f \beta' + \gamma V_z \gamma'\right)  +l =l 
		\end{align}
		Then by the assumption that $Q_r^{-\frac{1}{2}}$ is a row diagonally-dominant matrix we know any row sums of $Q_r^{-\frac{1}{2}}$ would be upper bounded by $2l^{-1/2}$ and thus   $\sup_i \mathbb{E}{e}_{it}^4\leq L $ implies that $\sup_i \mathbb{E}\widetilde{e}_{it}^4 $ is bounded. Therefore, Assumption \ref{assumption: weak cross-sectional correlation of the idiosyncratic component } (i) holds. 
		
		The term $c'c + \beta V_f \beta' + \gamma V_z \gamma' $ in $Q_r$ is a positive semi-definite matrix and we can rewrite that term as $A_{Q_r}\Lambda_{Q_r}  A_{Q_r}'$ such that $\Lambda_{Q_r}$ is a diagonal matrix containing all positive eigenvalues of $c'c + \beta V_f \beta' + \gamma V_z \gamma' $ and $A_{Q_r}$ are the corresponding eigenvectors. Therefore, $Q_r = A_{Q_r}\Lambda_{Q_r}  A_{Q_r}' + \Omega_e$ and thus \begin{align}
			Q_r^{-1} = \Omega_e^{-1} - \Omega_e^{-1} A_{Q_r}\left(\Lambda_{Q_r}^{-1}  + A_{Q_r}'\Omega_e^{-1} A_{Q_r}  \right)^{-1}A_{Q_r}' \Omega_e^{-1}
		\end{align}
		which then implies that  
		$$ \text{tr}\left(\mathbb{E}Q_r^{-1/2}e_te_t'Q_r^{-1/2}\right)/N=    \text{tr} \left( \Omega_e  Q_r^{-1}\right)/N  =    1-\text{tr} \left(  \left(\Lambda_{Q_r}^{-1}  + A_{Q_r}'\Omega_e^{-1} A_{Q_r}  \right)^{-1}\right)/N  $$ 
		From the fact that the eigenvalues of $\Lambda_{Q_r}$ explode by Assumption \ref{assum:factor loading strength} and the eigenvalues of $\Omega_e$ are bounded by  Assumption \ref{assumption: weak cross-sectional correlation of the idiosyncratic component } ,  Courant-Fischer minimax principle implies $\text{tr} \left(  \left(\Lambda_{Q_r}^{-1}  + A_{Q_r}'\Omega_e^{-1} A_{Q_r}  \right)^{-1}\right)/N \rightarrow 0$. Furthermore, 
		Assumption \ref{assumption: weak cross-sectional correlation of the idiosyncratic component } and  \ref{assum:factor loading strength}  imply that $\lim\inf_{N,T}  \lambda_{\min} \left( Q_r^{-1/2}\right)$ and $\lim\sup_{N,T}\lambda_{\min} \left( Q_r^{-1/2} \right)$ are bounded, and thus we also have $$ 0<l<\lim\inf_{N,T}  \lambda_{\min} \left(  Q_r^{-1/2} \Omega_e Q_r^{-1/2} \right) <\lim\sup_{N,T} \lambda_{\max}\left( Q_r^{-1/2}  \Omega_e \right) Q_r^{-1/2}  < L < \infty.$$ 
		Therefore, Assumption \ref{assumption: weak cross-sectional correlation of the idiosyncratic component } (ii) holds for $\widetilde{e}_{t}$. As for Assumption \ref{assumption: weak cross-sectional correlation of the idiosyncratic component } (iii), inequalities  for the trace of matrix product (e.g. \cite{fang1994inequalities}) suggest  
		\begin{align*}
			& \lambda_{\min}\left(Q_r^{-1} \right)\text{tr}\left(e_te_t'\right) - \lambda_{\max}\left(Q_r^{-1} \right)\text{tr}\left( \mathbb{E}\left(e_te_t'\right)  \right)  \\  \leq &\text{tr}\left(Q_r^{-1}\left(e_te_t'- \mathbb{E}\left(e_te_t'\right)  \right)  \right)  \leq 
			\lambda_{\max}\left(Q_r^{-1} \right)\text{tr}\left(e_te_t'\right) - \lambda_{\min}\left(Q_r^{-1} \right)\text{tr}\left( \mathbb{E}\left(e_te_t'\right)  \right) 
		\end{align*}
		which implies that $\mathbb{E}\left| \frac{1}{\sqrt{N}} \sum_{i=1}^{N}  \left(\widetilde{e}_{it}\widetilde{e}_{is}-  \mathbb{E}\widetilde{e}_{it}\widetilde{e}_{is} \right) \right|^4 < L < \infty$ given all eigenvalues of $Q_r^{-1}$ are bounded. 
	\end{proof}	  	
\end{lem}

Theorem \ref{theo:asymptotic for theta} illustrates how the weak (proxy) factors can affect the asymptotic properties of the estimator $\widehat{\theta}_G$. This theorem resembles \textit{Theorem one} from \cite{anatolyev2018factor}, and indeed we observe that the asymptotic behavior of the estimator $\widehat{\theta}_G$ is similar to the one of the two-pass FM risk premia analyzed in \cite{anatolyev2018factor}.  
\begin{theo} \label{theo:asymptotic for theta}
	~~\\ 	Case (1):
	Suppose Assumption \ref{assum:CLT for the proxy factors} and \ref{assum: factor structure in r}-\ref{assum:factor loading strength} hold, N is fixed and T increases to infinity: 
	\begin{align}
		&\sqrt{T} Q_{B_g, T}^{-1}	 \left(  
		\widehat{Q}_G (\widehat{\theta}_{G} - {\theta}_{G})  - \textit{bias}_e- \textit{bias}_m  
		\right) = O_p(1) \nonumber 
	\end{align} 	
	where  
	\begin{align*}
		&	\textit{bias}_e  =- \left( \widehat{B}_g' \widehat{Q}_r^{-1} \widehat{B}_g\right)^{-1}  \widehat{B}_{g,e}' \widehat{Q}_r^{-1} \widehat{B}_{g,e} \widehat{Q}_G {\theta}_{G}  \nonumber \\
		&	\textit{bias}_m  = -\left( \widehat{B}_g' \widehat{Q}_r^{-1} \widehat{B}_g\right)^{-1}  \left({B}_g+\widehat{B}_{g,m}\right)' \widehat{Q}_r^{-1} \widehat{B}_{g,m} \widehat{Q}_G {\theta}_{G} \\
		& \widehat{B}_{g,e} =\frac{1}{\sqrt{T}} \xi_{eG,T} \widehat{Q}_G^{-1};~  \widehat{B}_{g,m} =  \frac{1}{\sqrt{T}} \left(  \beta,\gamma  \right) \xi_{vz G, T} \widehat{Q}_G^{-1}; ~
		Q_{B_g, T}  = \textit{diag}\left(I_{1+K_{g,1}}, \sqrt{T}I_{K_{g,2}}\right)  
	\end{align*}	
	and   
	\begin{align*}
		&\sqrt{T} Q_{B_g, T}^{-1}	\textit{bias}_e =    
		\left( \begin{matrix}
			\sqrt{T}	\textit{bias}_{e,1} \\ 	\textit{bias}_{e,2}
		\end{matrix} \right) 
		\rightarrow_d W Q_{B_g}  \widetilde{e}'Q_r^{-1}  \widetilde{e} Q_G \theta_G     \\ 
		&\sqrt{T} Q_{B_g, T}^{-1}	\textit{bias}_m =    
		\left( \begin{matrix}
			\sqrt{T}	\textit{bias}_{m,1} \\ 	\textit{bias}_{m,2}
		\end{matrix} \right)  
		\rightarrow_d W \left(\left( \begin{matrix}
			I_{1+K_g} \\ 0
		\end{matrix} \right) + \widetilde{\xi }Q_{B_g} \right)' \eta'Q_r^{-1} \eta \widetilde{\xi} Q_G \theta_G\\ 
		&W= - \left(\left( \eta \left(\left( \begin{matrix}
			I_{1+K_g} \\ 0
		\end{matrix} \right) + \widetilde{\xi }Q_{B_g} \right) +  \widetilde{e}  Q_{B_g} \right)'Q_r^{-1} \left(\eta \left(\left( \begin{matrix}
			I_{1+K_g} \\ 0
		\end{matrix} \right) + \widetilde{\xi }Q_{B_g} \right) +  \widetilde{e}  Q_{B_g} \right)\right)^{-1} \end{align*} 
	Case (2):
	Suppose Assumption \ref{assum:CLT for the proxy factors} and \ref{assum: factor structure in r}-\ref{assum:factor loading strength} hold, let N, T increases to infinity  $N/T\rightarrow c$ and $Q_r$ is a known row diagonally-dominant matrix such that $\eta'Q_r^{-1}\eta /N \rightarrow \widetilde{D}$:	  
	\begin{align}
		&\sqrt{NT} Q_{B_g, T}^{-1}	 \left(  
		\widehat{Q}_G (\widehat{\theta}_{G} - {\theta}_{G,T})  - \textit{bias}_e- \textit{bias}_m   
		\right) = O_p(1) \nonumber 
	\end{align}
	where
	\begin{align*}
		& {\theta}_{G,T} =\theta_G + \widehat{Q}_G^{-1}\left( \widehat{B}_g' {Q}_r^{-1} \widehat{B}_g\right)^{-1}  \widehat{B}_g'  {Q}_r^{-1}\left( q_G\left( \widehat{Q}_G^{-1} 	-{Q}_G^{-1} \right)   -  \frac{1}{\sqrt{T}}\left( \beta_g  \xi_{g, T}  , 0\right)  \right) \widehat{Q}_G {\theta}_{G}
		\\&\sqrt{T} Q_{B_g, T}^{-1}\widehat{Q}_G	\left(      {\theta}_{G,T} -\theta_G    \right) =O_p(1)  \\
		&\sqrt{T} Q_{B_g, T}^{-1}	\textit{bias}_e =    
		\left( \begin{matrix}
			\sqrt{T}	\textit{bias}_{e,1} \\ 	\textit{bias}_{e,2}
		\end{matrix} \right) 
		\rightarrow_d W Q_{B_g} \widetilde{\Sigma}_{{e}} Q_G \theta_G     \\ 
		&\sqrt{T} Q_{B_g, T}^{-1}	\textit{bias}_m =    
		\left( \begin{matrix}
			\sqrt{T}	\textit{bias}_{m,1} \\ 	\textit{bias}_{m,2}
		\end{matrix} \right)  
		\rightarrow_d W \left(\left( \begin{matrix}
			I_{1+K_g} \\ 0
		\end{matrix} \right) + \widetilde{\xi }Q_{B_g} \right)' \widetilde{D} \widetilde{\xi} Q_G \theta_G\\ 
		&W= -  \left(\left(\left( \begin{matrix}
			I_{1+K_g} \\ 0
		\end{matrix} \right) + \widetilde{\xi }Q_{B_g} \right)'\widetilde{D} \left(\left( \begin{matrix}
			I_{1+K_g} \\ 0
		\end{matrix} \right) + \widetilde{\xi }Q_{B_g} \right) +   Q_{B_g}\widetilde{\Sigma}_{{e}}   Q_{B_g}\right)^{-1} \end{align*}
	\noindent \begin{proof}[Proof of Theorem \ref{theo:asymptotic for theta}]
		Case (1):  
		After simple algebra, we can express the term $\widehat{Q}_G  	\left(\widehat{\theta}_G -\theta\right)  - \textit{bias}_e- \textit{bias}_m$ in the following way
		\begin{align}
			&\widehat{Q}_G  	\left(\widehat{\theta}_G -\theta\right)  - \textit{bias}_e- \textit{bias}_m  
			=\left( \widehat{B}_g' \widehat{Q}_r^{-1} \widehat{B}_g\right)^{-1}  \left\{ \widehat{B}_g' \widehat{Q}_r^{-1}\left( q_G\left( \widehat{Q}_G^{-1} 	-{Q}_G^{-1} \right)   -  \frac{1}{\sqrt{T}}\left( \beta_g  \xi_{g, T}  , 0\right)  \right)  \right. \nonumber \\
			& \left. - \left( \frac{1}{\sqrt{T}}\left( \beta_g  \xi_{g, T}  , 0\right)+ \widehat{B}_{g,e}\right)' \widehat{Q}_r^{-1}   \widehat{B}_{g,m}    -    \left(B_g + \frac{1}{\sqrt{T}}\left( \beta_g  \xi_{g, T}  , 0\right)+ \widehat{B}_{g,m} \right)' \widehat{Q}_r^{-1}  \widehat{B}_{g,e}  \right\}  \widehat{Q}_G {\theta}_{G} 
		\end{align}
		From the proofs of Lemma \ref{lem:B_g} and Theorem \ref{theo: charaterization of the delta} we have     
		\begin{align}
			&	\widehat{B}_g = B_g +\frac{1}{\sqrt{T}}\left( \beta_g  \xi_{g, T}  , 0\right) + \widehat{B}_{g,m} + \widehat{B}_{g,e}                      \\
			&	\widehat{B}_g Q_{B_g,T}  = \left(c+O_p(1/\sqrt{T}), \eta_{\beta_g}\right) + \left(\frac{1}{\sqrt{T}}\widetilde{e}_T  +\frac{1}{\sqrt{T}} \eta\widetilde{\xi }_T  \right)Q_{B_g,T}   \label{eq:}
		\end{align}  where
		$
		\widetilde{\xi }_T =  \sqrt{T} \widehat{B}_{g,m}= \left(\begin{matrix}
		0 &0 \\ \left(d_g Q_{\beta_g} \right)^{-1} & 0 \\0 & I_{K_z}
		\end{matrix}\right)  \xi_{vz G, T} \widehat{Q}_G^{-1}, ~~\widetilde{e}_T =  \sqrt{T} \widehat{B}_{g,e}=  \xi_{eG,T} \widehat{Q}_G^{-1}, 
		$$ \widetilde{\xi }_T \rightarrow_d \widetilde{\xi}, \widetilde{e }_T \rightarrow_d \widetilde{e}, Q_{B_g,T}/ {\sqrt{T}}\rightarrow  Q_{B_g} =\text{diag}\left( 0_{1+K_{g,1}}, I_{K_{g,2}} \right)$ and $\widetilde{\xi} ,\widetilde{e} $ are Gaussian random matrices.
		Equation (\ref{eq:}) implies that
		\begin{align}
			\widehat{B}_g Q_{B_g,T} 
			\rightarrow_d &~~ \eta \left(\left( \begin{matrix}
				I_{1+K_g} \\ 0
			\end{matrix} \right) + \widetilde{\xi }Q_{B_g}\right) +  \widetilde{e}  Q_{B_g}
		\end{align}
		With the above intermediate results, we prove our statement for the term $\textit{bias}_{e}$ and the rests follow the same steps. Assumption \ref{assum:CLT for the proxy factors} implies that the above asymptotic results hold jointly and thus if we plug these into the equation below we would derive the asymptotic distribution of the term $\textit{bias}_{e}$
		\begin{align}
			\sqrt{T} Q_{B_g, T}^{-1} 	\textit{bias}_e  =- \left( (\widehat{B}_gQ_{B_g, T})' \widehat{Q}_r^{-1} (\widehat{B}_gQ_{B_g, T})\right)^{-1}  (\widehat{B}_{g,e}Q_{B_g, T})' \widehat{Q}_r^{-1}(\sqrt{T} \widehat{B}_{g,e}) \widehat{Q}_G {\theta}_{G} 
		\end{align}
		%
		
		%
		
		%

		\noindent
		Case (2): 
		Next we discuss the case where N,T both grows to infinity. 
		We first show the following two results:
		\begin{align}
			&(2.1)~~ 	T  \widetilde{e}_T' Q_r^{-1}\widetilde{e}_T /N  \rightarrow_p \tilde{\Sigma}_e \label{eq:2.1}\\ 
			&	(2.2) ~~\sqrt{T} \widetilde{e}_T' Q_r^{-1}  \eta     /N \rightarrow_p 0 \label{eq:2.2}
		\end{align}

		To prove statement  (2.1),	let $\widetilde{e}_t = Q_r^{-\frac{1}{2}}e_t$ and $\rho(s,t)=\frac{1}{\sqrt{N}} \sum_{i=1}^N \widetilde{e}_{it}' \widetilde{e}_{is}$. By construction, $\mathbb{E}\rho^2(s,t)=\frac{1}{N}\text{tr}\left(\Omega_{\widetilde{e}}\Omega_{\widetilde{e}} \right)\leq \lambda_{\max}^2(\Omega_{\widetilde{e}})$ with $\Omega_{\widetilde{e}} =\mathbb{E}\widetilde{e}_t\widetilde{e}_t' $ and thus Lemma \ref{lem:ass hold} implies that $\mathbb{E}\rho^2(s,t)$ is bounded. Assumption \ref{ass} (finite fourth moments of proxy factors), Assumption \ref{assumption: weak cross-sectional correlation of the idiosyncratic component } (i) and 
		the bounded $\mathbb{E}\rho^2(s,t)$ deliver the following inequality
		\begin{align}
			& \mathbb{E}	\left(\frac{1}{T\sqrt{N}} \sum_{i=1}^N \left(\sum_{t=1}^T \sum_{s\neq t}^T G_{tm}G_{sn} \widetilde{e}_{it}' \widetilde{e}_{is}  \right)\right)^2  \nonumber\\ 
			= &\mathbb{E}	\left(\frac{1}{T}   \left(\sum_{t=1}^T \sum_{s\neq t}^T G_{tm}G_{sn}\rho(s,t) \right)\right)\left(\frac{1}{T} \left(\sum_{t'=1}^T \sum_{s' \neq t}^T G_{t'm}G_{s'n}\rho(s',t') \right)\right) \nonumber \\
			= & \mathbb{E}	\left(\frac{1}{T^2}   \left(\sum_{t=1}^T \sum_{s\neq t}^T \left(G^2_{tm}G^2_{sn} +G_{tm}G_{sn}G_{sm}G_{tn}   \right)\rho^2(s,t) \right)\right) <L  \label{eq{assum: GG gamma}}
		\end{align} 
		which with Chebyshev's inequality gives   	
		\begin{align}
			\frac{1}{T\sqrt{N}} \sum_{i=1}^N \left(\sum_{t=1}^T \sum_{s\neq t}^T G_{t}G_{s}' \widetilde{e}_{it}' \widetilde{e}_{is}  \right) = O_P(1) \label{eq:ee1}
		\end{align} 	
		Finally we arrive at (2.1):  
		\begin{align}
			T  \widetilde{e}_T' Q_r \widetilde{e}_T /N 
			& = \widehat{Q}_g^{-1}\left(  \frac{1}{ N{T}} \sum_{t=1}^TG_t  \widetilde{e}_t' \widetilde{e}_t G_t'\right)   \widehat{Q}_g^{-1}   + O_p(1/\sqrt{N}) = \tilde{\Sigma}_e   + O_p(1/\sqrt{N}) 
		\end{align}
		where the first equality is due to equation (\ref{eq:ee1}) and the last equality is guaranteed by Lemma \ref{lem:ass hold}.

		Next, we prove the statement (2.2). We first look at the second moments  
		\begin{align}
			& 	\mathbb{E}\left( \left\|   \frac{1}{\sqrt{T}} \sum_{t=1}^T  G_t e_t'   Q_r^{-1} \eta      /\sqrt{N} \right\|^2 \right)  \leq  
			\mathbb{E}\left( \text{tr}\left(   \frac{1}{\sqrt{T}} \sum_{t=1}^T  G_t \widetilde{e}_t'   Q_r^{-1/2} \eta      /\sqrt{N} \right)\left(   \frac{1}{\sqrt{T}} \sum_{t=1}^T  G_t \widetilde{e}_t'   Q_r^{-1/2} \eta      /\sqrt{N} \right)' \right) \nonumber \\
			=&  \frac{1}{{T}} \sum_{t=1}^T  \text{tr} \left( \mathbb{E}\left(  G_t \widetilde{e}_t'   Q_r^{-1/2} \eta        \eta'     Q_r^{-1/2} \widetilde{e}_tG_t'      \right)/{N} \right)\leq L \text{tr}\left(\Omega_{\widetilde{e}}    Q_r^{-1/2} \eta        \eta'     Q_r^{-1/2}  /{N} \right)  
			\leq 	   L \lambda_{\max}\left(\Omega_{\widetilde{e}} \right)  \norm{Q_r^{-1/2}}^2_1 (\norm{\eta}^2_F/N)  \nonumber 
		\end{align}
		where the first inequality is because factors $g_t$ have finite fourth moments. The proof of Lemma \ref{lem:ass hold} shows that $\lambda_{\max}\left(\Omega_{\widetilde{e}} \right)$ and $\norm{Q_r^{-1/2}}^2_1$ are bounded, Assumption \ref{assum:factor loading strength 2} implies that $\norm{\eta}^2_F/N$ is bounded. Therefore, we know 
		\begin{align}
			\mathbb{E}\left( \left\|   \frac{1}{\sqrt{T}} \sum_{t=1}^T  G_t e_t'   Q_r^{-1} \eta /\sqrt{N} \right\|^2 \right)  \leq  L <\infty 
		\end{align}   
		and thus (2.2) holds since 	
		$ 	\sqrt{T} \widetilde{e}_T' Q_r^{-1}  \eta      /N  = \frac{1}{\sqrt{N}} {Q}_g^{-1} \left(\frac{1}{\sqrt{T}} \sum_{t=1}^T  G_t e_t'   Q_r^{-1} \eta      /\sqrt{N}\right)= O_p(1/\sqrt{N}) $. 
		%
		In the end, $\eta' Q_r^{-1}\eta/N \rightarrow \widetilde{D}$ and the result (2.2)  imply the following term  is of order $O_p(1)$ when $N/T\rightarrow c$: 
		\begin{align*}
			&\frac{\sqrt{T}}{\sqrt{N}} {Q}_{B_g,T} \left\{ -
			\left( \frac{1}{\sqrt{T}}\left( \beta_g  \xi_{g, T}  , 0\right)+ \frac{1}{\sqrt{T}}\widetilde{e}_T\right)'  {Q}_r^{-1}   \frac{1}{\sqrt{T}}\widetilde{\xi}_T     -    \left(B_g + \frac{1}{\sqrt{T}}\left( \beta_g  \xi_{g, T}  , 0\right)+ \frac{1}{\sqrt{T}}\widetilde{\xi}_T  \right)' {Q}_r^{-1}  \frac{1}{\sqrt{T}}\widetilde{e}_T  \right\} = O_p(1)
		\end{align*}
	\end{proof}
\end{theo}
\begin{assumption}\label{assum:2.2.2}
	Let $e_{g,t}({\theta}_G)= \iota_N -r_tG_t' {\theta}_G $, and the restrictions on $e_t$ from
	Assumption \ref{assumption: weak cross-sectional correlation of the idiosyncratic component } (i)(ii)(iii) also hold for $e_{g,t}({\theta}_G)$. 
\end{assumption}

\noindent \begin{proof}[Proof of Corollary \ref{theo: strong beta asymptotics}]
	\label{Proof of Corollary {theo: strong beta asymptotics}}

	Theorem \ref{theo: charaterization of the delta} suggest for given $N$, $	 	 T	\widehat{\delta}^2_g \rightarrow_d d_{1,N}$ with $d_{1,N}= \widetilde{\psi}_{B_g} ' M_{Q_r^{-\frac{1}{2}}  \left( \eta _{B_g} + (0\vdots \psi_{\beta_{g,2}} )  \right) } \widetilde{\psi}_{B_g},$  $ \widetilde{\psi}_{B_g}\sim N(0,S_{B_g})$ and $S_{B_g}= p\lim S_{B_g,T}, S_{B_g,T}=  \frac{1}{T}\sum_{t=1}^T e_{g,t}({\theta}_G) e_{g,t}({\theta}_G)'$. From the construction of the $\widehat{c}_{1-\alpha}$, we know it is drawn from the distribution  of the $d_{2,N} = {\psi}_{S} ' M_{Q_r^{-\frac{1}{2}}  \left( \eta _{B_g} + (0\vdots \psi_{\beta_{g,2}} )  \right) } {\psi}_S$ with $\psi_S \sim N(0,S), S = p\lim \widehat{S}, \widehat{S}= \frac{1}{T}\sum_{t=1}^T {e}_{g,t}(\widehat{\theta}_G) {e}_{g,t}(\widehat{\theta}_G)', \widehat{e}_{g,t}= \iota_N -r_tG_t'\widehat{\theta}_G$ and $\psi_S$ independent from $\psi_{B_g}$.   The proof of Theorem \ref{theo:asymptotic for theta} suggest that $\widetilde{\psi}_{B_g} ' P_{Q_r^{-\frac{1}{2}}  \left( \eta _{B_g} + (0\vdots \psi_{\beta_{g,2}} )  \right) } \widetilde{\psi}_{B_g}=O_p(1)$ for any given N, and the same for the term ${\psi}_{S} ' P_{Q_r^{-\frac{1}{2}}  \left( \eta _{B_g} + (0\vdots \psi_{\beta_{g,2}} )  \right) } {\psi}_S$.  Now we look at the difference $S_{B_g,T}-S_{T}$.
	\begin{align}
		S_T = &\frac{1}{T}\sum_{t=1}^T \left(\iota_N -r_tG_t'\left(\widehat{\theta}_G  - \theta_G + \theta_G  \right)    \right)\left(\iota_N -r_tG_t'\left(\widehat{\theta}_G  - \theta_G + \theta_G  \right)    \right)' \nonumber \\
		=&S_{B_g,T}   + \frac{1}{T}\sum_{t=1}^T\left( r_tG_t'\left(\widehat{\theta}_G  - \theta_ G \right)    \right)\left( r_tG_t'\left(\widehat{\theta}_G  - \theta_ G \right)    \right)' \nonumber \\
		&- \frac{1}{T}\sum_{t=1}^T\left( r_tG_t'\left(\widehat{\theta}_G  - \theta_ G \right)    \right) e_{g,t}' - \frac{1}{T}\sum_{t=1}^Te_{g,t}\left( r_tG_t'\left(\widehat{\theta}_G  - \theta_ G \right)    \right)'  \label{eq:S_T}
	\end{align}
	Assumption \ref{assum:2.2.2} and proofs of Theorem \ref{theo:asymptotic for theta} then suggest the last two terms be negligible in large samples, and thus for fixed $N$ when $T$ is large $S_T \approx S_{B_g,T}   + \frac{1}{T}\sum_{t=1}^T\left( r_tG_t'\left(\widehat{\theta}_G  - \theta_ G \right)    \right)\left( r_tG_t'\left(\widehat{\theta}_G  - \theta_ G \right)    \right)' $. This then lead to the conclusion.      
\end{proof}

\noindent \begin{proof}[Proof of Lemma \ref{lem: }]\label{Proof of Lemma reflem: }
	Assumption \ref{assum:2.2.2} implies that for $\theta_{G}$, $AR(\theta_G) \rightarrow_d \chi^2(N)$ which then implies the result. 
\end{proof} 	   

\noindent \begin{proof}[Proof of Theorem \ref{theo:2}] \label{proof of theorem theo:2}
	Notice if $\theta_G \in CS_{r,\alpha_1}$ then
	$
	T\widehat{\delta}_g^* \leq T e_{g,T}(\theta_G)'{Q}_{r}^{-1} e_{g,T}(\theta_G) 
	$ and  $ c_{1-\alpha_2}(\theta_G)\leq  c_{1-\alpha}^* $.  Assumption \ref{assum:2.2.2} implies in large samples 
	\begin{align*}
		\lim\inf \limits_{T}	\mathbb{P}\left(T e_{g,T}(\theta_G)'{Q}_{r}^{-1} e_{g,T}(\theta_G) \leq c_{1-\alpha_2}(\theta_G) \right) = 1- \alpha_2
	\end{align*}
	The Lemma \ref{lem: } and $(1- \alpha_1)(1- \alpha_2)=1- \alpha$ lead to conclusion. 
\end{proof}

\noindent \begin{proof}[Proof of Theorem \ref{theo:power of HJS}] \label{Proof of Theorem theo:power of HJS} 
	We first provide the proof related with the $\mathcal{J}$ statistic, which essentially results from the proof of Theorem 2 in \cite{gospodinov2017spurious}. 

	Denote $
	W= TL'\widehat{B}_g 'P_1 \left(P_1'{\Sigma} P_1 \right)^{-1}P_1'\widehat{B}_{g} L$, where $P_1$ is an $N\times (N-1) $ orthogonal matrix whose columns are orthogonal to $\iota_N$ such that $P_1'P_1 = I_{N-1}, P_1P-1'= M_{\iota_N}$; L is an lower triangular matrix such that ${Q}_{\widetilde{G}} = LL'$ and ${\Sigma}$ is the covariance matrix. Define $Z= (P_1'\Sigma P_1 )^{-1/2}P_1' \widehat{B}_g L$ and $M= (P_1'\Sigma P_1 )^{-1/2}P_1' {B}_g L$, and then 
	\begin{align*}
		\sqrt{T}\text{vec}\left(Z-M \right) \rightarrow_d N(0,I_{(N-1)K}) 
	\end{align*}  
	From the assumption on $H$, we know there exists $K\times k$ and $K\times K-k$ matrices $C_1, C_2$  where $(C_1,C_2)$ is a $K\times K$
	orthogonal matrix and $\widetilde{M}_1=MC_1,   \widetilde{M}_2=MC_2$ are of orders $O({1}/{\sqrt{T}}), O(1)$ respectively. Let $\sqrt{T}\widetilde{M}_1\rightarrow \widetilde{\mu}$, then in case (1) $\widetilde{\mu}=0$ and in case (2) $\widetilde{\mu}$ is bounded. Let $\widetilde{Z}=(\widetilde{Z}_1, \widetilde{Z}_2)=({Z}C_1,{Z}C_2)$ we would have 
	\begin{align*}
		\sqrt{T} \left(\begin{matrix}
			\text{vec}\left(\widetilde{Z}_1 \right) \\\text{vec}\left(\widetilde{Z}_2 - \widetilde{M}_2 \right) 
		\end{matrix} \right) \rightarrow_d N\left(\left(\begin{matrix}
			\text{vec}\left(\widetilde{\mu} \right) \\ 0 
		\end{matrix} \right),   ~~ I_{(N-1)K} \right) 
	\end{align*}   
	The proof of theorem  2 in \cite{gospodinov2017spurious} shows that: (i) the asymptotic distribution of the $\mathcal{J}$ statistic is the same as the one of the largest eigenvalue, $w_k$, of $\widetilde{W}^{-1}$ with 
	\begin{align}
		\widetilde{W} = T \widetilde{Z}_1' M_{\widetilde{Z}_2}\widetilde{Z}_1, \label{eq:wishart}
	\end{align}
	and thus (ii) in case (1) where $H$ is of reduced rank and $\widetilde{\mu}=0$,  $ \widetilde{W} \rightarrow_d \mathcal{W}_k\left(N-K-1+k, I_r \right) $ and
	\begin{align*}
		\mathbb{P}\left(w_k \leq a \right) \leq 	\mathbb{P}\left(x_k \leq a \right), ~~x_k \sim \chi^2_{N-k}, 
	\end{align*}

	In case (2), where $\widetilde{\mu}\neq 0$, $\widetilde{W}$ follows a non-central Wishart distribution $W_k(N-K-1+k, I_r, \mu)$ with $\mu= \widetilde{\mu}'M_{\widetilde{M}_2}\widetilde{\mu},    \norm{\mu}\leq L<\infty$  asymptotically, which then implies the inconsistency of the $\mathcal{J}$ test. The consistency of the HJS test is obvious since  $\norm{\iota_N -q_{\widetilde{G}} \theta_{\widetilde{G}}}>a > 0, \forall \theta_{\widetilde{G}} \in \Theta $ implies that $\norm{\iota_N -q_{\widetilde{G},T} \theta_{\widetilde{G}}} = O_p(\sqrt{T}), \forall \theta_{\widetilde{G}} \in \Theta$.

	%

\end{proof}

\begin{example}\label{example}
	We use an specific example, where we suppose Assumptions   \ref{assum: factor structure in r}, - \ref{assum:factor loading strength} hold with $Q_{\widetilde{G}} = \mathbb{E}({G}_t{G}_t')$, $K\geq K_{g,2}\geq 1$, to show that  over the 
	supreme, $T\sup_{\theta \in 	CS_{r, \alpha_1}} \delta_{g,T}(\theta)$, 
	is not properly bounded by $c_{1-\alpha}^*$ in the sense that there would be $\alpha>0$ such that
	\begin{align*}
		\lim\inf_T			\mathbb{P}\left(\left(T\sup_{\theta \in 	CS_{r, \alpha_1}     } \delta_{g,T}(\theta) \right) \geq  c_{1-\alpha}^* \right) > \alpha. 
	\end{align*}	
	We prove this by discussing elements in the confidence set.  We group $\theta$s in $CS_{r, \alpha_1}$ into two classes: (1) $\theta$s with entries corresponding to strong proxy factors deviating from their true values; (2) $\theta$s with entries corresponding to weak proxy factors deviating from their true values. Notice for any $\theta$s belong to class (1),  $T\delta_{g,T}(\theta)$ is of order $O_p(T)$, while for $\theta$s belong to class (2) we have  $T\delta_{g,T}(\theta)= T\delta_{g,T}(\theta_G) +O_p(1),$ and ${S}_T(\theta)={S}_T(\theta_G) + O_p(1/\sqrt{T})$.

	Now we only need to show that the confidence set in the presence of weak proxy factors contains some $\theta$s in class-(2) with positive probability $\widetilde{\alpha}_1\geq l>0$ in large samples, which then leads to the conclusion. The proof of Theorem 1 in \cite{gospodinov2017spurious} implies that 
	\begin{align*}
		\mathcal{J}(\theta) = CD(\theta) + T \frac{\left( \iota_N 'S_T(\theta)^{-1} \left(\iota_N - q_{G,T}\theta \right)\right)^2  }{\iota_N'S_T(\theta)^{-1} \iota_N } +O_p(\frac{1}{\sqrt{T}})
	\end{align*}
	where $CD(\theta) = T \theta'q_{G,T}' P_1 \left(\left(\theta'\otimes P_1' \right)\widehat{V}\left(\sqrt{T} \text{vec}(q_{G,T}) \right)  \left(\theta\otimes P_1 \right)  \right)P_1'q_{G,T} \theta$. Notice $\inf_\theta CD(\theta)$ is a rank test (\cite{kleibergen2006generalized}), and in the presence of weak factors $CD\left(\theta_G + (0_{K\times 1},1)\right)$ converges to a non-central chi-square distribution which would then implies our claim.

\end{example} 
\clearpage
\newpage

\section{Consistency of the $\theta$ estimator} \label{a2}
Proofs related to section \ref{sec:HJN} relies heavily on the properties of our four-pass estimator, and thus we first discuss our four-pass estimator before we provide proofs.  This section contains the following results: (1) the number of the strong factors can be estimated consistently; (2) the common component can be estimated consistently when $\sqrt{N}/T \rightarrow 0$; (3) finally, we show that our proposed estimator  consistently estimates $\theta$ even in the presence of weak identification issues and we also derive its asymptotic distribution.  
\begin{assumption}\label{ass: factor loadings additional requirement }
	$\norm{(\beta; \gamma )_i} = \norm{c_{\beta\gamma,i}}\leq L$. 	 $Q_{(\beta; \gamma )} = \lim\limits_{N\rightarrow\infty} (\beta; \gamma )'(\beta; \gamma )/N$ with $Q_{(\beta; \gamma )}$ a $K_{vz}\times K_{vz}$ positive definite matrix with $0<l  \leq \lambda_{K_g}\left(Q_{(\beta; \gamma )} \right)\leq  \lambda_{1}\left(Q_{(\beta; \gamma )} \right)<L<\infty $. (Assumption \ref{assum:factor loading strength 2} implies this assumption via the Ostrowski theorem and some extra mild assumptions on $Q_r$.) such that 
	\begin{align*}
		\norm{ N^{-1}c_{\beta \gamma}'c_{\beta \gamma} - Q_{\beta \gamma}  } = o_p(1)
	\end{align*}
\end{assumption}   

\begin{assumption} \label{assum: bar sigma upper bound}
	Let $\gamma_N(t,t')= \mathbb{E}\left( N^{-1} \sum_{i=1}^N e_{it}e_{it'} \right)$, there exists a positive constant $L$ such that 
	\begin{align*}
		& (1)~ T^{-1} \sum_{t=1}^T\sum_{t'=1}^T \lvert \gamma_N(t,t')\rvert \leq L; \max_{t} |\gamma_N(t,t')|\leq L \\
		& (2)~ T^{-2} \sum_{t=1}^T\sum_{t'=1}^T \mathbb{E}\left(\sum_{i=1}^N e_{it}e_{it'} - \mathbb{E}\left(\sum_{i=1}^N e_{it}e_{it'}   \right)   \right)^2 = T^{-2} \sum_{t=1}^T\sum_{t'=1}^T\left( \mathbb{E}\left(\sum_{i=1}^N e_{it}e_{it'} \right)^2 - N^2 \gamma_N(t,t')^2 \right) \leq LN 
	\end{align*}
\end{assumption} 
Assumption \ref{assum: bar sigma upper bound} is implied by  Assumption \ref{assumption: weak cross-sectional correlation of the idiosyncratic component }. 
\begin{assumption}\label{assum:?}
	$\mathbb{E}\left| \frac{1}{\sqrt{N}} \sum_{i=1}^{N}  \left(e_{it}e_{is}-  \mathbb{E}e_{it}e_{is}\right) \right|^4 < L < \infty$
\end{assumption}

\begin{assumption} \label{ass: Weak dependence between proxy factors and Idiosyncratic Errors}This is implied by Assumptions   \ref{ass} and  \ref{assumption: weak cross-sectional correlation of the idiosyncratic component }
	(Weak dependence between proxy factors and Idiosyncratic Errors from \cite{bai2002determining})
	\begin{align*}
		& (i)~~	\mathbb{E}\left( \frac{1}{N}\sum_{i=1}^{N}  \norm{\frac{1}{\sqrt{T}} \sum_{t=1}^{T}  g_t e_{it} }^2 \right) \leq L\\
		& (ii)~~ 	\mathbb{E}\left( \frac{1}{K_{vz}}\sum_{i=1}^{K_{vz}}  \norm{\frac{1}{\sqrt{T}} \sum_{t=1}^{T}  g_t u_{vz, it} }^2 \right) \leq L\\
		&(iii)	
		~~ 	\mathbb{E}\left( \frac{1}{N}\sum_{i=1}^{N}  \norm{\frac{1}{\sqrt{T}} \sum_{t=1}^{T} u_{vz,t} e_{it} }^2 \right) \leq L
	\end{align*}
	
\end{assumption}

\begin{assumption}  \label{assum:  {New}	(additional assumptions on factor loadings and idiosyncratic terms) }
	For all N and T,
	\begin{align}
		&\mathbb{E}\sum_{t=1}^{T}\left( \sum_{i=1}^N c_{\beta \gamma,{ij}} e_{it}  \right)^2 \leq LNT \nonumber \\
		& \mathbb{E}\left( \sum_{t=1}^{T}\sum_{i=1}^N c_{\beta \gamma,{ij}} e_{it}  \right)^2 \leq LNT  
	\end{align}
\end{assumption} 

\begin{assumption} \label{ass:residual innovation factor restriction} 
	\begin{align*}
		\norm{T^{-1} u_{vz}' u_{vz}  - \Sigma_{vz}}_{\max} = O_p(T^{-\frac{1}{2}})
	\end{align*}
	with $\Sigma_{vz}$ a positive definite matrix. 
\end{assumption}

\begin{assumption}  \label{assum:factor norm bounded}
	\begin{align*}
		\sup_t \mathbb{E} \norm{u_{vz,t}}^4  \leq L
	\end{align*}
	This is identical to one imposed on factors (Assumption A) in \cite{bai2002determining} and \cite{anatolyev2018factor} impose a slightly stronger assumption. 
\end{assumption}

\begin{assumption} \label{assum:(New, additional ,  Time and Cross-Section Dependence and Heteroskedasticity) } 
	\begin{align*}
		(1)& \sum_{t}^{T} \left| \gamma_N(s,t)\right|\leq L \\
		(2)& \sum_{i=1}^{N} \left| \tau(i,j)\right| \leq L 
	\end{align*}
	with $\mathbb{E}e_{it}e_{jt}  = \tau_t(i,j),  \left| \tau_t(i,j)\right|\leq \left| \tau(i,j)\right|  $.  
\end{assumption}

\begin{assumption} \label{assum: (New, products of factor and e)}

	(i)	 \begin{align*}
		\frac{1}{NT}	 \mathbb{E}\norm{  \sum_{s=1}^T \sum_{n=1}^N   {u}_{vz,s}\left(e_{is}e_{it}-\mathbb{E}(e_{is}e_{it})   \right)   }^2 \leq L
	\end{align*}
	
	(ii) for each t, as $N\rightarrow \infty$
	\begin{align*}
		\frac{1}{\sqrt{N}}  \sum_{i=1}^{N}  c_{\beta\gamma, i} e_{ti} \rightarrow_d N(0,\Pi_t)
	\end{align*}
	with $\Pi_t = \lim\limits_{N \rightarrow \infty} \frac{1}{N} \sum_{i=1}^N \sum_{j=1}^N   c_{\beta\gamma, i} c_{\beta\gamma, j}'  \mathbb{E}e_{ti} e_{tj} $. And for all $j=1,\cdots,N$
	\begin{align*}
		\frac{1}{\sqrt{NT}}  \sum_{i=1}^{N} \sum_{t=1}^{T}  c_{\beta\gamma, i} \left(e_{ti}e_{tj} - \mathbb{E} e_{ti}e_{tj} \right) = O_p(1)
	\end{align*} 
	
	(iii) 
	\begin{align*}
		\mathbb{E}\norm{\frac{1}{\sqrt{NT}} \sum_{t=1}^{T} \sum_{i=1}^{N} u_{vz,t} c_{\beta\gamma,i}'e_{it}  }^2_F\leq M 
	\end{align*}
	
	Assumption \ref{assum: (New, products of factor and e)}  is identical to Assumption F in \cite{bai2003inferential}. 
\end{assumption}

%
%
%

%

\subsection*{\footnotesize Step (1): the number of the strong factors can be estimated consistently}
In this step we prove that we can estimate the number of the strong factors in the $u_{g,t}$ consistently. Here we only provide one way to estimate the number, the estimation approach is not unique. \cite{bai2002determining} propose multiple consistent estimators for the number of strong factors with different penalty functions. Here we use the one employed in \cite{giglio2017inference}. Under Assumption \ref{assum: factor structure in r}, we know
\begin{align}
	r_t = \widetilde{c} + \beta d_g \bar{g}_t + u_{g,t}
\end{align}
with $u_{g,t}= \beta v_t + \gamma z_t +e_t, \widetilde{c}=c+ \beta d_g(\bar{g}-\mu_g)$. $v_t$ is assumed to be of a $K_g\times 1$ vector, but the dimension of $z_t$ is unknown. We estimate the number $K_{vz}$ of the omitted strong factors  by 
\begin{align}
	\widehat{K}_{vz} = \arg \min_{K_g\leq j \leq K_{vz,\max}} \left( N^{-1}T^{-1} \lambda_j\left( \widehat{u}_{g}\widehat{u}_{g}' \right)  + j\phi(N,T)  \right) -1  
\end{align}
where $\widehat{u}_{g}$ is $T\times N$ matrix stacked with the residuals $\widehat{u}_{g,t}$, $K_{vz,\max}$ is an arbitrary upper bound for $K_{vz}$ and $\phi(N,T)$ is a penalty function with the properties $\phi(N,T)\rightarrow 0, \phi(N,T)/(N^{-\frac{1}{2}}+T^{-\frac{1}{2}}) \rightarrow \infty$. 
Now we show this estimator is consistent. 
\begin{theo} \label{theo:consistent factor number}
	Suppose Assumptions \ref{assum: factor structure in r}, \ref{ass: factor loadings additional requirement } - \ref{ass:residual innovation factor restriction} hold, let N,T increase then 
	\begin{align*}
		\widehat{K}_{vz} \rightarrow_p K_{vz}
	\end{align*}
\end{theo}
\noindent \begin{proof}[Proof of Theorem \ref{theo:consistent factor number}]
	We basically follow the steps in \cite{giglio2017inference} with small changes in the middle. 
	
	(1)	 We first prove the claim such that for $1\leq j\leq K_{vz}$
	\begin{align}
		\lvert N^{-1}T^{-1} \lambda_j\left( \widehat{u}_{g}\widehat{u}_{g}' \right) - \lambda_j\left(  (Q_{(\beta; \gamma  )})^{\frac{1}{2}} \Sigma_{(v; z )} (Q_{(\beta; \gamma  )})^{\frac{1}{2}}      \right) \rvert = o_p(1) \label{eq: first Kvz eigenvalue consistency}
	\end{align}
	with $Q_{(\beta; \gamma )} = \lim\limits_{N\rightarrow\infty} (\beta; \gamma )'(\beta; \gamma )/N$.
	
	For convenience, in this proof, denote ${c}_{\beta\gamma} = (\beta; \gamma ), u_{vz}= (v; z )$.

	(1.1) Notice 
	\begin{align*}
		\widehat{u}_{g}\widehat{u}_{g}'-  M_{\bar{G}} u_{vz} c_{\beta \gamma}'c_{\beta \gamma}u_{vz}' M_{\bar{G}} = M_{\bar{G}} u_{vz} c_{\beta \gamma}' e' M_{\bar{G}} + M_{\bar{G}} e c_{\beta \gamma}u_{vz}' M_{\bar{G}} + M_{\bar{G}} ee' M_{\bar{G}}  
	\end{align*}  	 
	with $\widehat{u}_{g}= M_{\bar{G}}r$. We show in the following steps that the three terms on the right is negligible when divided by $NT$.

	For the term $M_{\bar{G}} e e'  M_{\bar{G}}$, we have 
	\begin{align}
		\norm{M_{\bar{G}} ee' M_{\bar{G}} - n\Gamma_u } \leq  \norm{e e' - n\Gamma_u }_F +2\norm{ P_{\bar{G}} ee'}_F+\norm{P_{\bar{G}} ee'P_{\bar{G}} }_F
	\end{align}

	Assumption \ref{assum: bar sigma upper bound}.(1) implies that
	\begin{align}
		\mathbb{E}\norm{e}_F^2 =  \mathbb{E}\sum_{i=1}^N\sum_{t=1}^T e_{it}^2 \leq LNT
	\end{align}
	
	\begin{align}
		\mathbb{E}\norm{\bar{e}}_F^2 =  T^{-2}\mathbb{E}\sum_{i=1}^N\sum_{t=1}^T\sum_{t'=1}^T e_{it} e_{it'}=N T^{-2} \sum_{t=1}^T\sum_{t'=1}^T \lvert \gamma_N(t,t')\rvert \leq L NT^{-1}
	\end{align}

	Assumption \ref{assum: bar sigma upper bound}.(2) implies that
	\begin{align}
		\mathbb{E}\norm{e e' - n\Gamma_u}_F^2 =  \sum_{t=1}^T\sum_{t'=1}^T \mathbb{E}\left(\sum_{i=1}^N e_{it}e_{it'} - \mathbb{E}\left(\sum_{i=1}^N e_{it}e_{it'}   \right)   \right)^2  \leq LNT^2 
	\end{align}
	
	Assumption \ref{ass: Weak dependence between proxy factors and Idiosyncratic Errors} implies that 	
	\begin{align}
		\mathbb{E}\norm{{{G}}' e}_F^2  =  	NT	\mathbb{E}\left( \frac{1}{N}\sum_{i=1}^{N}  \norm{\frac{1}{\sqrt{T}} \sum_{t=1}^{T}  G_t e_{it} }^2_F \right)  \leq LNT
	\end{align}
	
	\begin{align}
		\norm{{{G}}' e}_F   = O_p(N^{\frac{1}{2}}T^{\frac{1}{2}}); 	\norm{\left(0,\iota_T \bar{g}' \right)' e}_F  = O_p(N^{\frac{1}{2}}T^{\frac{1}{2}})
	\end{align}
	
	\begin{align}  
		\norm{{\bar{G}}' e}_F   \leq     	\norm{{{G}}' e}_F   + 	\norm{\left(0,\iota_T \bar{g}' \right)' e}_F  = O_p(N^{\frac{1}{2}}T^{\frac{1}{2}}) \label{eq: norm bound barG e}
	\end{align}

	\begin{align}
		\norm{P_{\bar{G}} e}_F \leq T^{-1} \norm{{\bar{G}}}_F \norm{ \left({\bar{G}}'{\bar{G}}/T \right)^{-1}}_F\norm{{\bar{G}}' e}_F = O_p(N^{\frac{1}{2}})  
	\end{align}
	
	Then the following holds: 
	\begin{align}
		\norm{M_{\bar{G}} ee' M_{\bar{G}} - n\Gamma_u } \leq & \norm{e e' - n\Gamma_u }_F +2\norm{ P_{\bar{G}} ee'}_F+\norm{P_{\bar{G}} ee'P_{\bar{G}} }_F \nonumber \\
		= &O_p(N^{\frac{1}{2}}T)+ O_p(NT^{\frac{1}{2}}) +O_p(N) 
	\end{align}
	
	From assumption \ref{ass: Weak dependence between proxy factors and Idiosyncratic Errors} (1) we would have 
	\begin{align}
		\mathbb{E}\norm{\Gamma_u}_F^2 = 	\mathbb{E}  \sum_{t=1}^T\sum_{t'=1}^T \lvert \gamma_N(t,t')\rvert \leq LT
	\end{align}
	\begin{align*}
		\norm{M_{\bar{G}} ee' M_{\bar{G}}}_F \leq 	\norm{M_{\bar{G}} ee' M_{\bar{G}} - n\Gamma_u }_F + \norm{n\Gamma_u}_F = O_p(N^{\frac{1}{2}}T)+ O_p(NT^{\frac{1}{2}}) 
	\end{align*}

	For the term $M_{\bar{G}} u_{vz} c_{\beta \gamma}' e'  M_{\bar{G}}$, we have 
	\begin{align}
		\norm{M_{\bar{G}} u_{vz} c_{\beta \gamma}' e' M_{\bar{G}}}_F \leq &  \norm{u_{vz} c_{\beta \gamma}' e'}_F +\norm{P_{\bar{G}}u_{vz} c_{\beta \gamma}' e'}_F+\norm{u_{vz} c_{\beta \gamma}' e'P_{\bar{G}}}_F +\norm{P_{\bar{G}}u_{vz} c_{\beta \gamma}' e'P_{\bar{G}}}_F \nonumber  \\
		\leq & \left( 2(K_{vz}+1) \right)^2\norm{u_{vz} c_{\beta \gamma}' e'}_F
	\end{align} 
	
	Assumption \ref{assum:  {New}	(additional assumptions on factor loadings and idiosyncratic terms) } implies that 
	\begin{align}
		\mathbb{E}\norm{ec_{\beta \gamma} }_F^2  = \mathbb{E}\sum_{j=1}^{K_{vz}}\sum_{t=1}^{T}\left( \sum_{i=1}^N c_{\beta \gamma,{ij}} e_{it}  \right)^2 \leq LNT \label{eq: e c bounded}
	\end{align}	
	and 
	\begin{align}
		\norm{u_{vz}}\leq  KT\norm{T^{-1} u_{vz}' u_{vz}}_{\max} =O_p(T^{\frac{1}{2}})
	\end{align}
	\begin{align}
		\norm{u_{vz} c_{\beta \gamma}' e'}  \leq \norm{u_{vz} } \norm{  c_{\beta \gamma}' e'}_F = O_p(N^{\frac{1}{2}}T)  \label{eq: norm bound uce}
	\end{align}
	\begin{align}
		\norm{M_{\bar{G}} u_{vz} c_{\beta \gamma}' e' M_{\bar{G}}}  \leq  
		O_p(N^{\frac{1}{2}}T)
	\end{align}

	Therefore,  
	\begin{align}
		N^{-1}T^{-1}	\lvert\lambda_j(\widehat{u}_{g}\widehat{u}_{g}')-  \lambda_j(M_{\bar{G}} u_{vz} c_{\beta \gamma}'c_{\beta \gamma}u_{vz}' M_{\bar{G}}) \rvert=  o_p(1)
	\end{align}

	(1.2)
	To finish the proof of part (1) we only need to show the following two results:
	\begin{align}
		& (1.2.1)~~	\lvert 	N^{-1}T^{-1} \lambda_j(M_{\bar{G}} u_{vz} c_{\beta \gamma}'c_{\beta \gamma}u_{vz}' M_{\bar{G}})  - T^{-1}\lambda_j(M_{\bar{G}} u_{vz} Q_{\beta \gamma}u_{vz}' M_{\bar{G}})   \rvert=  o_p(1)\label{eq:p consistent step 1 1} \\
		& (1.2.2)~~\lvert  T^{-1}\lambda_j(M_{\bar{G}} u_{vz} Q_{\beta \gamma}u_{vz}' M_{\bar{G}}) -   \lambda_j\left(  (Q_{(\beta; \gamma  )})^{\frac{1}{2}} \Sigma_{(v; z )} (Q_{(\beta; \gamma  )})^{\frac{1}{2}}      \right) \rvert=  o_p(1) \label{eq:p consistent step 1 2}
	\end{align} 
	Equation  (\ref{eq:p consistent step 1 1}) is a direct result of Assumption \ref{ass: factor loadings additional requirement } and Weyl's inequality such that 
	\begin{align}
		& 	\lvert 	N^{-1}T^{-1} \lambda_j(M_{\bar{G}} u_{vz} c_{\beta \gamma}'c_{\beta \gamma}u_{vz}' M_{\bar{G}})  - T^{-1}\lambda_j(M_{\bar{G}} u_{vz} Q_{\beta \gamma}u_{vz}' M_{\bar{G}})   \rvert\nonumber\\ \leq & T^{-1}\norm{ N^{-1}c_{\beta \gamma}'c_{\beta \gamma} - Q_{\beta \gamma}  }\norm{M_{\bar{G}} u_{vz} }^2 =o_p(1)
	\end{align}
	Equation  (\ref{eq:p consistent step 1 1}) is a direct result of Assumption \ref{ass:residual innovation factor restriction} and Weyl's inequality such that 
	\begin{align}
		& 	\lvert   T^{-1}\lambda_j(M_{\bar{G}} u_{vz} Q_{\beta \gamma}u_{vz}' M_{\bar{G}}) -   \lambda_j\left(  (Q_{(\beta; \gamma  )})^{\frac{1}{2}} \Sigma_{vz} (Q_{(\beta; \gamma  )})^{\frac{1}{2}}      \right)   \rvert\nonumber\\ \leq & L \norm{T^{-1}  u_{vz}' u_{vz}  - \Sigma_{vz} }  =O_p(T^{-\frac{1}{2}})
	\end{align}

	Therefore, for $1\leq j\leq K_{vz}$, we have
	\begin{align}
		\lvert N^{-1}T^{-1} \lambda_j\left( \widehat{u}_{g}\widehat{u}_{g}' \right) - \lambda_j\left(  (Q_{(\beta; \gamma  )})^{\frac{1}{2}} \Sigma_{(v; z )} (Q_{(\beta; \gamma  )})^{\frac{1}{2}}      \right) \rvert = o_p(1)
	\end{align}
	(2) In part (2) we finish the proof of the consistency $\widehat{K}_{vz}$ by showing the following statement    
	\begin{align}
		\lim\inf\limits_{T,N\rightarrow \infty}	 \mathbb{P}\left(  \widehat{F}(K_{vz}+1) \leq  \widehat{F}(j), j=1,\cdots,N    \right) = 1
	\end{align}
	with $\widehat{F}({j}) =  N^{-1}T^{-1} \lambda_j\left( \widehat{u}_{g}\widehat{u}_{g}' \right)  + j\phi(N,T)    $.   Notice a direct result from step (1) is that for $1\leq j \leq  K_{vz},$ we can find $l>0$ such that $ \widehat{F}({j}) > l+o_p(1) $, while for  $K_{vz}+1\leq j,$ $ \widehat{F}({j}) =o_p(1)$. Then we only need to show that for $j>K_{vz}+1$, $\widehat{F}({j})>\widehat{F}(K_{vz}+1) $ with probability approaching to one. Notice for $j\geq K_{vz}+1$,
	\begin{align}
		\lambda_j\left( \widehat{u}_{g}\widehat{u}_{g}' \right)  \leq  \norm{ee'} = O_p(TN^{-\frac{1}{2}} + T^{-\frac{1}{2}}N  ) \label{eq: eigenvalue after the first K vz}
	\end{align} 
	and this implies $\widehat{F}({j})>\widehat{F}(K_{vz}+1) $ with probability approaching to one as $ (N^{\frac{1}{2}}T^{\frac{1}{2}} ) (j-K_{vz}-1)\phi(N,T)  >  (N^{-\frac{1}{2}}T^{-\frac{1}{2}} ) \left(	\lambda_{K_{vz}+1}\left( \widehat{u}_{g}\widehat{u}_{g}' \right) - 	\lambda_j\left( \widehat{u}_{g}\widehat{u}_{g}' \right) \right)$ with probability approaching to one.
	
\end{proof}	

\subsection*{\footnotesize Step (2): the common component can be estimated consistently when $\sqrt{N}/T \rightarrow 0$}	 
Denote $m_{NT}= \min\{N,T\}$.
\begin{lem} \label{lem: consistency of the estimated factors}
	Suppose Assumptions \ref{assum: factor structure in r} - \ref{assumption: weak cross-sectional correlation of the idiosyncratic component }, \ref{ass: factor loadings additional requirement } - \ref{ass:residual innovation factor restriction} hold, let N,T increase then 
	\begin{align*}
		m_{NT} \frac{1}{T} \sum_{t=1}^{T} \norm{ \widetilde{u}_{vz,t} - H' u_{vz,t}  }_F^2 = O_p(1)
	\end{align*}
	with $H=\left( c_{\beta\gamma}'c_{\beta\gamma}/N  \right)\left(u_{vz}' \widetilde{u}_{vz}/T  \right)V_{NT} $  and $V_{NT}$ being the $K_{vz}\times K_{vz}$ diagonal matrix of $\lambda_{i}\left(\widehat{u}_g\widehat{u}_g' /NT\right), i=1,\cdots, K_{vz}$.\\
	
	\noindent \begin{proof}[Proof of Lemma \ref{lem: consistency of the estimated factors}]
		This proof resembles the proof of Theorem 1 in \cite{bai2002determining}. From the normalization $\widetilde{u}_{vz}' \widetilde{u}_{vz}/T =I_{K_{vz}}$  we know 
		\begin{align}
			\norm{ \widetilde{u}_{vz}}_F^2=    \sum_{t=1}^T\norm{ \widetilde{u}_{vz,t}}_F^2  = O_p(T)  \label{eq: u_vz norm power 2}
		\end{align}	 
		From the proof in section   and the above equation we know
		\begin{align}
			\norm{ H }_F \leq  \norm{  c_{\beta\gamma}'c_{\beta\gamma}/N }_F \norm{u_{vz}'{u}_{vz}/T  }_F^{1/2}\norm{\widetilde{u}_{vz}' \widetilde{u}_{vz}/T  }_F^{1/2}  = O_p(1)
		\end{align} 	 
		From equations (\ref{eq:decompose factor difference}) and (\ref{eq: first Kvz eigenvalue consistency}), we  know 
		\begin{align}
			\norm{ \widetilde{u}_{vz,t} - H' u_{vz,t} }_F^2 \leq L \sum_{i=1}^{6} a_{i,t}^2   
		\end{align}  
		with 
		\begin{align}
			a_{1,t}^2= & T^{-2} \norm{ \sum_{s=1}^{T}  \widetilde{u}_{vz,s}\gamma_N(s,t)     }_F^2   \nonumber \\
			a_{2,t}^2= & T^{-2}\norm{  \sum_{s=1}^{T}  \widetilde{u}_{vz,s} \zeta_{ee, st} }_F^2   \nonumber \\
			a_{3,t}^2= & T^{-2}\norm{ \sum_{s=1}^{T}  \widetilde{u}_{vz,s} \zeta_{ue, st}  }_F^2  \nonumber \\
			a_{4,t}^2= & T^{-2}\norm{  \sum_{s=1}^{T}  \widetilde{u}_{vz,s} \zeta_{eu, st}}_F^2  \nonumber \\
			a_{5,t}^2= &  \norm{\frac{1}{NT} \widetilde{u}_{vz}'P_{\bar{G}}u_g u_{g,t}'    }_F^2 \nonumber \\
			a_{6,t}^2= &  \norm{ \frac{1}{NT}\widetilde{u}_{vz}'M_{\bar{G}} u_g \left( \widehat{u}_{g,t} - u_{g,t} \right)}_F^2  
		\end{align}	 
		From equation (\ref{eq: u_vz norm power 2})  and Assumption \ref{assum: bar sigma upper bound}  (i), 
		\begin{align}
			\sum_{t=1}^T	a_{1,t}^2  \leq  \left( T^{-1}  \sum_{s=1}^T\norm{ \widetilde{u}_{vz,s}}_F^2   \right) \left(T^{-1} \sum_{t=1}^T\sum_{s=1}^T\gamma_N(s,t)^2  \right) = O_p(1)
		\end{align}

		From equation (\ref{eq: u_vz norm power 2}) and Assumption \ref{assum: forth moments of ee},  
		\begin{align}
			\sum_{t=1}^T	a_{2,t}^2  \leq  \left( T^{-1}  \sum_{s=1}^T\norm{ \widetilde{u}_{vz,s}}_F^2   \right) \left(T^{-2}  \sum_{t=1}^T\sum_{s=1}^T \left(\sum_{u=1}^{T}\zeta_{ee, su}\zeta_{ee, ut} \right)^2   \right)^{\frac{1}{2}} = O_p(T/N) 
		\end{align}
		
		and from equation (\ref{eq: e c bounded}) 
		\begin{align}
			\sum_{t=1}^T	a_{3,t}^2 \leq  \left(  \sum_{t=1}^T\norm{ c_{\beta\gamma}' e_t/N}_F^2  \right) \left( T^{-1}  \sum_{s=1}^T\norm{ \widetilde{u}_{vz,s}}_F^2  \right)\left( T^{-1}  \sum_{s=1}^T\norm{  {u}_{vz,s}}_F^2  \right) = O_p(T/N)
		\end{align}  
		and similarly $\sum_{t=1}^T	a_{4,t}^2 \leq O_p(T/N)$.

		From Assumption \ref{assumption: weak cross-sectional correlation of the idiosyncratic component },
		\begin{align}
			\norm{P_{\bar{G}} u_g}_F \leq T^{-1} \norm{{\bar{G}}}_F \norm{ \left({\bar{G}}'{\bar{G}}/T \right)^{-1}}_F\norm{{\bar{G}}' u_g}_F = O_p(N^{\frac{1}{2}})
		\end{align}
		and thus
		\begin{align}
			\sum_{t=1}^T	a_{5,t}^2 \leq  \left( T^{-1}  \sum_{s=1}^T\norm{ \widetilde{u}_{vz,s}}_F^2   \right) \left( N^{-1}   \norm{P_{\bar{G}}  {u}_{g} }_F^2   \right) \left( N^{-1}T^{-1}  \sum_{t=1}^T\norm{ {u}_{g,t}}_F^2   \right) = O_p(1)
		\end{align}    
		
		From equations (\ref{eq: first Kvz eigenvalue consistency}) and (\ref{eq: eigenvalue after the first K vz})
		\begin{align}
			\norm{M_{\bar{G}}  {u}_{g} }_F^2 = O_p(NT)
		\end{align}
		and from $\sqrt{T}\left( \widehat{\beta}_g - \beta_g \right) = O_p(1) $ we know
		\begin{align}
			\sum_{t=1}^T\norm{ \widehat{u}_{g,t} - {u}_{g,t}}_F^2     \leq  \norm{\sqrt{T}\left( \widehat{\beta}_g - \beta_g \right)}_F^2 \norm{  \bar{G}/\sqrt{T} }_F^2 \leq O_p(N)
		\end{align}

		Therefore, 
		\begin{align}
			\sum_{t=1}^T	a_{6,t}^2 \leq   \left( T^{-1}  \sum_{s=1}^T\norm{ \widetilde{u}_{vz,s}}_F^2   \right) \left( N^{-1} T^{-1}  \norm{M_{\bar{G}}  {u}_{g} }_F^2   \right) \left(  N^{-1} \sum_{t=1}^T\norm{ \widehat{u}_{g,t} - {u}_{g,t}}_F^2   \right) \leq O_p(1)
		\end{align}    
		
		Finally, all the above imply the result:
		\begin{align}
			\frac{1}{T} \sum_{t=1}^{T} \norm{ \widetilde{u}_{vz,t} - H' u_{vz,t}  }_F^2 = O_p(1/N) + O_p(1/T)
		\end{align}

	\end{proof}
\end{lem}

Given the discussion in the last subsection, it would be safe to assume that we know the value of $K_{vz}$, namely the number of the strong omitted factors are known. In this section, we show the asymptotic properties of the common component estimator.  Essentially, this section is a building block for the consistent $\theta_g$ estimator as introduced in the next section. 

The method of principal components gives the following estimators 
\begin{align}
	\left( \widetilde{c}_{\beta \gamma},  \widetilde{u}_{vz} \right) = \arg \min_{c_i, u_t  ~s.t.~ \sum_{t=1}^{T} u_tu_t'  /T = I_{K_{vz}}  } \sum_{i,t}  \left( \widehat{u}_{g,it} -   c_{i}'  u_{t} \right)^2
\end{align}
The estimator $ \widetilde{u}_{vz} $ is equal to the $\sqrt{T}$ times eigenvector associated with the $K_{vz}$ largest eigenvalues of the matrix $\widehat{u}_g\widehat{u}_g' $, and $\widetilde{c}_{\beta \gamma}= \widetilde{u}_{vz}' \widehat{u}_g/T$ corresponds to the OLS estimator regressing $\widehat{u}_g$ over $\widetilde{u}_vz$.  Specially we are interested in the common component estimator $\widetilde{c} =   \widetilde{u}_{vz}(\widetilde{c}_{\beta \gamma})' $, which would serve as a proxy for the common component in $u_g$ in our proposed estimator.  
\begin{theo}\label{theo: (factor)}
	Suppose Assumptions \ref{assum: factor structure in r} - \ref{assumption: weak cross-sectional correlation of the idiosyncratic component }, \ref{ass: factor loadings additional requirement } - \ref{ass:residual innovation factor restriction}, \ref{assum:factor norm bounded}, \ref{assum: (New, products of factor and e)} hold, and  $\sqrt{N}/T \rightarrow 0$,  
	\begin{align}
		\sqrt{N}\left( \widetilde{u}_{vz,t} -H' {u}_{vz,t}  \right) = V_{NT}^{-1}\frac{1}{T} \sum_{s=1}^{T} \left( \widetilde{u}_{vz,s}{u}_{vz,s}' \right)   \frac{1}{\sqrt{N}} \sum_{i=1}^N c_{\beta\gamma,i} e_{it} + o_p(1)
	\end{align} 
	when  $\lim\inf \sqrt{N}/T \rightarrow \tau > 0$,  
	\begin{align}
		T \left( \widetilde{u}_{vz,t} -H' {u}_{vz,t}  \right) = O_p(1)
	\end{align} 
	
	\noindent \begin{proof}[Proof of Theorem \ref{theo: (factor)}]
		We make use the equation $ \widetilde{u}_{vz} =\frac{1}{NT}M_{\bar{G}} u_g u_g' M_{\bar{G}} \widetilde{u}_{vz} V_{NT}^{-1}  $, and $H=\left( c_{\beta\gamma}'c_{\beta\gamma}/N  \right)\left(u_{vz}' \widetilde{u}_{vz}/T  \right)V_{NT}^{-1} $, and $V_{NT}$ is the $K_{vz}\times K_{vz}$ diagonal matrix of $\lambda_{i}\left(\widehat{u}_g\widehat{u}_g' /NT\right), i=1,\cdots, K_{vz}$. 
		Then 
		\begin{align}
			& \widetilde{u}_{vz,t} - H' u_{vz,t} \nonumber \\=&\frac{1}{NT}V_{NT}^{-1} \widetilde{u}_{vz}'     \left(M_{\bar{G}} u_g \widehat{u}_{g,t}   -    u_{vz} \left( c_{\beta\gamma}'c_{\beta\gamma}  \right) u_{vz,t} \right) \nonumber \\ 
			= & \frac{1}{NT}V_{NT}^{-1} \widetilde{u}_{vz}'     \left(  u_g u_{g,t}  -    u_{vz} \left( c_{\beta\gamma}'c_{\beta\gamma}  \right) u_{vz,t} - P_{\bar{G}}u_g u_{g,t}  +  M_{\bar{G}} u_g \left( \widehat{u}_{g,t} - u_{g,t} \right) \right) \nonumber \\ 
			= & \frac{1}{NT}V_{NT}^{-1} \widetilde{u}_{vz}'     \left(   \left( u_{vz} c_{\beta\gamma}'  +e \right)\left( c_{\beta\gamma}u_{vz,t}  + e_{t}\right) -    u_{vz} \left( c_{\beta\gamma}'c_{\beta\gamma}  \right) u_{vz,t} - P_{\bar{G}}u_g u_{g,t}  +  M_{\bar{G}} u_g \left( \widehat{u}_{g,t} - u_{g,t} \right) \right) \nonumber \\ 
			= & V_{NT}^{-1}\left( \frac{1}{T} \sum_{s=1}^{T}  \widetilde{u}_{vz,s}\gamma_N(s,t) + \frac{1}{T} \sum_{s=1}^{T}  \widetilde{u}_{vz,s} \zeta_{ee, st} + \frac{1}{T} \sum_{s=1}^{T}  \widetilde{u}_{vz,s} \zeta_{ue, st} + \frac{1}{T} \sum_{s=1}^{T}  \widetilde{u}_{vz,s} \zeta_{eu, st}        \right)   \nonumber \\
			&+V_{NT}^{-1}\left(  -\frac{1}{NT} \widetilde{u}_{vz}'P_{\bar{G}}u_g u_{g,t}  +  \frac{1}{NT}\widetilde{u}_{vz}'M_{\bar{G}} u_g \left( \widehat{u}_{g,t} - u_{g,t} \right)      \right)  \label{eq:decompose factor difference}
		\end{align} 
		with 
		\begin{align}
			\zeta_{ee, st}= & e_s' e_t /N -\gamma_N(s,t) \\
			\zeta_{ue, st}= &   u_{vz, s}' c_{\beta\gamma}' e_t/N \\
			\zeta_{eu, st}= & u_{vz, t}' c_{\beta\gamma}' e_s/N
		\end{align}
		Now we analyze the terms on the left hand of equation (\ref{eq:decompose factor difference}).

		(1) Assumption \ref{assum:factor norm bounded} and Assumption \ref{assum: bar sigma upper bound} imply that 
		\begin{align}
			\frac{1}{T} \sum_{s=1}^{T}   {u}_{vz,s}\gamma_N(s,t) = O_p(1/T)
		\end{align}
		Lemma \ref{lem: consistency of the estimated factors} and Assumption \ref{assum: bar sigma upper bound} imply that
		\begin{align}
			\frac{1}{T} \norm{\sum_{s=1}^{T}  \left( \widetilde{u}_{vz,s}-  {u}_{vz,s}\right)\gamma_N(s,t)}_F \leq & \frac{1}{\sqrt{T}} \left(\frac{1}{T}   \sum_{s=1}^{T} \norm{ \widetilde{u}_{vz,s}-  {u}_{vz,s} }_F^2\right)^{1/2}  \left(     \sum_{t}^{T} \left| \gamma_N(s,t)\right|^2\right)^{1/2}  \nonumber \\=& O_p(1/(\sqrt{Tm_{NT}}  ))
		\end{align}
		and thus
		\begin{align}
			\frac{1}{T} \sum_{s=1}^{T}  \widetilde{u}_{vz,s}\gamma_N(s,t)  = &    \frac{1}{T} \sum_{s=1}^{T}   {u}_{vz,s}\gamma_N(s,t)   +   \frac{1}{T} \sum_{s=1}^{T}  \left(\widetilde{u}_{vz,s}-  {u}_{vz,s}\right)\gamma_N(s,t)   \leq O_p(1/(\sqrt{Tm_{NT}}  ))
		\end{align}
		
		(2) 
		
		Assumption \ref{assum: (New, products of factor and e)}  implies that 
		\begin{align}
			\frac{1}{T} \sum_{s=1}^{T}   {u}_{vz,s}\zeta_{ee, st} = O_p(1/\sqrt{NT})
		\end{align}

		Assumption \ref{assum: forth moments of ee} implies that
		\begin{align}
			\frac{1}{ {T}}     \sum_{s=1}^{T} \left| \zeta_{ee, st}\right|^2 = &	\frac{1}{ {T}} \frac{1}{N}    \sum_{s=1}^{T} \left| \frac{1}{\sqrt{N}}  \left( e_{is}e_{it}-\mathbb{E}(e_{is}e_{it})  \right)  \right|^2 = O_p(1/(N)) \label{eq:ee2}
		\end{align}
		Lemma \ref{lem: consistency of the estimated factors} and  equation (\ref{eq:ee2}) imply that
		\begin{align}
			\frac{1}{T} \norm{\sum_{s=1}^{T}  \left( \widetilde{u}_{vz,s}-  {u}_{vz,s}\right)\zeta_{ee, st}} \leq &  \left(\frac{1}{T}   \sum_{s=1}^{T} \norm{ \widetilde{u}_{vz,s}-  {u}_{vz,s} }^2\right)^{1/2}  \left(\frac{1}{ {T}}     \sum_{s=1}^{T} \left| \zeta_{ee, st}\right|^2\right)^{1/2}  \nonumber \\=& O_p(1/(\sqrt{Nm_{NT}}  ))
		\end{align}
		and thus
		\begin{align}
			\frac{1}{T} \sum_{s=1}^{T}  \widetilde{u}_{vz,s}\zeta_{ee, st}  = &    \frac{1}{T} \sum_{s=1}^{T}   {u}_{vz,s}\zeta_{ee, st}   +   \frac{1}{T} \sum_{s=1}^{T}  \left(\widetilde{u}_{vz,s}-  {u}_{vz,s}\right)\zeta_{ee, st}  \leq O_p(1/(\sqrt{Nm_{NT}}  ))
		\end{align}

		(3)
		Assumption \ref{assum: (New, products of factor and e)}   implies that 
		\begin{align}
			\frac{1}{T} \sum_{s=1}^{T}   {u}_{vz,s}\zeta_{ue, st} =  \left(\frac{1}{T} \sum_{s=1}^{T}   {u}_{vz,s}{u}_{vz,s}' \right) \frac{1}{\sqrt{N}} \left(\frac{1}{\sqrt{N}}  \sum_{i=1}^{N}  c_{\beta\gamma, i} e_{ti}  \right)    =O_p(1/\sqrt{N})
		\end{align}
		and
		\begin{align}
			\frac{1}{ {T}}     \sum_{s=1}^{T} \left| \zeta_{ue, st}\right|^2 = &	\frac{1}{ {T}}     \sum_{s=1}^{T}   \left(  u_{vz, s}' c_{\beta\gamma}' e_t/N  \right)^2 \leq \norm{u_{vz} }_F/T  \norm{  c_{\beta \gamma}' e'_t/\sqrt{N}}_F/N  =  O_p(1/ {N})  \label{eq:ue2}
		\end{align}
		Lemma \ref{lem: consistency of the estimated factors} and  equation (\ref{eq:ue2}) imply that
		\begin{align}
			\frac{1}{T} \norm{\sum_{s=1}^{T}  \left( \widetilde{u}_{vz,s}-  {u}_{vz,s}\right)\zeta_{ue, st}} \leq &  \left(\frac{1}{T}   \sum_{s=1}^{T} \norm{ \widetilde{u}_{vz,s}-  {u}_{vz,s} }^2\right)^{1/2}  \left(\frac{1}{ {T}}     \sum_{t}^{T} \left| \zeta_{ue, st}\right|^2\right)^{1/2}  \nonumber \\=& O_p(1/(\sqrt{Nm_{NT}}  ))
		\end{align}
		and thus
		\begin{align}
			\frac{1}{T} \sum_{s=1}^{T}  \widetilde{u}_{vz,s}\zeta_{ue, st}  = &    \frac{1}{T} \sum_{s=1}^{T}   {u}_{vz,s}\zeta_{ue, st}   +   \frac{1}{T} \sum_{s=1}^{T}  \left(\widetilde{u}_{vz,s}-  {u}_{vz,s}\right)\zeta_{ue, st}  = O_p(1/\sqrt{ N})
		\end{align}
		
		(4)  
		Lemma \ref{lem: consistency of the estimated factors} and   Assumption \ref{assum: (New, products of factor and e)}   imply that
		\begin{align}
			\norm{\frac{1}{T} \sum_{s=1}^{T}  \left(\widetilde{u}_{vz,s}- {u}_{vz,s}\right)e_s'c_{\beta\gamma} /N u_{vz, t}}_F \leq & \frac{1}{\sqrt{N}} \left(\frac{1}{T} \sum_{s=1}^T \norm{ \widetilde{u}_{vz,s}- {u}_{vz,s}  }^2_F \right)^{\frac{1}{2}} \left(\frac{1}{T} \sum_{s=1}^T \norm{ e_s'c_{\beta\gamma}/\sqrt{N}}^2_F \right)^{\frac{1}{2}} \norm{u_{vz,t}}_F \nonumber \\
			= & O_p(1/\sqrt{Nm_{NT}}) 
		\end{align}
		Assumption \ref{assum: (New, products of factor and e)} implies
		\begin{align}
			\frac{1}{\sqrt{NT}}\left(\frac{1}{\sqrt{NT}} \sum_{s=1}^{T} {u}_{vz,s}e_s'c_{\beta\gamma}\right) u_{vz, t} = O_p(\frac{1}{\sqrt{NT}})  
		\end{align}
		and thus
		\begin{align}
			\frac{1}{T} \sum_{s=1}^{T}  \widetilde{u}_{vz,s} \zeta_{ue, st} = \frac{1}{T} \sum_{s=1}^{T}  \left(\widetilde{u}_{vz,s}- {u}_{vz,s}\right)e_s'c_{\beta\gamma} /N u_{vz, t} +\frac{1}{T} \sum_{s=1}^{T} {u}_{vz,s}e_s'c_{\beta\gamma} /N u_{vz, t}   \leq O_p(1/\sqrt{Nm_{NT}}) 
		\end{align}

		(5)
		Similar to the derivation of equation (\ref{eq: norm bound barG e}), we have
		\begin{align}
			\norm{{\bar{G}}'u_g}_F = O_p(N^{\frac{1}{2}}T^{\frac{1}{2}});  ~~	\norm{{\bar{G}}'u_{vz}}_F = O_p( T^{\frac{1}{2}})
		\end{align} 
		which imply that
		\begin{align}
			& \frac{1}{NT}\norm{ \left(\widetilde{u}_{vz}- {u}_{vz}\right)'{\bar{G}}\left({\bar{G}}'{\bar{G}}\right)^{-1}{\bar{G}}'u_g u_{g,t}'}_F \nonumber \\
			\leq &  \left(\frac{1}{T} \sum_{t=1}^T \norm{\widetilde{u}_{vz,t}- {u}_{vz,t}}^2_F  \right)^\frac{1}{2} \frac{1}{\sqrt{T}}\left(\norm{{\bar{G}}/\sqrt{T}\left({\bar{G}}'{\bar{G}}/T\right)^{-1}}_F\right) \left(\frac{1}{\sqrt{NT}} \norm{{\bar{G}}'u_g}_F \right) \norm{u_{g,t}/\sqrt{N}}_F \nonumber \\
			= & O_p(1/\sqrt{Tm_{NT}})  	
		\end{align}
		and 
		\begin{align}
			& \frac{1}{NT} \norm{{u}_{vz}'{\bar{G}}\left({\bar{G}}'{\bar{G}}\right)^{-1}{\bar{G}}'u_g u_{g,t}'}_F \nonumber \\
			\leq  &   \left( \frac{1}{\sqrt{T}}\norm{{u}_{vz}'{\bar{G}}}_F  \right)\frac{1}{ {T}}\left(\norm{ \left({\bar{G}}'{\bar{G}}/T\right)^{-1}}_F\right) \left(\frac{1}{\sqrt{NT}} \norm{{\bar{G}}'u_g}_F \right) \norm{u_{g,t}/\sqrt{N}}_F  \nonumber \\
			= & O_p(1/T) \label{eq:81}
		\end{align}
		Therefore,
		\begin{align}
			\frac{1}{NT} \widetilde{u}_{vz}'P_{\bar{G}}u_g u_{g,t}'  = & 
			\frac{1}{NT} \left(\widetilde{u}_{vz}- {u}_{vz}\right)'{\bar{G}}\left({\bar{G}}'{\bar{G}}\right)^{-1}{\bar{G}}'u_g u_{g,t}' + 
			\frac{1}{NT} {u}_{vz}'{\bar{G}}\left({\bar{G}}'{\bar{G}}\right)^{-1}{\bar{G}}'u_g u_{g,t}' \nonumber \\
			\leq  & O_p(1/\sqrt{Tm_{NT}})
		\end{align}

		(6)
		
		\begin{align}
			\frac{1}{NT} \norm{{u}_{vz}'{u}_{vz}c_{\beta\gamma}'  \left( \widehat{u}_{g,t} - u_{g,t} \right)  }_F \leq  &  \frac{1}{\sqrt{NT}} \norm{{u}_{vz}' {u}_{vz}  /T }_F \norm{\frac{\sqrt{T}c_{\beta \gamma}'\left( \widehat{u}_{g,t} - u_{g,t} \right)}{\sqrt{N}}}_F = O_p(\frac{1}{\sqrt{NT}}) 
		\end{align}

		\begin{align}
			\frac{1}{NT}\norm{{u}_{vz}' e  \left( \widehat{u}_{g,t} - u_{g,t} \right) }_F \leq  \frac{1}{ T} \norm{{u}_{vz}' e  /\sqrt{NT} }_F \norm{\frac{\sqrt{T} \left( \widehat{u}_{g,t} - u_{g,t} \right)}{\sqrt{N}}}_F = O_p(\frac{1}{T}) 
		\end{align}
		and thus
		\begin{align}
			\frac{1}{NT}\norm{\widetilde{u}_{vz}'M_{\bar{G}} u_g   }_F\leq  & 	\frac{1}{NT}\norm{\widetilde{u}_{vz}'  u_g \left( \widehat{u}_{g,t} - u_{g,t} \right)}_F +	\frac{1}{NT}\norm{\widetilde{u}_{vz}'P_{\bar{G}} u_g \left( \widehat{u}_{g,t} - u_{g,t} \right)}_F =O_p(\frac{1}{\sqrt{NT}}) + O_p( \frac{1}{T}) \label{eq:85}
		\end{align}

		From the above discussion, when   $\sqrt{N}/T \rightarrow 0$, only the third term in equation (\ref{eq:decompose factor difference}) matters in the asymptotic behavior and thus 
		\begin{align}
			\sqrt{N}\left( \widetilde{u}_{vz,t} -H' {u}_{vz,t}  \right) = V_{NT}^{-1}\frac{1}{T} \sum_{s=1}^{T} \left( \widetilde{u}_{vz,s}{u}_{vz,s}' \right)   \frac{1}{\sqrt{N}} \sum_{i=1}^N c_{\beta\gamma,i} e_{it} + o_p(1)
		\end{align} 
		Furthermore, when $\lim\inf \sqrt{N}/T \rightarrow \tau > 0$, we have
		\begin{align}
			T \left( \widetilde{u}_{vz,t} -H' {u}_{vz,t}  \right) = O_p(1)
		\end{align}

	\end{proof}
\end{theo}

\begin{theo} \label{theo: (factor loading)}
	Suppose Assumptions \ref{assum: factor structure in r} - \ref{assumption: weak cross-sectional correlation of the idiosyncratic component }, \ref{ass: factor loadings additional requirement }  - \ref{assum: (New, products of factor and e)} hold, and $\sqrt{T}/N\rightarrow 0$ 
	\begin{align*}
		\sqrt{T}\left( 	 \widetilde{c}_{\beta\gamma,i} -H^{-1}{c}_{\beta\gamma,i}  \right) = H' \frac{1}{\sqrt{T} }u_{vz}' e_i + o_p(1)
	\end{align*}	 
	
	\noindent \begin{proof}[Proof of Theorem \ref{theo: (factor loading)}] 		
		\begin{align}
			\widetilde{c}_{\beta\gamma, i} = & \frac{1}{T}\widetilde{u}_{vz}' \widehat{u}_{g,i} \nonumber \\
			=& \frac{1}{T}\widetilde{u}_{vz}'   \left(  \left(\widetilde{u}_{vz}H^{-1} +u_{vz} - \widetilde{u}_{vz}H^{-1} \right) c_{\beta\gamma,i}  + e_i \right)-\frac{1}{T}\widetilde{u}_{vz}'\left(\widehat{u}_{g,i}-u_{g,i} \right)   \nonumber \\
			=& H^{-1} c_{\beta\gamma, i}  + \frac{1}{T} H' u_{vz}' e_i +\frac{1}{T} \widetilde{u}_{vz}'  \left( u_{vz} - \widetilde{u}_{vz}H^{-1} \right) c_{\beta\gamma,i} +  \frac{1}{T} \left(\widetilde{u}_{vz}' -  H' u_{vz}'\right) e_i-\frac{1}{T}\widetilde{u}_{vz}'\left(\widehat{u}_{g,i}-u_{g,i} \right)  
		\end{align}
		
		We show the last three terms are at most $O_p(1/m_{NT})$.

		(1) For the term   $\frac{1}{T}   \left( \widetilde{u}_{vz} -  {u}_{vz}H \right)'  e_{i}$, equation (\ref{eq:decompose factor difference}) implies that 
		\begin{align}
			&\frac{1}{T}   \left( \widetilde{u}_{vz} -  {u}_{vz}H \right)'  e_{i} =   V_{NT}^{-1}\left( \frac{1}{T^2} \sum_{t=1}^{T} \sum_{s=1}^{T}  \widetilde{u}_{vz,s}\gamma_N(s,t)e_{it} + \frac{1}{T^2} \sum_{t=1}^{T}  \sum_{s=1}^{T}  \widetilde{u}_{vz,s} \zeta_{ee, st}e_{it} + \frac{1}{T^2} \sum_{t=1}^{T} \sum_{s=1}^{T}  \widetilde{u}_{vz,s} \zeta_{ue, st} e_{it}       \right)   \nonumber \\
			&+V_{NT}^{-1}   \left(\frac{1}{T^2} \sum_{t=1}^{T} \sum_{s=1}^{T}  \widetilde{u}_{vz,s} \zeta_{eu, st}  e_{it}  -\frac{1}{NT^2}\sum_{t=1}^{T} \widetilde{u}_{vz}'P_{\bar{G}}u_g u_{g,t}e_{it}  +  \frac{1}{NT^2}\sum_{t=1}^{T}\widetilde{u}_{vz}'M_{\bar{G}} u_g \left( \widehat{u}_{g,t} - u_{g,t} \right)e_{it}    \right) \nonumber \\
			& = V_{NT}^{-1}\sum_{j=1}^6 a_j
		\end{align}
		There are six terms on the left hand side of the above equation, and we analyze each of $a_j, j=1,\cdots, 6$ to determine the order of  $\frac{1}{T}   \left( \widetilde{u}_{vz} -  {u}_{vz}H \right)'  e_{i}$.

		Assumption \ref{assumption: weak cross-sectional correlation of the idiosyncratic component } and Lemma 1(i) in  \cite{bai2002determining} imply that  
		\begin{align}
			&	\norm{\frac{1}{T^2} \sum_{t=1}^{T} \sum_{s=1}^{T}  \left(\widetilde{u}_{vz,s} -H' {u}_{vz,s} \right) \gamma_N(s,t)e_{it}}_F \nonumber \\
			\leq &\frac{1}{\sqrt{T}} \left(\frac{1}{T} \sum_{s=1}^{T}  \norm{\widetilde{u}_{vz,s} -H' {u}_{vz,s} }_F^2 \right)^{\frac{1}{2}} \left( \frac{1}{T}   \sum_{t=1}^{T} \sum_{s=1}^{T}   \gamma_N(s,t)^2  \frac{1}{T}\sum_{t=1}^{T}e_{it}^2\right)^{\frac{1}{2}} \leq O_p(\frac{1}{\sqrt{Tm_{NT}}})
		\end{align}
		
		Assumptions \ref{assumption: weak cross-sectional correlation of the idiosyncratic component }, \ref{assum:factor norm bounded} and Lemma 1(i) in  \cite{bai2002determining} imply that  \begin{align}
			\mathbb{E}\norm{\frac{1}{T^2} \sum_{t=1}^{T} \sum_{s=1}^{T}    {u}_{vz,s}   \gamma_N(s,t)e_{it}}_F \leq    \frac{1}{T^2} \sum_{t=1}^{T} \sum_{s=1}^{T}    \left(\mathbb{E}\norm{{u}_{vz,s} }^2_F\right)^{\frac{1}{2}}  \gamma_N(s,t)^2(\mathbb{E}e_{it}^2)^{\frac{1}{2}} =O(\frac{1}{T})
		\end{align}
		
		Therefore,
		\begin{align}
			a_1 = \frac{1}{T^2} \sum_{t=1}^{T} \sum_{s=1}^{T}  \left(\widetilde{u}_{vz,s} -H' {u}_{vz,s} \right) \gamma_N(s,t)e_{it} + \frac{1}{T^2} \sum_{t=1}^{T} \sum_{s=1}^{T}   H' {u}_{vz,s}   \gamma_N(s,t)e_{it} \leq O_p(\frac{1}{\sqrt{Tm_{NT}}})
		\end{align}

		Assumption \ref{assumption: weak cross-sectional correlation of the idiosyncratic component } implies that
		\begin{align}
			\frac{1}{T}\sum_{t=1}^{T}  \zeta_{ee, st}e_{it}  = \frac{1}{\sqrt{N}} \frac{1}{T} \sum_{t=1}^{T} \left( \frac{1}{\sqrt{N}} \sum_{i=1}^N\left(e_{is}e_{it}-\mathbb{E}\left(e_{is}e_{it}\right)  \right)  \right)e_{it} =O_p(\frac{1}{\sqrt{N}}) \label{eq:92}
		\end{align}
		Equation (\ref{eq:92}) and Lemma \ref{lem: consistency of the estimated factors} imply that 
		\begin{align}
			\norm{\frac{1}{T^2} \sum_{t=1}^{T} \sum_{s=1}^{T}  \left(\widetilde{u}_{vz,s} -H' {u}_{vz,s} \right) \zeta_{ee, st}e_{it} }_F \leq &  \left(\frac{1}{T} \sum_{s=1}^{T}  \norm{\widetilde{u}_{vz,s} -H' {u}_{vz,s} }_F^2 \right)^{\frac{1}{2}} \left( \frac{1}{T}    \sum_{s=1}^{T}  \left(    \frac{1}{T}\sum_{t=1}^{T}  \zeta_{ee, st}e_{it}\right)^2\right)^{\frac{1}{2}} \nonumber \\
			\leq & O_p(\frac{1}{\sqrt{Nm_{NT}}})
		\end{align}
		Assumption \ref{assum: (New, products of factor and e)} implies that 
		\begin{align}
			\frac{1}{T^2} \sum_{t=1}^{T} \sum_{s=1}^{T}   {u}_{vz,s}  \zeta_{ee, st}e_{it} = & \frac{1}{\sqrt{NT}} \left(\frac{1}{T} \sum_{t=1}^{T} \left( \frac{1}{\sqrt{NT}}\sum_{s=1}^T \sum_{i=1}^N{u}_{vz,s}\left(e_{is}e_{it}-\mathbb{E}\left(e_{is}e_{it}\right)  \right)  \right)e_{it}\right) \nonumber \\
			\leq & O_p(\frac{1}{\sqrt{NT}})
		\end{align}
		and thus
		\begin{align}
			a_2 = \frac{1}{T^2} \sum_{t=1}^{T} \sum_{s=1}^{T}  \left(\widetilde{u}_{vz,s} -H' {u}_{vz,s} \right) \zeta_{ee, st}e_{it} + \frac{1}{T^2} \sum_{t=1}^{T} \sum_{s=1}^{T}   H' {u}_{vz,s}  \zeta_{ee, st}e_{it} \leq O_p(\frac{1}{\sqrt{Nm_{NT}}})
		\end{align}
		Assumption \ref{assum: (New, products of factor and e)}  implies that
		\begin{align}
			\frac{1}{T}\sum_{t=1}^{T}  \zeta_{ue, st}e_{it} = & \frac{1}{\sqrt{N}} u_{vz,s}'   \left(\frac{1}{T} \sum_{t=1}^{T}  \left(\frac{1}{\sqrt{N}} \sum_{i=1} c_{\beta\gamma,i} e_{it}  \right)   e_{it}\right)  = O_p(\frac{1}{\sqrt{N}})
		\end{align}
		and from Theorem \ref{theo: (factor)}, we know
		\begin{align}
			\norm{	\frac{1}{T^2} \sum_{t=1}^{T} \sum_{s=1}^{T}  \left(\widetilde{u}_{vz,s} -H' {u}_{vz,s} \right) \zeta_{ue, st}e_{it}}_F \leq & \left(\frac{1}{T} \sum_{s=1}^{T}  \norm{\widetilde{u}_{vz,s} -H' {u}_{vz,s} }_F^2 \right)^{\frac{1}{2}} \left( \frac{1}{T}    \sum_{s=1}^{T}  \left(    \frac{1}{T}\sum_{t=1}^{T}  \zeta_{ue, st}e_{it}\right)^2\right)^{\frac{1}{2}} \nonumber \\
			\leq & O_p(\frac{1}{\sqrt{Nm_{NT}}})
		\end{align}
		
		Assumptions \ref{assum:(New, additional ,  Time and Cross-Section Dependence and Heteroskedasticity) } and \ref{ass: factor loadings additional requirement } imply that 
		\begin{align}
			\frac{1}{NT}\sum_{t=1}^{T} \sum_{j=1}^{N}   \norm{ c_{\beta\gamma,i}  \tau_t(i,j)}_F  \leq  & L 	\frac{1}{N} \sum_{j=1}^{N}   \norm{   \tau(i,j)}_F = O_p(\frac{1}{N}) \label{eq:90}
		\end{align}
		Assumption \ref{assum: (New, products of factor and e)} implies that $\frac{1}{NT}\sum_{t=1}^{T} \sum_{j=1}^{N}    c_{\beta\gamma,i}  \left(e_{jt}e_{it}- \tau_t(i,j)\right) = O_p(\frac{1}{\sqrt{NT}})$ and thus from equation  (\ref{eq:90}) we know
		\begin{align}
			&\frac{1}{NT}\sum_{t=1}^{T} \sum_{j=1}^{N}    c_{\beta\gamma,i}  e_{jt}e_{it}  =   	\frac{1}{NT}\sum_{t=1}^{T} \sum_{j=1}^{N}    c_{\beta\gamma,i}  \left(e_{jt}e_{it}- \tau_t(i,j) + \tau_t(i,j) \right) \nonumber \\ 
			\leq & O_p(\frac{1}{\sqrt{NT}}) + O_p(\frac{1}{N})
		\end{align}
		Therefore, 
		\begin{align}
			a_3 = \frac{1}{T^2} \sum_{t=1}^{T} \sum_{s=1}^{T}  \left(\widetilde{u}_{vz,s} -H' {u}_{vz,s} \right) \zeta_{ue, st}e_{it} + \frac{1}{T^2} \sum_{t=1}^{T} \sum_{s=1}^{T}   H' {u}_{vz,s}  \zeta_{ue, st}e_{it} \leq O_p(\frac{1}{{m_{NT}}})
		\end{align}

		Assumption \ref{assum: (New, products of factor and e)}  implies that
		\begin{align}
			\frac{1}{T}\sum_{t=1}^{T}  \zeta_{eu, st}e_{it} = &  \frac{1}{\sqrt{N}}\left( \frac{1}{T} \sum_{t=1}^{T} 	e_{it} u_{vz, t}' \left(\frac{1}{\sqrt{N}}\sum_{j=1}^{N} c_{\beta\gamma,j} e_{js} \right)\right)   = O_p(\frac{1}{\sqrt{N}})
		\end{align}
		and thus from Theorem \ref{theo: (factor)} we know
		\begin{align}
			\norm{	\frac{1}{T^2} \sum_{t=1}^{T} \sum_{s=1}^{T}  \left(\widetilde{u}_{vz,s} -H' {u}_{vz,s} \right) \zeta_{eu, st}e_{it}}_F \leq & \left(\frac{1}{T} \sum_{s=1}^{T}  \norm{\widetilde{u}_{vz,s} -H' {u}_{vz,s} }_F^2 \right)^{\frac{1}{2}} \left( \frac{1}{T}    \sum_{s=1}^{T}  \left(    \frac{1}{T}\sum_{t=1}^{T}  \zeta_{ue, st}e_{it}\right)^2\right)^{\frac{1}{2}} \nonumber \\
			\leq & O_p(\frac{1}{\sqrt{Nm_{NT}}})
		\end{align}
		Assumption \ref{assum: (New, products of factor and e)}  implies that
		\begin{align}
			\frac{1}{T^2} \sum_{t=1}^{T} \sum_{s=1}^{T}   {u}_{vz,s}  \zeta_{eu, st}e_{it} = & \frac{1}{\sqrt{NT}}\left(\frac{1}{\sqrt{NT}}\sum_{s=1}^{T}    {u}_{vz,s}    e_s' c_{\beta\gamma} \right) \left(\frac{1}{T} \sum_{t=1}^{T} u_{vz,t} e_{it} \right)= O_p(\frac{1}{\sqrt{NT}})
		\end{align}
		Therefore,
		\begin{align}
			a_4=    \frac{1}{T^2} \sum_{t=1}^{T} \sum_{s=1}^{T}  \left(\widetilde{u}_{vz,s} -H' {u}_{vz,s} \right) \zeta_{eu, st}e_{it} + \frac{1}{T^2} \sum_{t=1}^{T} \sum_{s=1}^{T}   H' {u}_{vz,s}  \zeta_{eu, st}e_{it} \leq O_p(\frac{1}{{m_{NT}}})
		\end{align}

		Similar to the derivation of equation (\ref{eq:81}), 
		\begin{align}
			\norm{\frac{1}{NT^2}\sum_{t=1}^{T} \widetilde{u}_{vz}'P_{\bar{G}}u_g u_{g,t}e_{it}  }_F \leq  \norm{\frac{1}{NT}\sum_{t=1}^{T} \widetilde{u}_{vz}'P_{\bar{G}}u_g }_F \frac{\sqrt{N}}{\sqrt{T}} \norm{\frac{1}{\sqrt{NT}}\sum_{t=1}^{T}  u_{g,t}e_{it}  }_F = O_p(\frac{1}{T})
		\end{align}
		and thus 
		\begin{align}
			a_5 = -\frac{1}{NT^2}\sum_{t=1}^{T} \widetilde{u}_{vz}'P_{\bar{G}}u_g u_{g,t}e_{it}   \leq  O_p(\frac{1}{T}) \label{eq: a5 1}
		\end{align}
		
		Equation (\ref{eq:85}) implies that 
		\begin{align}
			a_6= \frac{1}{NT^2}\sum_{t=1}^{T}\widetilde{u}_{vz}'M_{\bar{G}} u_g \left( \widehat{u}_{g,t} - u_{g,t} \right)e_{it} \leq O_p(\frac{1}{m_{NT}}) \label{eq:a6 1}
		\end{align}

		(2) Now we analyze the term   $\frac{1}{T} \widetilde{u}_{vz}'  \left( u_{vz} - \widetilde{u}_{vz}H^{-1} \right) c_{\beta\gamma,i}$. Similar to the analysis of the term  $\frac{1}{T}   \left( \widetilde{u}_{vz} -  {u}_{vz}H \right)'  e_{i}$, equation (\ref{eq:decompose factor difference}) implies that
		\begin{align}
			&\frac{1}{T} \sum_{t=1}^{T}  \left( \widetilde{u}_{vz} -  {u}_{vz}H \right)' u_{vz}  \nonumber \\=&   V_{NT}^{-1}\left( \frac{1}{T^2} \sum_{t=1}^{T} \sum_{s=1}^{T}  \widetilde{u}_{vz,s}{u}_{vz,t}' \gamma_N(s,t) + \frac{1}{T^2} \sum_{t=1}^{T}  \sum_{s=1}^{T}  \widetilde{u}_{vz,s} {u}_{vz,t}' \zeta_{ee, st}  + \frac{1}{T^2} \sum_{t=1}^{T} \sum_{s=1}^{T}  \widetilde{u}_{vz,s} {u}_{vz,t}' \zeta_{ue, st}        \right)   \nonumber \\
			&+V_{NT}^{-1}   \left(\frac{1}{T^2} \sum_{t=1}^{T} \sum_{s=1}^{T}  \widetilde{u}_{vz,s}{u}_{vz,t}'  \zeta_{eu, st}     -\frac{1}{NT^2}\sum_{t=1}^{T} \widetilde{u}_{vz}'P_{\bar{G}}u_g u_{g,t}{u}_{vz,t}'   +  \frac{1}{NT^2}\sum_{t=1}^{T}\widetilde{u}_{vz}'M_{\bar{G}} u_g \left( \widehat{u}_{g,t} - u_{g,t} \right){u}_{vz,t}'    \right) \nonumber \\
			& = V_{NT}^{-1}\sum_{j=1}^6 a_j
		\end{align} 
		Following the same proof as the one for $a_1$ in the term $\frac{1}{T}   \left( \widetilde{u}_{vz} -  {u}_{vz}H \right)'  e_{i}$, Assumption  \ref{ass:residual innovation factor restriction} and Lemma 1(i) in  \cite{bai2002determining} imply that
		\begin{align}
			&	\norm{\frac{1}{T^2} \sum_{t=1}^{T} \sum_{s=1}^{T}  \left(\widetilde{u}_{vz,s} -H' {u}_{vz,s} \right) \gamma_N(s,t)u_{vz,t}' }_F \nonumber \\
			\leq &\frac{1}{\sqrt{T}} \left(\frac{1}{T} \sum_{s=1}^{T}  \norm{\widetilde{u}_{vz,s} -H' {u}_{vz,s} }_F^2 \right)^{\frac{1}{2}} \left( \frac{1}{T}   \sum_{t=1}^{T} \sum_{s=1}^{T}   \gamma_N(s,t)^2  \frac{1}{T}\sum_{t=1}^{T}u_{vz,t}'u_{vz,t}\right)^{\frac{1}{2}} \leq O_p(\frac{1}{\sqrt{Tm_{NT}}})
		\end{align}
		and 
		Assumption \ref{assum:factor norm bounded} and Lemma 1(i) in  \cite{bai2002determining} imply that 
		\begin{align}
			\mathbb{E}\norm{\frac{1}{T^2} \sum_{t=1}^{T} \sum_{s=1}^{T}    {u}_{vz,s}   \gamma_N(s,t)u_{vz,t}}_F \leq    \frac{1}{T^2} \sum_{t=1}^{T} \sum_{s=1}^{T}    \left(\mathbb{E}\norm{{u}_{vz,s} }^2_F\right)^{\frac{1}{2}}  \gamma_N(s,t)^2(\mathbb{E} \norm{u_{vz,t}}_F^2   )^{\frac{1}{2}} =O(\frac{1}{T})
		\end{align}
		and thus 
		\begin{align}
			a_1 \leq O_p(\frac{1}{\sqrt{Tm_{NT}}})
		\end{align}

		Assumption \ref{assum: (New, products of factor and e)} and Lemma \ref{lem: consistency of the estimated factors} imply that 
		\begin{align}
			\norm{\frac{1}{T^2} \sum_{t=1}^{T} \sum_{s=1}^{T}  \left(\widetilde{u}_{vz,s} -H' {u}_{vz,s} \right) \zeta_{ee, st}{u}_{vz,t}'}_F \leq &  \left(\frac{1}{T} \sum_{s=1}^{T}  \norm{\widetilde{u}_{vz,s} -H' {u}_{vz,s} }_F^2 \right)^{\frac{1}{2}} \left( \frac{1}{T}    \sum_{s=1}^{T}    \norm{   \frac{1}{T}\sum_{t=1}^{T}  \zeta_{ee, st}{u}_{vz,t}'}_F^2\right)^{\frac{1}{2}} \nonumber \\
			\leq & O_p(\frac{1}{\sqrt{NTm_{NT}}})
		\end{align}
		Assumption \ref{assum: (New, products of factor and e)} implies that
		\begin{align}
			\frac{1}{T^2} \sum_{s=1}^{T}\sum_{t=1}^{T}   {u}_{vz,s} {u}_{vz,t}' \zeta_{ee, st} = \frac{1}{\sqrt{NT}}   \left(\frac{1}{T} \sum_{s=1}^{T} {u}_{vz,s} \left(  \frac{1}{\sqrt{NT}}  \sum_{t=1}^{T}  \sum_{i=1}^{N}   \left( e_{it}e_{is} - \mathbb{E} e_{it}e_{is}  \right) {u}_{vz,t}'
			\right) \right) = O_p(\frac{1}{\sqrt{NT}})
		\end{align}
		and thus
		\begin{align}
			a_2 = \frac{1}{T^2} \sum_{t=1}^{T}  \sum_{s=1}^{T} \left( \widetilde{u}_{vz,s}- H'   {u}_{vz,s}\right) {u}_{vz,t}' \zeta_{ee, st} + \frac{1}{T^2} \sum_{t=1}^{T}  \sum_{s=1}^{T}  H' {u}_{vz,s} {u}_{vz,t}' \zeta_{ee, st} \leq  O_p(\frac{1}{m_{NT}})
		\end{align}

		Assumption \ref{assum: (New, products of factor and e)}  implies that
		\begin{align}
			\frac{1}{T}\sum_{t=1}^{T}  \zeta_{ue, st}u_{vz, t}'  = & \frac{1}{\sqrt{NT}} u_{vz,s}'   \left(\frac{1}{\sqrt{NT}}   \sum_{t=1}^{T}  \sum_{i=1}^N c_{\beta\gamma,i} e_{it} u_{vz, t}'  \right)    = O_p(\frac{1}{\sqrt{NT}})
		\end{align}
		and thus 
		\begin{align}
			\frac{1}{T^2} \sum_{t=1}^{T} \sum_{s=1}^{T}   {u}_{vz,s}  \zeta_{ue, st}u_{vz,t}'  =  \frac{1}{\sqrt{NT}} \frac{1}{T} \sum_{t=1}^{T} \sum_{s=1}^{T}   {u}_{vz,s}u_{vz,s}'   \left(\frac{1}{\sqrt{NT}}   \sum_{t=1}^{T}  \sum_{i=1}^N c_{\beta\gamma,i} e_{it} u_{vz, t}'  \right)    = O_p(\frac{1}{\sqrt{NT}})
		\end{align}
		From Lemma \ref{lem: consistency of the estimated factors} we know
		\begin{align}
			\norm{	\frac{1}{T^2} \sum_{t=1}^{T} \sum_{s=1}^{T}  \left(\widetilde{u}_{vz,s} -H' {u}_{vz,s} \right) \zeta_{ue, st}u_{vz,t}}_F \leq & \left(\frac{1}{T} \sum_{s=1}^{T}  \norm{\widetilde{u}_{vz,s} -H' {u}_{vz,s} }_F^2 \right)^{\frac{1}{2}} \left( \frac{1}{T}    \sum_{s=1}^{T}  \left(    \frac{1}{T}\sum_{t=1}^{T}  \zeta_{ue, st}u_{vz,t}\right)^2\right)^{\frac{1}{2}} \nonumber \\
			\leq & O_p(\frac{1}{\sqrt{NTm_{NT}}})
		\end{align}
		Therefore, 
		\begin{align}
			a_3 = \frac{1}{T^2} \sum_{t=1}^{T} \sum_{s=1}^{T}  \left(\widetilde{u}_{vz,s} -H' {u}_{vz,s} \right)  \zeta_{ue, st} u_{vz,t}' + \frac{1}{T^2} \sum_{t=1}^{T} \sum_{s=1}^{T}   H' {u}_{vz,s}  \zeta_{ue, st}u_{vz,t}' \leq O_p(\frac{1}{{\sqrt{NT}}})
		\end{align}

		Assumption \ref{assum: (New, products of factor and e)}  implies that
		\begin{align}
			\frac{1}{T^2} \sum_{t=1}^{T} \sum_{s=1}^{T}   {u}_{vz,s}  \zeta_{eu, st}u_{vz,t}'  =  \frac{1}{\sqrt{NT}}\left(\frac{1}{\sqrt{NT}}   \sum_{s=1}^{T}  \sum_{i=1}^N {u}_{vz,s}c_{\beta\gamma,i}'  e_{is}    \right) \left( \frac{1}{T} \sum_{t=1}^{T}   {u}_{vz,t}u_{vz,t}'  \right)    = O_p(\frac{1}{\sqrt{NT}})
		\end{align}
		
		and 
		\begin{align}
			\frac{1}{T}\sum_{t=1}^{T}  \zeta_{eu, st}u_{vz,t} = 	   \frac{1}{\sqrt{N}}\left(\frac{1}{\sqrt{N}}    \sum_{i=1}^N  c_{\beta\gamma,i}'  e_{is}    \right) \left( \frac{1}{T} \sum_{t=1}^{T}   {u}_{vz,t}u_{vz,t}'  \right)    = O_p(\frac{1}{\sqrt{N}})
		\end{align}
		Therefore, from Lemma \ref{lem: consistency of the estimated factors} we know
		\begin{align}
			\norm{	\frac{1}{T^2} \sum_{t=1}^{T} \sum_{s=1}^{T}  \left(\widetilde{u}_{vz,s} -H' {u}_{vz,s} \right) \zeta_{ue, st}u_{vz,t}}_F \leq & \left(\frac{1}{T} \sum_{s=1}^{T}  \norm{\widetilde{u}_{vz,s} -H' {u}_{vz,s} }_F^2 \right)^{\frac{1}{2}} \left( \frac{1}{T}    \sum_{s=1}^{T}  \left(    \frac{1}{T}\sum_{t=1}^{T}  \zeta_{ue, st}u_{vz,t}\right)^2\right)^{\frac{1}{2}} \nonumber \\
			\leq & O_p(\frac{1}{\sqrt{Nm_{NT}}})
		\end{align}
		\begin{align}
			a_4=  \frac{1}{T^2} \sum_{t=1}^{T} \sum_{s=1}^{T}  \left(\widetilde{u}_{vz,s} -H' {u}_{vz,s} \right)  \zeta_{eu, st} u_{vz,t}' + \frac{1}{T^2} \sum_{t=1}^{T} \sum_{s=1}^{T}   H' {u}_{vz,s}  \zeta_{eu, st}u_{vz,t}' \leq O_p(\frac{1}{{\sqrt{Nm_{NT}}}})
		\end{align}

		Similar to the derivation of equation (\ref{eq: a5 1}),
		\begin{align}
			a_5= -\frac{1}{NT^2}\sum_{t=1}^{T} \widetilde{u}_{vz}'P_{\bar{G}}u_g u_{g,t}{u}_{vz,t}'     \leq O_p(\frac{1}{T})
		\end{align}
		and similar to the derivation of equation (\ref{eq:a6 1}),
		\begin{align}
			a_6=	\frac{1}{NT^2}\sum_{t=1}^{T}\widetilde{u}_{vz}'M_{\bar{G}} u_g \left( \widehat{u}_{g,t} - u_{g,t} \right){u}_{vz,t}' \leq O_p(\frac{1}{m_{NT}})
		\end{align}
		Therefore,
		\begin{align}
			\frac{1}{T}\left(  \widetilde{u}_{vz}- u_{vz}H  \right)'  \widetilde{u}_{vz} =	\frac{1}{T}\left(  \widetilde{u}_{vz}- u_{vz}H  \right)'  u_{vz}H    + 	\frac{1}{T}\left(  \widetilde{u}_{vz}- u_{vz}H  \right)' \left(  \widetilde{u}_{vz}- u_{vz}H  \right)
			\leq O_p(\frac{1}{m_{NT}}) 
		\end{align}

		(3) We analyze the term   $\frac{1}{T}\widetilde{u}_{vz}'\left(\widehat{u}_{g,i}-u_{g,i} \right)$. Firstly,
		\begin{align}
			\norm{\frac{1}{T}\left(\widetilde{u}_{vz}  -{u}_{vz} H \right)' P_{\bar{G}} u_{g,i} }_F \leq  \left(\frac{1}{T} \sum_{s=1}^{T}  \norm{\widetilde{u}_{vz,s} -H' {u}_{vz,s} }_F^2 \right)^{\frac{1}{2}}  \frac{1}{\sqrt{T}}   \norm{P_{\bar{G}} u_{g,i}}_F  \leq O_p(\frac{1}{\sqrt{Tm_{NT}}})
		\end{align}
		Secondly, similar to the derivation of equation (\ref{eq:81}) we know
		\begin{align}
			\norm{\frac{1}{T}u_{vz}' P_{\bar{G}} u_{g,i} }_F \leq  O_p(\frac{1}{T}),
		\end{align}
		and thus
		\begin{align}
			\frac{1}{T}\widetilde{u}_{vz}'\left(\widehat{u}_{g,i}-u_{g,i} \right) = & \frac{1}{T}\left(\widetilde{u}_{vz}  -{u}_{vz} H \right)' P_{\bar{G}} u_{g,i}   + \frac{1}{T} {u}_{vz}'P_{\bar{G}} u_{g,i}
		\end{align}

	\end{proof}

\end{theo}

\begin{theo} \label{t}
	Suppose Assumptions \ref{assum: factor structure in r} - \ref{assumption: weak cross-sectional correlation of the idiosyncratic component }, \ref{ass: factor loadings additional requirement }  - \ref{assum: (New, products of factor and e)} hold, let N,T increase then 
	\begin{align*}
		{m}^{\frac{1}{2}}_{NT}  \left(\widetilde{c}_{it} - c_{it}\right)   =&   \frac{ {m}^{\frac{1}{2}}_{NT}}{\sqrt{N}} {c}_{\beta\gamma,i}'	  \left( c_{\beta\gamma}'c_{\beta\gamma}/N \right)^{-1}   \frac{1}{\sqrt{N}} \sum_{j=1}^N c_{\beta\gamma,j} e_{jt} \nonumber\\ &  +\frac{ {m}^{\frac{1}{2}}_{NT}}{\sqrt{T}}{u}_{vz,t}' \left(u_{vz}'u_{vz}/T \right)^{-1}  \frac{1}{\sqrt{T} }u_{vz}' e_i+ O_p(\frac{1}{\sqrt{m}_{NT}})  
	\end{align*}
\end{theo}

\noindent \begin{proof}[Proof of Theorem \ref{t}]  	
	\begin{align}
		\widetilde{c}_{it} - c_{it}  = &\widetilde{u}_{vz,t}' \widetilde{c}_{\beta\gamma,i} - u_{vz,t}' c_{\beta\gamma,i}=  \left( \widetilde{u}_{vz,t} -H' {u}_{vz,t} \right)' H^{-1}{c}_{\beta\gamma,i} + {u}_{vz,t}' H \left( \widetilde{c}_{\beta\gamma,i} -H^{-1}{c}_{\beta\gamma,i} \right) \nonumber \\
		&+ \left( \widetilde{u}_{vz,t} -H' {u}_{vz,t} \right)' \left( \widetilde{c}_{\beta\gamma,i} -H^{-1}{c}_{\beta\gamma,i} \right) \nonumber \\
		=&  \left( \widetilde{u}_{vz,t} -H' {u}_{vz,t} \right)' H^{-1}{c}_{\beta\gamma,i} + {u}_{vz,t}' H \left( \widetilde{c}_{\beta\gamma,i} -H^{-1}{c}_{\beta\gamma,i} \right) + O_p(\frac{1}{m_{NT}}) \label{eq:cc diff 1}
	\end{align}
	
	Theorem \ref{theo: (factor)}  implies that 
	\begin{align}
		{m}^{\frac{1}{2}}_{NT}  {c}_{\beta\gamma,i}'	H^{'-1}\left( \widetilde{u}_{vz,t} -H' {u}_{vz,t} \right)= & \frac{ {m}^{\frac{1}{2}}_{NT}}{\sqrt{N}} {c}_{\beta\gamma,i}'	H^{'-1} V_{NT}^{-1}\frac{1}{T} \sum_{s=1}^{T} \left( \widetilde{u}_{vz,s}{u}_{vz,s}' \right)   \frac{1}{\sqrt{N}} \sum_{i=1}^N c_{\beta\gamma,i} e_{it} + o_p(1) \nonumber\\
		=&  \frac{ {m}^{\frac{1}{2}}_{NT}}{\sqrt{N}} {c}_{\beta\gamma,i}'	  \left( c_{\beta\gamma}'c_{\beta\gamma}/N \right)^{-1}   \frac{1}{\sqrt{N}} \sum_{i=1}^N c_{\beta\gamma,i} e_{it} + o_p(1)
	\end{align}
	and Theorem \ref{theo: (factor loading)} implies that 
	\begin{align}
		{m}^{\frac{1}{2}}_{NT} {u}_{vz,t}' H  \left( \widetilde{c}_{\beta\gamma,i} -H^{-1}{c}_{\beta\gamma,i} \right)  = & \frac{ {m}^{\frac{1}{2}}_{NT}}{\sqrt{T}}{u}_{vz,t}' H H' \frac{1}{\sqrt{T} }u_{vz}' e_i + o_p(1) \nonumber\\
		= & \frac{ {m}^{\frac{1}{2}}_{NT}}{\sqrt{T}}{u}_{vz,t}' \left(u_{vz}'u_{vz}/T \right)^{-1}  \frac{1}{\sqrt{T} }u_{vz}' e_i + o_p(1)  
	\end{align}
	where the last equality is due to
	\begin{align}
		HH' = \left(u_{vz}'u_{vz}/T \right)^{-1} + O_p(\frac{1}{m_{NT}})
	\end{align}
	Finally, 
	\begin{align}
		{m}^{\frac{1}{2}}_{NT}  \left(\widetilde{c}_{it} - c_{it}\right)  =&   \frac{ {m}^{\frac{1}{2}}_{NT}}{\sqrt{N}} {c}_{\beta\gamma,i}'	  \left( c_{\beta\gamma}'c_{\beta\gamma}/N \right)^{-1}   \frac{1}{\sqrt{N}} \sum_{i=1}^N c_{\beta\gamma,i} e_{it} \nonumber\\ &  +\frac{ {m}^{\frac{1}{2}}_{NT}}{\sqrt{T}}{u}_{vz,t}' \left(u_{vz}'u_{vz}/T \right)^{-1}  \frac{1}{\sqrt{T} }u_{vz}' e_i+ O_p(\frac{1}{m_{NT}}) 
	\end{align}

\end{proof}

\subsection*{\footnotesize Step (3): $\theta$ can be consistently estimated}	 
\noindent \begin{proof}[Proof of Theorem \ref{theo:consistency of theta tilde}]  \label{Proof of Theorem theo:consistency of theta tilde}
	Rewrite moment conditions as
	\begin{align}
		\iota_N   = &  \widetilde{q}_{G,T}	\widetilde{\theta }  	+ \widetilde{\epsilon},  \label{eq:138}
	\end{align}
	with 
	$\widetilde{\theta } =\theta_G +\left( 0,\left(\left(-\widehat{V}_g^{-1} \left( \bar{g} - \mu_g  \right),  \left(  \widehat{V}_g^{-1}V_g- I_K  \right) \right)\theta_G\right) ' \right)'$.  
	
	We consider the following difference:
	\begin{align}
		\widetilde{\theta}_{g,I}^{(1)} - \widetilde{\theta}_g^{(1)}  = \left(\widetilde{q}^{(1)'}_{G,T}  P_{\widetilde{q}^{(2)}_{G,T} } \widetilde{q}^{(1)}_{G,T}  \right)^{-1} \widetilde{q}^{(1)}_{G,T} P_{ \widetilde{q}^{(2)}_{G,T}   } \widetilde{\epsilon}^{(1)}
	\end{align}
	where
	\begin{align}
		\widetilde{q}^{(1)}_{G,T}  = &  \left( c, \beta_g   \right)\widehat{Q}^{(1)}_G  + \left(\widetilde{cc}^{(1)}-cc^{(1)} \right)'\bar{G}^{(1)}/|\mathcal{T}_{(1)}| + e^{(1)'}\bar{G}^{(1)}/|\mathcal{T}_{(1)}| \nonumber \\
		=&\widetilde{X}_{g,c\beta_g}^{(1)} + \widetilde{X}_{g,cc}^{(1)}+ \widetilde{X}_{g,e}^{(1)}  \label{eq} 
	\end{align}
	For the three terms in equation (\ref{eq}), we have	
	\begin{align}
		& 	\widetilde{x}_{g,c\beta_g,i}^{(1)} =  \widehat{Q}^{(1)}_G \left( c_i, \beta_{g,i}'   \right)' \\
		&\widetilde{x}_{g,cc,i}^{(1)} =  \frac{1}{|\mathcal{T}_{(1)}|} \sum_{t\in \mathcal{T}_{(1)}} \left(\widetilde{cc}_{it} - cc_{it}\right) \bar{G}_t   \nonumber \\=&   \frac{1}{|\mathcal{T}_{(1)}|} \sum_{t\in \mathcal{T}_{(1)}} \frac{1}{\sqrt{N}} {c}_{\beta\gamma,i}'	  \left( c_{\beta\gamma}'c_{\beta\gamma}/N \right)^{-1}   \frac{1}{\sqrt{N}} \sum_{j=1}^N c_{\beta\gamma,j} e_{jt}\bar{G}_t  \nonumber\\ &  + \frac{1}{|\mathcal{T}_{(1)}|} \sum_{t\in \mathcal{T}_{(1)}}\frac{ 1}{\sqrt{T}}{u}_{vz,t}' \left(u_{vz}'u_{vz}/T \right)^{-1}  \frac{1}{\sqrt{T} }u_{vz}' e_i \bar{G}_t + O_p(\frac{1}{ {m}_{NT}}) \nonumber \\
		=&    \frac{1}{\sqrt{N|\mathcal{T}_{(1)}|}} {c}_{\beta\gamma,i}'	  \left( c_{\beta\gamma}'c_{\beta\gamma}/N \right)^{-1}   \frac{1}{\sqrt{N|\mathcal{T}_{(1)}|}} \sum_{t\in \mathcal{T}_{(1)}} \sum_{i=1}^N c_{\beta\gamma,i} e_{it}\bar{G}_t  \nonumber\\ &  + \frac{ 1}{\sqrt{T|\mathcal{T}_{(1)}|}}\frac{1}{\sqrt{|\mathcal{T}_{(1)}|}} \sum_{t\in \mathcal{T}_{(1)}} \bar{G}_t {u}_{vz,t}'   \frac{1}{\sqrt{T} } \sum_{s=1}^{T} \left(u_{vz}'u_{vz}/T \right)^{-1} u_{vz,s} e_{is}  + O_p(\frac{1}{ {m}_{NT}}) \nonumber \\ \leq & O_p(\frac{1}{\sqrt{m_{NT}|\mathcal{T}_{(1)}|}})\\
		& 	\widetilde{x}_{g,e,i}^{(1)} =    \frac{1}{\sqrt{|\mathcal{T}_{(1)}|}}\frac{1}{\sqrt{|\mathcal{T}_{(1)}|}} \sum_{t\in \mathcal{T}_{(1)}} e_{it} \bar{G}_t'  
	\end{align}

	In the following, we suppose $|\mathcal{T}_{(1)}|=|\mathcal{T}_{(2)}|=\tau$ without losing generality. Next we discuss the properties of the following three terms: 
	(1)~ $Q_{B_g,T}	\widetilde{q}^{(2)'}_{G,T} \widetilde{q}^{(2)}_{G,T}  Q_{B_g,T} $; ~
	(2)~$ \widetilde{q}^{(1)'}_{G,T} \widetilde{q}^{(2)}_{G,T} $; ~
	(3)~$ \widetilde{q}^{(2)'}_{G,T} \widetilde{\epsilon}^{(1)}$. \\

	\paragraph{\scriptsize Term (1)}   
	\begin{align}
		& \frac{\tau}{N} \sum_{i=1}^{N}	\widetilde{x}_{g,cc,i}^{(2)} \widetilde{x}_{g,cc, i}^{(2)'}  \leq O_p(\frac{1}{\sqrt{\tau}})  
	\end{align}

	\begin{align}
		\frac{1}{N}	Q_{B_g,T}\widetilde{X}_{g,c\beta_g}^{(2); } \widetilde{X}_{g,c\beta_g}^{(2)}Q_{B_g,T}\rightarrow_p &   {Q}_g \eta_{c\beta_g} Q_g   
	\end{align}
	
	Equation (\ref{eq{assum: GG gamma}}) implies that
	\begin{align}
		\frac{1}{N|\mathcal{T}_{(1)}|} \sum_{i=1}^{N}  \sum_{t\in \mathcal{T}_{(1)}} e_{it}^2 \bar{G}_t' \bar{G}_t  \rightarrow_p &\Sigma_{\bar{G}\bar{G}\gamma_N},
	\end{align}
	and
	\begin{align}
		\frac{1}{\sqrt{N}}	\frac{1}{\sqrt{N}|\mathcal{T}_{(1)}|} \sum_{i=1}^{N}  \sum_{t\in \mathcal{T}_{(1)}}\sum_{s\neq t} e_{it} \bar{G}_t'e_{is} \bar{G}_s =O_p(\frac{1}{\sqrt{N}})
	\end{align}
	Then we know 
	\begin{align}
		&\frac{\tau}{N}	\widetilde{X}_{g,e}^{(2)'} \widetilde{X}_{g,e}^{(2)} =   \frac{\tau}{N} \sum_{i=1}^{N} \widetilde{x}_{g,e,i}^{(2)} \widetilde{x}_{g,e,i}^{(2)'} = \frac{1}{N} \sum_{i=1}^{N} \left( \frac{1}{\sqrt{|\mathcal{T}_{(2)}|}} \sum_{t\in \mathcal{T}_{(2)}} e_t \bar{G}_t'\right)\left( \frac{1}{\sqrt{|\mathcal{T}_{(2)}|}} \sum_{t\in \mathcal{T}_{(2)}} e_t \bar{G}_t'\right)' \nonumber \\
		= & \frac{1}{N|\mathcal{T}_{(2)}|} \sum_{i=1}^{N}  \sum_{t\in \mathcal{T}_{(2)}} e_{it} \bar{G}_t' \bar{G}_t e_{it}  +\frac{1}{N|\mathcal{T}_{(2)}|} \sum_{i=1}^{N}  \sum_{t\in \mathcal{T}_{(2)}}\sum_{s\neq t} e_{it} \bar{G}_t' \bar{G}_s e_{is} \rightarrow_p \Sigma_{\bar{G}\bar{G}\gamma_N}
	\end{align}

	\begin{align}
		& \frac{\sqrt{\tau}}{N} \sum_{i=1}^{N} Q_{B_g,T} \widetilde{x}_{g,c\beta_g,i}^{(2)} \widetilde{x}_{g,cc,i}^{(2)'} \nonumber \\
		=	& \frac{1}{\sqrt{N|\mathcal{T}_{(2)}|}}\frac{\sqrt{\tau}}{N} \sum_{i=1}^{N} \left( \widehat{Q}^{(1)}_GQ_{B_g,T} \left( c_i, \beta_{g,i}'   \right)'\right) \left(   {c}_{\beta\gamma,i}'	  \left( c_{\beta\gamma}'c_{\beta\gamma}/N \right)^{-1}   \frac{1}{\sqrt{N|\mathcal{T}_{(2)}|}} \sum_{t\in \mathcal{T}_{(2)}} \sum_{i=1}^N c_{\beta\gamma,i} e_{it}\bar{G}_t    \right)'  \nonumber \\
		& +\frac{ 1}{\sqrt{T|\mathcal{T}_{(2)}|}}\frac{\sqrt{\tau}}{N} \sum_{i=1}^{N} \left( \widehat{Q}^{(2)}_g  Q_{B_g,T}\left( c_i, \beta_{g,i}'   \right)'\right)\left(  \frac{1}{\sqrt{|\mathcal{T}_{(2)}|}} \sum_{t\in \mathcal{T}_{(2)}} \bar{G}_t {u}_{vz,t}'   \frac{1}{\sqrt{T} } \sum_{s=1}^{T} \left(u_{vz}'u_{vz}/T \right)^{-1} u_{vz,s} e_{is} \right)'  \nonumber \\
		&+ O_p(\frac{1}{  {{m}_{NT}}})  \leq   O_p(\frac{1}{ \sqrt{{m}_{NT}}})  \nonumber
	\end{align}

	\begin{align}
		\frac{\sqrt{\tau}}{N}Q_{B_g,T} \sum_{i=1}^{N} \widetilde{x}_{g,c\beta_g,i}^{(2)} \widetilde{x}_{g,e,i}^{(2)'} = \frac{ {\sqrt{\tau}}}{N} \sum_{i=1}^{N} \left( \widehat{Q}^{(1)}_GQ_{B_g,T} \left( c_i, \beta_{g,i}'   \right)'\right) \left(    \frac{1}{ {|\mathcal{T}_{(2)}|}} \sum_{t\in \mathcal{T}_{(2)}}  e_{it}\bar{G}_t    \right)' = O_p(\frac{1}{\sqrt{N}})
	\end{align}

	\begin{align}
		&\frac{\tau}{N} \sum_{i=1}^{N} \widetilde{x}_{g,cc,i}^{(2)} \widetilde{x}_{g,e,i}^{(2)'} \nonumber\\
		=& \frac{\tau}{N} \sum_{i=1}^{N}  \left(  \frac{1}{\sqrt{N|\mathcal{T}_{(2)}|}} {c}_{\beta\gamma,i}'	  \left( c_{\beta\gamma}'c_{\beta\gamma}/N \right)^{-1}   \frac{1}{\sqrt{N|\mathcal{T}_{(2)}|}} \sum_{t\in \mathcal{T}_{(2)}} \sum_{j=1}^N c_{\beta\gamma,j} e_{jt}\bar{G}_t\right)\left(    \frac{1}{ {|\mathcal{T}_{(2)}|}} \sum_{t\in \mathcal{T}_{(2)}}  e_{it}\bar{G}_t    \right)'  \nonumber\\ 
		&  + \frac{\tau}{N} \sum_{i=1}^{N}    \left(\frac{ 1}{\sqrt{T|\mathcal{T}_{(2)}|}}\frac{1}{\sqrt{|\mathcal{T}_{(2)}|}} \sum_{t\in \mathcal{T}_{(2)}} \bar{G}_t {u}_{vz,t}'   \frac{1}{\sqrt{T} } \sum_{s=1}^{T} \left(u_{vz}'u_{vz}/T \right)^{-1} u_{vz,s} e_{is}  + O_p(\frac{1}{ {m}_{NT}})\right)\nonumber \\
		& \times \left(    \frac{1}{ {|\mathcal{T}_{(2)}|}} \sum_{t\in \mathcal{T}_{(2)}}  e_{it}\bar{G}_t    \right)' \nonumber \\
		=& O_p(\frac{1}{\sqrt{\tau}})
	\end{align}

	Therefore,
	\begin{align}
		&\frac{1}{N} Q_{B_g,T} \widetilde{q}^{(2)'}_{G,T} \widetilde{q}^{(2)}_{G,T}  Q_{B_g,T} =  \frac{1}{N} \sum_{i=1}^{N}Q_{B_g,T} \widetilde{x}^{(2) }_{g,i}\widetilde{x}^{(2)'}_{g,i} Q_{B_g,T} \nonumber \\
		=& \frac{1}{N} \sum_{i=1}^{N}Q_{B_g,T} \left(  \widetilde{x}_{g,c\beta_g,i}^{(2)} + \widetilde{x}_{g,cc,i}^{(2)}+ \widetilde{x}_{g,e,i}^{(2)} \right)\left(  \widetilde{x}_{g,c\beta_g,i}^{(2)} + \widetilde{x}_{g,cc,i}^{(2)}+ \widetilde{x}_{g,e,i}^{(2)} \right)' Q_{B_g,T} \nonumber \\
		=& \frac{1}{N} Q_{B_g,T} \left(\sum_{i=1}^{N} \widetilde{x}_{g,c\beta_g,i}^{(2)}\widetilde{x}_{g,c\beta_g,i}^{(2)'}\right)Q_{B_g,T}+  Q_{B_g,T}/\sqrt{\tau} \left(\frac{\tau}{N}\sum_{i=1}^{N} \widetilde{x}_{g,cc,i}^{(2)}\widetilde{x}_{g,cc,i}^{(2)'}\right)Q_{B_g,T}/\sqrt{\tau} \nonumber\\
		& +  Q_{B_g,T}/\sqrt{\tau} \left(\frac{\tau}{N}\sum_{i=1}^{N} \widetilde{x}_{g,e,i}^{(2)}\widetilde{x}_{g,e,i}^{(2)'}\right)Q_{B_g,T}/\sqrt{\tau} + O_p(\frac{1}{\sqrt{\tau}}) \nonumber \\
		= & {Q}_g \eta_{c\beta_g} Q_g+W_x\Sigma_{\bar{G}\bar{G}\gamma_N}W_x  + O_p(\frac{1}{\sqrt{\tau}})  \label{eq,}
	\end{align}
	with ${W}_x= \lim\limits_{\tau \rightarrow \infty} Q_{B_g,T}/\sqrt{\tau}$. 
	
	\paragraph{\scriptsize Term (2)}  
	\begin{align}
		\frac{1}{N}	Q_{B_g,T}\widetilde{X}_{g,c\beta_g}^{(1); } \widetilde{X}_{g,c\beta_g}^{(2)}Q_{B_g,T}\rightarrow_p &   {Q}_g \eta_{c\beta_g} Q_g   
	\end{align}

	\begin{align}
		& \frac{\tau}{N} \sum_{i=1}^{N}	\widetilde{x}_{g,cc,i}^{(1)} \widetilde{x}_{g,cc, i}^{(2)'}  \leq O_p(\frac{1}{\sqrt{\tau}})  
	\end{align}
	
	\begin{align}
		&\frac{\tau}{N} \sum_{i=1}^{N} \widetilde{x}_{g,e,i}^{(1)} \widetilde{x}_{g,e,i}^{(2)'} = \frac{1}{N} \sum_{i=1}^{N} \left( \frac{1}{\sqrt{|\mathcal{T}_{(1)}|}} \sum_{s\in \mathcal{T}_{(1)}} e_s \bar{G}_s'\right)\left( \frac{1}{\sqrt{|\mathcal{T}_{(2)}|}} \sum_{t\in \mathcal{T}_{(2)}} e_t \bar{G}_t'\right)' \nonumber \\
		= &  \frac{1}{N \tau} \sum_{i=1}^{N} \sum_{s\in \mathcal{T}_{(1)}}  \sum_{t\in \mathcal{T}_{(2)}}  e_{it} \bar{G}_t' \bar{G}_s e_{is}  = O_p( \frac{1}{\sqrt{N}}) \label{eq11}
	\end{align}

	\begin{align}
		& \frac{\sqrt{\tau}}{N} \sum_{i=1}^{N} Q_{B_g,T} \widetilde{x}_{g,c\beta_g,i}^{(1)} \widetilde{x}_{g,cc,i}^{(2)'} \nonumber \\
		=	& \frac{1}{\sqrt{N|\mathcal{T}_{(2)}|}}\frac{\sqrt{\tau}}{N} \sum_{i=1}^{N} \left( \widehat{Q}^{(1)}_GQ_{B_g,T} \left( c_i, \beta_{g,i}'   \right)'\right) \left(   {c}_{\beta\gamma,i}'	  \left( c_{\beta\gamma}'c_{\beta\gamma}/N \right)^{-1}   \frac{1}{\sqrt{N|\mathcal{T}_{(2)}|}} \sum_{t\in \mathcal{T}_{(2)}} \sum_{i=1}^N c_{\beta\gamma,i} e_{it}\bar{G}_t    \right)'  \nonumber \\
		& +\frac{ 1}{\sqrt{T|\mathcal{T}_{(2)}|}}\frac{\sqrt{\tau}}{N} \sum_{i=1}^{N} \left( \widehat{Q}^{(1)}_G Q_{B_g,T}\left( c_i, \beta_{g,i}'   \right)'\right)\left(  \frac{1}{\sqrt{|\mathcal{T}_{(2)}|}} \sum_{t\in \mathcal{T}_{(2)}} \bar{G}_t {u}_{vz,t}'   \frac{1}{\sqrt{T} } \sum_{s=1}^{T} \left(u_{vz}'u_{vz}/T \right)^{-1} u_{vz,s} e_{is} \right)'  \nonumber \\
		&+ O_p(\frac{1}{ \sqrt{{m}_{NT}}})   =  O_p(\frac{1}{ \sqrt{{m}_{NT}}})  \nonumber
	\end{align}

	\begin{align}
		\frac{\sqrt{\tau}}{N}Q_{B_g,T} \sum_{i=1}^{N} \widetilde{x}_{g,c\beta_g,i}^{(1)} \widetilde{x}_{g,e,i}^{(2)'} = \frac{ {\sqrt{\tau}}}{N} \sum_{i=1}^{N} \left( \widehat{Q}^{(1)}_GQ_{B_g,T} \left( c_i, \beta_{g,i}'   \right)'\right) \left(    \frac{1}{ {|\mathcal{T}_{(2)}|}} \sum_{t\in \mathcal{T}_{(2)}}  e_{it}\bar{G}_t    \right)' = O_p(\frac{1}{\sqrt{N}}) \label{eq:2}
	\end{align}

	\begin{align}
		&\frac{\tau}{N} \sum_{i=1}^{N} \widetilde{x}_{g,cc,i}^{(1)} \widetilde{x}_{g,e,i}^{(2)'} \nonumber\\
		=& \frac{\tau}{N} \sum_{i=1}^{N}  \left(  \frac{1}{\sqrt{N|\mathcal{T}_{(1)}|}} {c}_{\beta\gamma,i}'	  \left( c_{\beta\gamma}'c_{\beta\gamma}/N \right)^{-1}   \frac{1}{\sqrt{N|\mathcal{T}_{(1)}|}} \sum_{t\in \mathcal{T}_{(1)}} \sum_{j=1}^N c_{\beta\gamma,j} e_{jt}\bar{G}_t\right)\left(    \frac{1}{ {|\mathcal{T}_{(2)}|}} \sum_{t\in \mathcal{T}_{(2)}}  e_{it}\bar{G}_t    \right)'  \nonumber\\ 
		&  + \frac{\tau}{N} \sum_{i=1}^{N}    \left(\frac{ 1}{\sqrt{T|\mathcal{T}_{(1)}|}}\frac{1}{\sqrt{|\mathcal{T}_{(1)}|}} \sum_{t\in \mathcal{T}_{(1)}} \bar{G}_t {u}_{vz,t}'   \frac{1}{\sqrt{T} } \sum_{s=1}^{T} \left(u_{vz}'u_{vz}/T \right)^{-1} u_{vz,s} e_{is}  + O_p(\frac{1}{ {m}_{NT}})\right)\nonumber \\
		& \times \left(    \frac{1}{ {|\mathcal{T}_{(2)}|}} \sum_{t\in \mathcal{T}_{(2)}}  e_{it}\bar{G}_t    \right)' \nonumber \\
		=& O_p(\frac{1}{\sqrt{N}}) + O_p(\frac{1}{\sqrt{\tau}}) \label{eq:3}
	\end{align}

	Therefore,
	\begin{align}
		&\frac{1}{N} Q_{B_g,T} \widetilde{q}^{(1)'}_{G,T} \widetilde{q}^{(2)}_{G,T}  Q_{B_g,T} =  \frac{1}{N} \sum_{i=1}^{N}Q_{B_g,T} \widetilde{x}^{(1) }_{g,i}\widetilde{x}^{(2)'}_{g,i} Q_{B_g,T} \nonumber \\
		=& \frac{1}{N} \sum_{i=1}^{N}Q_{B_g,T} \left(  \widetilde{x}_{g,c\beta_g,i}^{(1)} + \widetilde{x}_{g,cc,i}^{(1)}+ \widetilde{x}_{g,e,i}^{(1)} \right)\left(  \widetilde{x}_{g,c\beta_g,i}^{(2)} + \widetilde{x}_{g,cc,i}^{(2)}+ \widetilde{x}_{g,e,i}^{(2)} \right)' Q_{B_g,T} \nonumber \\
		=& \frac{1}{N} Q_{B_g,T} \left(\sum_{i=1}^{N} \widetilde{x}_{g,c\beta_g,i}^{(1)}\widetilde{x}_{g,c\beta_g,i}^{(2)'}\right)Q_{B_g,T}+    O_p(\frac{1}{\sqrt{\tau}})={Q}_g \eta_{c\beta_g} Q_g  + O_p(\frac{1}{\sqrt{\tau}}) 
	\end{align}

	\paragraph{\scriptsize Term (3)}  
	\begin{align}
		\widetilde{\epsilon}_i^{(1)} = \left(\widetilde{x}_{g,e,i}^{(1)} + \widetilde{x}_{g,cc,i}^{(1)} \right)\widetilde{\theta} \label{ha3}
	\end{align}

	Equations (\ref{eq11})-(\ref{eq:3}) imply that 
	\begin{align}
		&{\frac{\sqrt{\tau}}{\sqrt{N}}}Q_{B_g,T} \sum_{i=1}^N \widetilde{x}^{(2)}_{g,i} \left(\widetilde{x}_{g,e,i}^{(1)} + \widetilde{x}_{g,cc,i}^{(1)} \right)'     \nonumber \\
		=& 	{\frac{\sqrt{\tau}}{\sqrt{N}}}Q_{B_g,T} \sum_{i=1}^N  \left(      \widetilde{x}_{g,c\beta_g,i}^{(2)} + \widetilde{x}_{g,cc,i}^{(2)}+ \widetilde{x}_{g,e,i}^{(2)}  \right) \left(\widetilde{x}_{g,e,i}^{(1)} + \widetilde{x}_{g,cc,i}^{(1)} \right)'  = O_p(1)
	\end{align}
	which provided that $\norm{\theta_G}_F\leq L$ then gives 
	\begin{align}
		\sqrt{\frac{\tau}{N}}Q_{B_g,T} \widetilde{q}^{(2)'}_{G,T} \widetilde{\epsilon}^{(1)}  = O_p(1)
	\end{align}

	Finally, 
	\begin{align}
		&\sqrt{N\tau }Q_{B_g,T}^{-1}\left(\widetilde{\theta}_{g,I}^{(1)} - \widetilde{\theta}_g^{(1)} \right) \nonumber \\
		=& \left(\frac{1}{N} Q_{B_g,T}\widetilde{q}^{(1)'}_{G,T}  \widetilde{q}^{(2)}_{G,T} Q_{B_g,T}  \left(\frac{1}{N} Q_{B_g,T}\widetilde{q}^{(2)'}_{G,T}\widetilde{q}^{(2)}_{G,T}  Q_{B_g,T} \right)^{-1}\frac{1}{N} Q_{B_g,T}\widetilde{q}^{(2)'}_{G,T}  \widetilde{q}^{(1)}_{G,T}  Q_{B_g,T}\right)^{-1} \nonumber \\
		&\times \frac{1}{N} Q_{B_g,T} \widetilde{q}^{(1)}_{G,T} \widetilde{q}^{(2)}_{G,T} Q_{B_g,T}  \left(\frac{1}{N} Q_{B_g,T}\widetilde{q}^{(2)'}_{G,T}\widetilde{q}^{(2)}_{G,T}  Q_{B_g,T} \right)^{-1}\sqrt{\frac{\tau}{N}}Q_{B_g,T}\widetilde{q}^{(2)'}_{G,T}  \widetilde{\epsilon}^{(1)} = O_p(1),
	\end{align}
	which then leads to the final result $
	\sqrt{NT}	Q^{-1}_{B_g,T}\left(\widetilde{\theta}_{G} - \theta_G\right) \rightarrow O_p(1)$.\\
	
	Next we propose the following $\widehat{\Sigma}_{\theta_G}$ to complete the discussion concerning the footnote in Theorem \ref{theo:consistency of theta tilde}:
	\begin{align*}
		\widehat{\Sigma}_{\theta_G} = \frac{1}{2N}\sum_{i=1}^2\widehat{\Sigma}_{IV}^{(i)} + \frac{1}{T} \widehat{\Sigma}_{\widetilde{\theta}}
	\end{align*} 
	where $ \widehat{\Sigma}_{\widetilde{\theta}}$ is a consistent estimator for the variance of $\widetilde{\theta}$ and
	\begin{align*}
		\widehat{\Sigma}_{IV}^{(1)}=  \left(\widetilde{q}^{(1)'}_{G,T}  P_{\widetilde{q}^{(2)}_{G,T}}\widetilde{q}^{(1)}_{G,T} /N    \right)^{-1} 
		\left( \sum_{i=1}^{N}\left( \left(\widetilde{q}^{(1)'}_{G,T}  P_{\widetilde{q}^{(2)}_{G,T}}\right)_i \widehat{\widetilde{\epsilon}}_i^{(1)}\right)\left( \left(\widetilde{q}^{(1)'}_{G,T}  P_{\widetilde{q}^{(2)}_{G,T}}\right)_i \widehat{\widetilde{\epsilon}}_i^{(1)}\right)' \right)
		\left(\widetilde{q}^{(1)'}_{G,T}  P_{\widetilde{q}^{(2)}_{G,T}}\widetilde{q}^{(1)}_{G,T} /N    \right)^{-1} \\
		\widehat{\Sigma}_{IV}^{(2)}=  \left(\widetilde{q}^{(2)'}_{G,T}  P_{\widetilde{q}^{(1)}_{G,T}}\widetilde{q}^{(2)}_{G,T} /N    \right)^{-1} 
		\left( \sum_{i=1}^{N}\left( \left(\widetilde{q}^{(2)'}_{G,T}  P_{\widetilde{q}^{(1)}_{G,T}}\right)_i \widehat{\widetilde{\epsilon}}_i^{(2)}\right)\left( \left(\widetilde{q}^{(2)'}_{G,T}  P_{\widetilde{q}^{(1)}_{G,T}}\right)_i \widehat{\widetilde{\epsilon}}_i^{(2)}\right)' \right)
		\left(\widetilde{q}^{(2)'}_{G,T}  P_{\widetilde{q}^{(1)}_{G,T}}\widetilde{q}^{(2)}_{G,T} /N    \right)^{-1}  	
	\end{align*}
	The validity of our proposed covariance estimator relies on some additional regular assumptions: 
	\begin{assumption}\label{assum:additional}
		We assume the following holds:\\
		(1) 
		\begin{align*}
			\frac{1}{N} \sum_{i=1}^N \left(\begin{matrix}
				\sqrt{\tau}\widetilde{e}_{G,i}^{(j)}  \\ \sqrt{\tau} \widetilde{e}_{G,i}^{(j)}\widetilde{e}_{G,i}^{(j^*)}\\ \widetilde{u}_{G,i}^{(j)} \\ \widetilde{u}_{e,i}^{(j)}
			\end{matrix} \right) = 	\frac{1}{N} \sum_{i=1}^N \left(\begin{matrix}
				\xi_{1,i,T}  \\ \xi_{2,jj^*,i,T}\\ \eta_{1,i,T}\\ \eta_{2,i,T}
			\end{matrix} \right) \rightarrow_d \left(\begin{matrix}
				\xi_{1}  \\ \xi_{2,jj^*}\\ \eta_{1}\\ \eta_{2}
			\end{matrix} \right)
		\end{align*}	
		with $\widetilde{e}_{G,i}^{(j)} =  \frac{1}{ {|\mathcal{T}_{(j)}|}} \sum_{t\in \mathcal{T}_{(j)}}  e_{it}\bar{G}_t', \widetilde{u}_G^{(j)} =  \frac{1}{ {|\mathcal{T}_{(j)}|}} \sum_{t\in \mathcal{T}_{(j)}}  u_{vz,it}\bar{G}_t', \widetilde{u}_e^{(j)} =  \frac{1}{ {|\mathcal{T}_{(j)}|}} \sum_{t\in \mathcal{T}_{(j)}}  u_{vz,it}e_{t}$. \\
		(2) $	\xi_{1}, \xi_{2,jj^*}, $ are independent from $\eta_{1}, \eta_{2}$, and $\frac{1}{N}\sum_{i=1}^N \xi_{i,T}\xi_{i,T}'\rightarrow_p \Sigma_\xi$ with $\xi_i= \left(\xi_{1,i,T}',\xi_{2,jj^*,i,T}' \right)'$.
	\end{assumption} 	 
	
	Assumption \ref{assum:additional} can be relaxed if we assume for example $\sqrt{T}/N \rightarrow 0$, as under such case certain sampling errors would be negligible. Now we briefly discuss the validity of our proposed covariance estimator. From equations (\ref{eq,}) and (\ref{eq:3}), we know 
	\begin{align}
		& \left(\frac{1}{N} Q_{B_g,T}\widetilde{q}^{(1)'}_{G,T}  \widetilde{q}^{(2)}_{G,T} Q_{B_g,T}  \left(\frac{1}{N} Q_{B_g,T}\widetilde{q}^{(2)'}_{G,T}\widetilde{q}^{(2)}_{G,T}  Q_{B_g,T} \right)^{-1}\frac{1}{N} Q_{B_g,T}\widetilde{q}^{(2)'}_{G,T}  \widetilde{q}^{(1)}_{G,T}  Q_{B_g,T}\right)^{-1} \nonumber \\
		&\times \frac{1}{N} Q_{B_g,T} \widetilde{q}^{(1)}_{G,T} \widetilde{q}^{(2)}_{G,T} Q_{B_g,T}  \left(\frac{1}{N} Q_{B_g,T}\widetilde{q}^{(2)'}_{G,T}\widetilde{q}^{(2)}_{G,T}  Q_{B_g,T} \right)^{-1} \rightarrow_p \Theta_{(1)}
	\end{align}
	with $\Theta_{(1)}$ a deterministic positive definite matrix. Then we only need to look at the term $\sqrt{\frac{\tau}{N}}Q_{B_g,T} \widetilde{q}^{(2)'}_{G,T} \widetilde{\epsilon}^{(1)}$. Follow the previous discussion, denote 
	\begin{align}
		\widetilde{x}_{g,c\beta_g,i}^{(2)} =    \widehat{Q}^{(2)}_g  \left( c_i, \beta_{g,i}'   \right)',~~  \widetilde{x}_{g,e,i}^{(2)} =   \frac{1}{\sqrt{|\mathcal{T}_{(2)}|}}\frac{1}{\sqrt{|\mathcal{T}_{(2)}|}} \sum_{t\in \mathcal{T}_{(2)}} e_{it} \bar{G}_t' \label{ha1}
	\end{align}
	Equation (\ref{eq:cc diff 1}) implies that 
	\begin{align}
		\widetilde{c}_{it} - c_{it}  
		=&  \left( \widetilde{u}_{vz,t} -H' {u}_{vz,t} \right)' H^{-1}{c}_{\beta\gamma,i} + {u}_{vz,t}' H \left( \widetilde{c}_{\beta\gamma,i} -H^{-1}{c}_{\beta\gamma,i} \right)\nonumber\\
		&+\left( \frac{1}{\sqrt{N}} V_{NT}^{-1}\frac{1}{T} \sum_{s=1}^{T} \left( \widetilde{u}_{vz,s}{u}_{vz,s}' \right)   \frac{1}{\sqrt{N}} \sum_{i=1}^N c_{\beta\gamma,i} e_{it}  \right)'  \left( \frac{1}{\sqrt{T}} H' \frac{1}{\sqrt{T} }u_{vz}' e_i  \right) + o_p(\frac{1}{m_{NT}}) \nonumber \\
		= & \left( \widetilde{u}_{vz,t} -H' {u}_{vz,t} \right)' H^{-1}{c}_{\beta\gamma,i} + {u}_{vz,t}' H \left( \widetilde{c}_{\beta\gamma,i} -H^{-1}{c}_{\beta\gamma,i} \right)\nonumber\\
		&+\left(   \frac{1}{{T} }u_{vz}' e_i  \right)'\left( \frac{1}{\sqrt{N}} HH'  \left( c_{\beta\gamma}'c_{\beta\gamma}/N \right)^{-1}    \frac{1}{\sqrt{N}} \sum_{j=1}^N c_{\beta\gamma,j} e_{jt}  \right)   + o_p(\frac{1}{m_{NT}}) \nonumber \\
		=& \left( \widetilde{u}_{vz,t} -H' {u}_{vz,t} \right)' H^{-1}{c}_{\beta\gamma,i} + {u}_{vz,t}' H \left( \widetilde{c}_{\beta\gamma,i} -H^{-1}{c}_{\beta\gamma,i} \right)\nonumber\\
		&+\left(   \frac{1}{{T} }u_{vz}' e_i  \right)'\left( \frac{1}{\sqrt{N}}  \left(u_{vz}'u_{vz}/T \right)^{-1}  \left( c_{\beta\gamma}'c_{\beta\gamma}/N \right)^{-1}    \frac{1}{\sqrt{N}} \sum_{j=1}^N c_{\beta\gamma,j} e_{jt}  \right)   + o_p(\frac{1}{m_{NT}}) \nonumber \\
		= &  \frac{ 1}{\sqrt{N}} {c}_{\beta\gamma,i}'	  \left( c_{\beta\gamma}'c_{\beta\gamma}/N \right)^{-1}   \frac{1}{\sqrt{N}} \sum_{j=1}^N c_{\beta\gamma,j} e_{jt}   +\frac{1}{\sqrt{T}}{u}_{vz,t}' \left(u_{vz}'u_{vz}/T \right)^{-1}  \frac{1}{\sqrt{T} }u_{vz}' e_i \nonumber\\
		&+\left(   \frac{1}{{T} }u_{vz}' e_i  \right)'\left( \frac{1}{\sqrt{N}}  \left(u_{vz}'u_{vz}/T \right)^{-1}  \left( c_{\beta\gamma}'c_{\beta\gamma}/N \right)^{-1}    \frac{1}{\sqrt{N}} \sum_{j=1}^N c_{\beta\gamma,j} e_{jt}  \right)   + o_p(\frac{1}{m_{NT}}), 
	\end{align}	   
	and thus
	\begin{align}
		\widetilde{x}_{g,cc,i}^{(2)} = &\frac{1}{|\mathcal{T}_{(2)}|}\sum_{t\in \mathcal{T}_{(2)}}  \frac{ 1}{\sqrt{N}} {c}_{\beta\gamma,i}'	  \left( c_{\beta\gamma}'c_{\beta\gamma}/N \right)^{-1}   \frac{1}{\sqrt{N}} \sum_{j=1}^N c_{\beta\gamma,j} e_{jt}\bar{G}_t  \nonumber \\
		& + \frac{1}{|\mathcal{T}_{(2)}|}\sum_{t\in \mathcal{T}_{(2)}} \frac{1}{\sqrt{T}}{u}_{vz,t}' \left(u_{vz}'u_{vz}/T \right)^{-1}  \frac{1}{\sqrt{T} }u_{vz}' e_i\bar{G}_t \nonumber\\
		&+\frac{1}{|\mathcal{T}_{(2)}|}\sum_{t\in \mathcal{T}_{(2)}}\left(\frac{1}{{T} }u_{vz}' e_i  \right)'  \frac{1}{\sqrt{N}}  \left(u_{vz}'u_{vz}/T \right)^{-1}  \left( c_{\beta\gamma}'c_{\beta\gamma}/N \right)^{-1}    \frac{1}{\sqrt{N}} \sum_{j=1}^N c_{\beta\gamma,j} e_{jt} \bar{G}_t  + o_p(\frac{1}{m_{NT}})\nonumber \\
		= &  \frac{ 1}{\sqrt{N{ {|\mathcal{T}_{(2)}|}}}}     \left( \frac{1}{\sqrt{N|\mathcal{T}_{(2)}|}} \sum_{j=1}^N    \sum_{t\in \mathcal{T}_{(2)}} {c}_{\beta\gamma,i}'	  \left( c_{\beta\gamma}'c_{\beta\gamma}/N \right)^{-1}c_{\beta\gamma,j}e_{jt}\bar{G}_t \right) \nonumber \\
		& +\frac{1}{\sqrt{T|\mathcal{T}_{(2)}|}}  \left(\frac{1}{\sqrt{|\mathcal{T}_{(2)}|}}\sum_{t\in \mathcal{T}_{(2)}} \bar{G}_t {u}_{vz,t}'  \right) \left(u_{vz}'u_{vz}/T \right)^{-1}\left(\frac{1}{\sqrt{T} }  u_{vz}'e_i \right) \nonumber\\
		&+\frac{1}{\sqrt{|\mathcal{T}_{(2)}|}} \left(\frac{1}{{T} }u_{vz}' e_i  \right)'  \frac{1}{\sqrt{N}}  \left(u_{vz}'u_{vz}/T \right)^{-1}  \left( c_{\beta\gamma}'c_{\beta\gamma}/N \right)^{-1}    \frac{1}{\sqrt{N}} \sum_{j=1}^N c_{\beta\gamma,j} \left(\frac{1}{\sqrt{|\mathcal{T}_{(2)}|}}\sum_{t\in \mathcal{T}_{(2)}} e_{jt} \bar{G}_t\right) \nonumber \\
		& + o_p(\frac{1}{m_{NT}}) = O_p(\frac{1}{\sqrt{NT}}) + O_p(\frac{1}{T})+ O_p(\frac{1}{\sqrt{N}T})+ o_p(\frac{1}{m_{NT}}) \label{ha2}
	\end{align}
	which then implies that $	\widetilde{x}_{g,cc,i}^{(2)}$ is of higher order. From equations (\ref{ha3}), (\ref{ha1}) and (\ref{ha2}), it is clear that 	
	\begin{align*}
		\sqrt{\frac{\tau}{N}}Q_{B_g,T} \widetilde{q}^{(2)'}_{G,T} \widetilde{\epsilon}^{(1)} \rightarrow_d H(\eta)\xi 
	\end{align*}	
	with $H(.)$ a deterministic function (the limiting distribution is a mixed Gaussian distribution). The rests follow Assumption \ref{assum:additional} and the proof of Theorem 3 in \cite{anatolyev2018factor}.  	
	%
	%

	%
	%
	%
	%
	
\end{proof}

\noindent \begin{proof}[Proof of Corollary \ref{coro:novel test 1}] \label{proof coro:novel test 1}
	By the construction of the HJN statistic, $\sqrt{T} q_{g,T,R} \left(\widetilde{\theta}_G- \theta_G \right)$ would be of the order $O_p(\frac{1}{N^{1/2}})$ when $N/T\rightarrow c$, 
	and thus these sampling errors would be negligible. Therefore, the asymptotic distribution of the HJN statistic is determined by  the distribution of the sample pricing errors ${e}_{T,R}(\theta_G)$. The consistency of $\widetilde{S}= \frac{1}{T}\sum_{t=1}^T e_{g,t,R} (\widetilde{\theta}_G)e_{g,t,R} (\widetilde{\theta}_G)'$ is implied by Theorem \ref{theo:consistency of theta tilde} and  Assumptions \ref{assum: factor structure in r} - \ref{assum:factor loading strength 2}. 
\end{proof}

\end{document}